\documentclass[11pt,a4paper]{article}
\usepackage{jheppub}

\usepackage{multirow, graphicx,amssymb,url,mathrsfs,amsmath}
\usepackage{wrapfig,boxedminipage,setspace,subfigure,epsfig}
\usepackage{amsxtra,amstext,latexsym,dsfont,amsfonts}
\usepackage{color,eucal}
\usepackage[dvipsnames]{xcolor}
\usepackage{float}
\usepackage{slashed,comment}
\usepackage{amsmath, amssymb}
\usepackage{kotex}
\usepackage{tensor}
\usepackage{mathtools}
\usepackage{indentfirst}
\usepackage{amsthm}

\usepackage{array} 
\usepackage[section]{placeins}

\newcommand{\eg}{{\it e.g.}}
\newcommand{\ie}{{\it i.e.}}

\newtheorem{theorem}{Theorem}

\makeatletter

\renewcommand{\maketag@@@}[1]{\hbox{\m@th\normalsize\normalfont#1}}%

\makeatother

\title{Entanglement wedge cross section
triangle information and holographic entanglement of assistance}

\author[a]{Xin-Xiang Ju,}
\author[b]{Wen-Bin Pan,}
\author[a,c]{Ya-Wen Sun}
\author[a]{and Yang Zhao}

\emailAdd{juxinxiang21@mails.ucas.ac.cn}
\emailAdd{panwb@buct.edu.cn}
\emailAdd{yawen.sun@ucas.ac.cn}
\emailAdd{zhaoyang20a@mails.ucas.ac.cn}

\affiliation[a]{School of Physical Sciences, University of Chinese Academy of Sciences, Beijing 100190, China}
\affiliation[b]{College of Mathematics and Physics, Beijing University of Chemical Technology, Beijing 100029, China}
\affiliation[c]{Kavli Institute for Theoretical Sciences, University of Chinese Academy of Sciences, Beijing 100049, China}
	
\abstract{
We identify a non-negative and upper-bounded entanglement signal
in holography which is defined as a combination of entanglement wedge cross sections (EWCS) for a tripartite {mixed state} \(ABE\): \(\mathrm{EI}_{\Delta}(A:B|E) = \mathrm{EWCS}(A:EB) + \mathrm{EWCS}(B:EA) - \mathrm{EWCS}(E:AB)\). 
This quantity is an analogue of conditional mutual information (CMI) and shares similar mathematical structures in both quantum information theory and holography. We show that CMI is upper bounded by a quantum information quantity, the entanglement of assistance, which quantifies the entanglement that can be generated between two parties \(A\) and \(B\), given assistance from a third party \(E\). We prove that \(\mathrm{EI}_{\Delta}\) is also upper bounded by the entanglement of assistance in the canonical purification state. We analyze its upper bound by  maximizing \(\mathrm{EI}_{\Delta}(A:B|E)\) over all configurations of the auxiliary subsystem \(E\) in AdS\(_3\)/CFT\(_2\). The maximized \(\mathrm{EI}_{\Delta}\) displays a rich phase structure governed by the cross ratio \(X_{AB}\): it vanishes below a critical threshold and, beyond a second phase transition point, saturates the bound of entanglement of assistance. We comment on the interpretation of \(\mathrm{EI}_{\Delta}\) as characterizing the assisted bipartite {quantum} entanglement between $A$ and $B$ with the help of $E$.}

\arxivnumber{}

\begin{document}
\maketitle
\flushbottom

	
\section{Introduction}\label{sec1}

\noindent Holographic boundary states are typically highly entangled. The entanglement structure of such states is often probed via entanglement entropy, given by the area of the Ryu-Takayanagi/Hubeny-Rangamani-Takayanagi (RT/HRT) surface \cite{Ryu:2006bv,Hubeny:2007xt}. However, entanglement entropy captures only bipartite entanglement in pure states and fails to characterize mixed-state or genuine multipartite entanglement.
It is thus intriguing to explore the entanglement structure beyond the bipartite setting.
In recent years, rapid progress has been made in the holographic quantification of mixed-state and multi-partite entanglement. 
In particular, the entanglement wedge cross section (EWCS) was proposed as the dual of mixed-state correlation measures such as the entanglement of purification and the reflected entropy \cite{Takayanagi:2017knl,Dutta:2019gen}.
Furthermore, some extensions or generalizations of entropy and EWCS, e.g. multi-entropy \cite{Gadde:2022cqi,Penington:2022dhr,Gadde:2023zzj} and multi-EWCS \cite{Umemoto:2018jpc,Bao:2018fso}, have been developed to probe the intricate structure of multi-partite entanglement in holography  {\cite{Balasubramanian:2024ysu,Iizuka:2025bcc,Ju:2025eyn,Akers:2019gcv,Hayden:2021gno,Basak:2024uwc,Liu:2024ulq,Bao:2025psl}}.

In previous works \cite{Ju:2024kuc,Ju:2024hba,Ju:2025tgg}, we studied the multi-partite entanglement in holography by investigating the upper bound for holographic conditional mutual information (CMI) \(I(A:B|E)=S_{AE}+S_{BE}-S_E-S_{ABE}\).
CMI is a key quantity in both quantum information theory and holography. In quantum information theory, its nonnegativity—strong subadditivity \cite{Araki:1970ba}—is one of the most fundamental entropy inequalities, and the squashed entanglement \cite{Li:2014bep,Ju:2023dzo}, defined via a lower bound on CMI, is among the best entanglement measures for bipartite mixed states. In holography, CMI also appears in a  wide range of contexts, including in the c-theorem \cite{Bhattacharya:2014vja}, the differential entropy framework \cite{Balasubramanian:2013lsa}, the holographic entropy cone program \cite{Bao:2015bfa,Hernandez-Cuenca:2022pst,He:2020xuo}, and the entanglement contour \cite{Freedman:2016zud}.

In this work, we propose a new quantity, the EWCS triangle information on a mixed holographic state $ABE$, defined as a linear combination of EWCS terms as follows
\begin{equation}\label{EWCStri}
    \mathrm{EI}_{\Delta}(A:B|E):=\mathrm{EWCS}(A:EB)+\mathrm{EWCS}(B:EA)-\mathrm{EWCS}(E:AB).
\end{equation} 
We will see that it mirrors CMI closely both in its expression and in its mathematical structure. Moreover, by analyzing its upper bound, we find that a nontrivial generalization of the {multi-entanglement phase transition (MPT) diagram, a graphical framework introduced in \cite{Ju:2024kuc,Ju:2024hba} that identifies the critical configurations with maximal CMI,} provides a natural way to characterize and determine the configurations that saturate this bound  {on $\mathrm{EI}_{\Delta}(A:B|E)$}. While generic EWCS combinations are not strictly positive nor upper bounded, the special EWCS combination we proposed retains positivity and is also bounded from above. 

As EWCS corresponds to both entanglement of purification \cite{Takayanagi:2017knl} and reflected entropy \cite{Dutta:2019gen} of the boundary system, the boundary dual of this EWCS triangle information also has two expressions in terms of either the reflected entropy or the entanglement of purification of the boundary system. As our focus is primarily holographic, we will refer to each term in \eqref{EWCStri} simply as an EWCS, without specifying which boundary quantity it corresponds to. Nevertheless, we will discuss the quantum-information aspects of \eqref{EWCStri} under both proposals in the following and in Section \ref{secCMI}.

It is useful to first analyze this EWCS combination within the framework of canonical purification. For a mixed tripartite boundary state \(\rho_{ABE}\) with spectral decomposition \(\rho_{ABE}= \sum_{i} p_i \lvert \psi_i \rangle_{ABE} \langle \psi_i \rvert_{ABE}\),
its canonical purification is the pure state 
\begin{equation}\label{sqrtrho}
   \lvert \sqrt{\rho_{ABE}} \rangle= \sum_{i} \sqrt{p_i} \lvert \psi_i \rangle_{ABE} \lvert \psi_i \rangle_{A^*B^*E^*},
\end{equation}
living in the doubled Hilbert space \( \mathcal{H}_{ABE} \otimes \mathcal{H}_{A^*B^*E^*}\), 
where \(A^*B^*E^*\) is an auxiliary copy of \(ABE\). 
In holography, this corresponds to doubling the bulk entanglement wedge of \(ABE\) (\(\mathrm{EW}(ABE)\)) and gluing it to the original one along the minimal RT surfaces that separate the boundary regions \cite{Dutta:2019gen,Bao:2019zqc,Chu:2019etd,Harper:2020wad}, as illustrated in Figure \ref{wormhole}.
In the replicated geometry, the systems living on the three mouths of the wormhole $AA^*$, $BB^*$ and $EE^*$ share a pure tripartite state.
The third system \(EE^*\) is the complement of \(AA^*BB^*\) and therefore a purification of \(AA^*BB^*\) which possesses a mixed state.
Notice that each of the three EWCS terms in \eqref{EWCStri} corresponds to a holographic entanglement entropy in this setup, \eg, \(\mathrm{EWCS}(A:EB)=S_{AA^*}\). $\mathrm{EI}_{\Delta}(A:B|E)$ is then found to be a half of the mutual information between $AA^*$ and $BB^*$ (\ie, $\mathrm{EI}_{\Delta}(A:B|E)=\frac{1}{2}I(AA^*:BB^*)$) in the state $\lvert \sqrt{\rho_{ABE}} \rangle$.
It follows that $\mathrm{EI}_{\Delta}(A:B|E)$ is strictly positive due to the positivity of the mutual information\footnote{There are some potential extensions to $\mathrm{EI}_{\Delta}(A:B|E)$. 
First, observing that if \(E\) is absent, one has $\mathrm{EI}_{\Delta}(A:B|E)=\frac{1}{2}I(AA^*:BB^*)=\mathrm{EWCS}(A:B)$, one may define the net enhancement in entanglement(correlation) provided by $EE^*$: the EWCS combination $\mathrm{EI}_{\Delta}(A:B|E)-\mathrm{EWCS}(A:B)$. The maximum value of this new combination among possible configurations of \(E\) is non-negative, but the combination itself fails to be non-negative in general.
Second, because $I(AA^*:BB^*)$ is bounded from below by $I(A:B)$ via monogamy of mutual information and strong subadditivity, one may construct another non-negative quantity $\mathrm{EI}_{\Delta}(A:B|E)-I(A:B)$.
Since both the additional terms do not depend on the alterable system \(E\) and remain fixed if \(AB\) is fixed, we focus only on $\mathrm{EI}_{\Delta}(A:B|E)$ in this work.}.

\begin{figure}[H]
\centering
\includegraphics[scale=0.7]{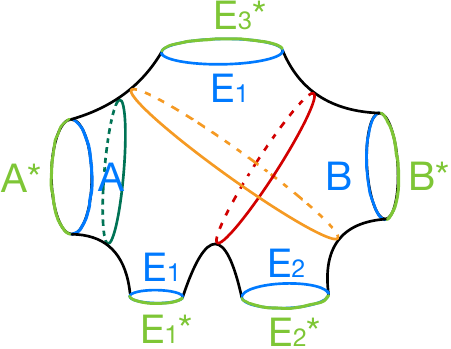}
\caption{{An illustration of a multi-mouth wormhole that is dual to the state \(\lvert \sqrt{\rho_{ABE}} \rangle\) in \eqref{sqrtrho}, with region \(E\) consisting of three intervals. 
The wormhole is constructed by doubling the entanglement wedge of \(ABE\) and gluing them together along the minimal RT surfaces that separate the three boundary regions, \ie, the black curves.
Three candidates of \(\mathrm{EoA}(A:B|E)\) are shown in three different colors in the wormhole. }
}
\label{wormhole}
\end{figure}

\subsection*{Entanglement of assistance (EoA)}

In this paper, we will show that both the CMI and EI$_\Delta$ are upper bounded by a quantum information quantity termed as ``entanglement of assistance" \(\mathrm{EoA}(A:B|T)\) which is a possible mixed state entanglement quantification that also involves a third auxilary party \(T\). Inspired by this upper bound, we propose that the maximum value of $\mathrm{EI}_{\Delta}(A:B|E)$ is a quantification of the {quantum entanglement} between systems \(A\) and \(B\) with the help of a third party \(E\).

\(\mathrm{EoA}(A:B|T)\) was first proposed in \cite{DiVincenzo:1998zz} to be the maximal entanglement that $A$ and $B$ can distill with the assistance of a third system that purifies $AB$.
It was also shown in \cite{Horodecki:2005vvo,Smolin:2005khf} that the entanglement of assistance is given by the minimum entanglement across any bipartite cut of the system which separates \(A\) from \(B\)
\begin{equation}\label{EoA}
    \mathrm{EoA}(A:B|T)=\underset{T_A,T_B}{\mathrm{Min}}(S(A T_A),S(T_B B)),
\end{equation}
where \(T=T_A \bigcup T_B\) is an auxiliary system that purifies $AB$ and the minimum is taken over all possible partitions of \(T\) into \(T_A\) and \(T_B\).

This expression of \(\mathrm{EoA}\) admits a straightforward holographic dual, by identifying the entanglement entropies with the areas of the corresponding RT surfaces.
In the special case of the canonical purification state \(\lvert \sqrt{\rho_{ABE}} \rangle\), which is dual to a multi-mouth wormhole, the holographic EoA can also be directly defined as the area of the minimal surface in the wormhole that separates region \(AA^*\) from \(BB^*\) (\(\mathrm{EoA}(AA^*:BB^*|EE^*)\)).
Several candidates for the minimal surface are illustrated in Figure \ref{wormhole}.
Due to the replication in the canonical purification, \(\mathrm{EoA}(AA^*:BB^*|EE^*)\) equals twice the area of the minimal surface that separates region \(A\) from \(B\) in the entanglement wedge of \(ABE\).
We will refer to this minimal surface in the original geometry as the ``minimal EWCS" and denote it by \(\mathrm{HE}(A:B|E)\),
\begin{equation}\label{ubofEI}
    \mathrm{HE}(A:B|E):=\underset{E_A,E_B}{\mathrm{Min}}(\mathrm{EWCS}(A E_A:E_B B))=\frac{1}{2}\mathrm{EoA}(AA^*:BB^*|EE^*),
\end{equation}
where the minimum is taken over all bipartite divisions of region \(E\) into \(E_A\) and \(E_B\) within the entanglement wedge of $ABE$ in the original state. Note that the left-hand-side of \eqref{ubofEI} is a quantity defined in the original geometry dual to \( \rho_{ABE} \), while the right-hand-side is a half of the EoA in the canonical purification state \(\lvert \sqrt{\rho_{ABE}} \rangle\).

More explicitly, to calculate the upper bound for the EWCS triangle information \(\mathrm{EI}_{\Delta}(A:B|E)\), we compute the maximum value of $\mathrm{EI}_{\Delta}(A:B|E)$ among all configurations of a $n$-partite region $E=\bigcup_{i=1}^{n} E_i$, while regions $A$ and $B$ remain fixed.
We work exclusively in AdS$_3$/CFT$_2$, although the EWCS triangle information can also be generalized to higher dimensions.
A generalization of the MPT diagrams \cite{Ju:2024kuc} is developed to analyze the problem, and it is revealed that $\mathrm{Max}(\mathrm{EI}_{\Delta}(A:B|E))$ has a rich structure that emerges as $X_{AB}$, the cross ratio of \(A\) and \(B\), varies.
We refer to {each distinct pattern of this structure in the parameter space as a ``phase"}, and \(X_{AB}\) works as an ``order parameter".
In different phases the optimal configuration of \(E\) that maximizes $\mathrm{EI}_{\Delta}(A:B|E)$ are qualitatively different, \ie, determined by different conditions involving connecting conditions of entanglement wedges and conditions that constrain EWCSs.
Notably, for each fixed interval number $n$ of \(E\), there exists a critical threshold of $X_{AB}$, below which $\mathrm{Max}(\mathrm{EI}_{\Delta}(A:B|E))$ vanishes. 
It is also observed that $\mathrm{Max}(\mathrm{EI}_{\Delta}(A:B|E))$ saturates the upper bound of entanglement of assistance \eqref{ubofEI} as long as $X_{AB}$ is larger than the second phase transition point.

The rest of this paper is organized as follows.
Section \ref{secCMI} develops the analogy between CMI and the EWCS triangle information and analyzes the bound of both quantities.
Section \ref{sec3} focuses on maximizing the EWCS triangle information by optimizing over \(E\) in the simplest setting where \(E\) is a single interval.
Section \ref{sec4} extends the analysis to the generalized case where $E$ consists of $n$ intervals.
Section \ref{sec5} concludes with implications for our results and comments on extensions to broader classes of entropy/EWCS inequalities.

\section{From CMI upper bound to EWCS upper bound and the entanglement of assistance}\label{secCMI}

\noindent Before any specific discussion of EWCS combinations, it is important to first consider the upper bound of CMI investigated in \cite{Ju:2024kuc} as the foundation of our study. By evaluating the maximal value of $I(A:B|E)$ with $A$ and $B$ fixed and $E$ arbitrarily chosen, \cite{Ju:2024kuc} provides insights on the  multi-partite entanglement structures in holography. We extend such approach to the analysis of EWCS combinations, which further probes the holographic multi-partite entanglement structures in the canonical purification geometry. Meanwhile, \cite{Ju:2024kuc} introduced a set of crucial tools of upper bounding holographic entanglement measures, most notably the MPT diagram. These tools can be naturally generalized to the EWCS scenario.

Specifically, in \ref{secubCMI} we briefly review the holographic upper bound of CMI discovered in \cite{Ju:2024kuc}, and in \ref{QIcmi} we further discuss the quantum information theoretic interpretation of this bound. Subsection \ref{cmitoewcs} focuses on the generalization from CMI to the upper bounds of EWCS combinations, especially the quantity that we study\textemdash the EWCS triangle information (\ref{EWCStri}). Finally, in \ref{qiewcs} the quantum information aspects of (\ref{EWCStri}) are explored.
\subsection{Holographic upper bound of CMI and the entanglement of assistance}\label{secubCMI}

\noindent In this section, we introduce the derivation of the holographic upper bound of the CMI $I(A:B|E)$ in \cite{Ju:2024kuc}. When fixing $A$ and $B$ while fine-tuning $E$, we first found that $I(A:B|E)$ reaches its maximum at the critical point of a series of simultaneous entanglement wedge phase transitions. Specifically, with $m$ intervals of $E$, there are $2m$ phase transition conditions uniquely fixing the $2m$ endpoints.

Aiming at a clear illustration of the various critical configurations, the MPT diagrams were introduced in \cite{Ju:2024kuc}. Each diagram corresponds to a critical configuration, with constraints ensuring that the intervals of $E$ are precisely fixed. From the diagrams, the (dis)connectivity of the entanglement wedges, as well as the $2m$ phase transition conditions can be directly read off. Across different diagrams, these phase transition conditions vary.

Thus, by comparing the CMI values of the diagrams, we first obtained three constraints for the diagram with maximal CMI. 
\begin{itemize}
    \item The entanglement wedge of $ABE$ must be totally connected.
    \item The entanglement wedge of $E$  must be totally disconnected.
    \item $I(A:E)=I(B:E)=0$.
\end{itemize}
Further analysis also provides rules that allow us to directly identify the diagram with the largest CMI. 

In a word, the MPT diagram encodes the phase-transition conditions satisfied by the maximal configuration, and thereby captures the (dis)connectivity requirements of all the entanglement wedges. Details about the MPT diagrams are introduced in Section {\ref{sec4.1}}. We will see that, when upper bounding EWCS combinations, these diagrams continue to be powerful in revealing the essential conditions that the maximal configuration must satisfy.

Meanwhile, with the (dis)connectivity of entanglement wedges in the maximal configuration known from the MPT diagrams, we have successfully obtained a general upper bound of CMI that is valid for arbitrary bulk geometries and dimensions:

\begin{equation}\label{FinalCMI}
    I(A:B|E) \leq 2S_D \quad\quad (D \supseteq A, \ D^c \supseteq B).
\end{equation}
Here $D\supseteq A$ is defined as the boundary region where $S_D$ corresponds to the minimal surface separating $\mathrm{EW}(A)$ from $\mathrm{EW}(B)$. This bound can be approached arbitrarily closely as the UV cutoff $\epsilon \to 0$, allowing the number of intervals in $E$ to become sufficiently large. 

In quantum information, this upper bound coincides with the definition of the { entanglement of assistance}\footnote{We thank Bart{\l}omiej Czech for that.} $\mathrm{EoA}(A:B|T)$ \cite{Smolin:2005khf,Horodecki:2005vvo}, the maximal entanglement that $A$ and $B$ can distill with the help of another subsystem $T$ that purifies $AB$. A more detailed quantum information theoretic analysis of this upper bound will be presented in the next subsection.

\subsection{Quantum information upper bound of CMI and the entanglement of assistance}\label{QIcmi}
\noindent In holography, we conclude that the upper bound of CMI $I(A:B|E)$ is given by the minimal surface that separates $A$ and $B$, while in quantum information, this corresponds to the entanglement of assistance. Such results naturally leads us to further explore the information theoretic interpretation of this upper bound.

Using subadditivity and strong subadditivity twice each, we can show that the CMI is always smaller than the entanglement of assistance, as follows.
\begin{theorem}
 Conditional mutual information is always smaller than the entanglement of assistance.
\end{theorem}
\begin{proof}
Given a pure state density matrix $\rho_{ABEG}$, where $E=e_1\cup e_2$ and $G=g_1\cup g_2$, denote $T=e_1g_1$\footnote{Note that $e_1,e_2$ and $g_1,g_2$ serve as arbitrary partitions of $E$ and $G$ respectively. These partitions may be trivial, \eg, $e_1=\varnothing$ and $e_2=E.$}. We now prove that $I(A:B|E)\leq 2S_{AT}$ by applying subadditivity and strong subadditivity twice each, as follows\footnote{Although the proof is straightforward in principle, we have not found it in previous works.}:
\begin{equation}
    \begin{aligned}
    2S_{AT}&=(S_{A e_1g_1})+S_{Be_2g_2}\\
    &\geq (S_{A e_1g_1g_2})-(S_{g_1g_2})+[(S_{g_1})]+[S_{Be_2g_2}]\\
    &\geq S_{A e_1g_1g_2}+[S_{Be_2g_1g_2}]-S_{g_1g_2}\\
    &=(S_{Be_2})+S_{A e_1}+(S_{e_1e_2})-S_{e_1e_2}-S_{g_1g_2}\\
    &\geq (S_{Be_1e_2})+[(S_{e_2})]+[S_{A e_1}]-S_{e_1e_2}-S_{g_1g_2}\\
    &\geq [S_{A e_1e_2}]+S_{Be_1e_2}-S_{e_1e_2}-S_{g_1g_2}\\
    &= I(A:B|E).
    \end{aligned}
\end{equation}
In this derivation, all terms in parentheses correspond to applications of strong subadditivity, while all terms in square bracket correspond to applications of subadditivity. For all these inequalities to be saturated, we must have
\begin{equation}\label{saturation}
    \begin{aligned}
        I(Ae_1:g_2|g_1)&=0, \quad\quad I(Be_2:g_1|g_2)=0,\\
        I(g_1:Be_2g_2)&=0, \quad\quad I(g_2:Ae_1g_2)=0,\\
        I(e_1:B|e_2)&=0, \quad\quad I(e_2:A|e_1)=0,\\
        I(Ae_1:e_2)&=0,\quad\quad I(Be_2:e_1)=0,
    \end{aligned}
\end{equation}
where the four equalities in the left column are equivalent to those in the right column. Finally, we have shown that in quantum information theory, the CMI is also upper bounded by the entanglement of assistance {$\mathrm{EoA}(A:B|T)=\underset{T_A,T_B}{\mathrm{Min}}(S(A T_A),S(T_B B))$}. The nontrivial observation in holography is that the entanglement of assistance is always a \textbf{strict} upper bound of the CMI, because in principle we can always find a subregion $E$ that saturates the bound. This means that, for arbitrary $A$ and $B$, we can always find $E=e_1\cup e_2$ and $G=g_1\cup g_2$ such that (\ref{saturation}) is satisfied. This reveals the absence of genuine bipartite entanglement (like Bell pairs) between smaller subsystems in $A$, $B$ and their complements. In {\cite{Ju:2024kuc,Ju:2024hba}} we conclude that the bipartite entanglement can always ``emerge" from the tripartite entanglement in a finer partition of the bipartite system.
\end{proof}
\subsection{From CMI to EWCS combination}\label{cmitoewcs}

\noindent In the previous sections, we have already seen that evaluating the holographic upper bound of CMI leads to meaningful results in holography and even provides insights into quantum information theory. This motivates us to study the upper bounds of other holographic entanglement entropy combinations \cite{Ju:2025tgg}. Unfortunately, MPT diagrams lose their  power in general cases; that is, the upper bound may fail to correspond to any geometric configuration that reaches a multi-phase-transition point. As a result, it becomes difficult to determine the strict upper bound using a single geometric configuration. Instead, we developed methods to generate families of inequalities and compare them to identify the tightest one\textemdash a complicated procedure that we will not elaborate on here. In this paper, we aim to generalize the analysis of upper bounding CMI in a different direction: replacing the CMI with a nontrivial EWCS combination.

As an analogue of the CMI, we require the EWCS combination to be strictly positive in holography. We also demand that it takes a form similar to the CMI, and to have a highly nontrivial upper-bound configuration. To clarify what we mean by nontrivial, consider the combination $\mathrm{EWCS}(E:AB)-\mathrm{EWCS}(E:A)-\mathrm{EWCS}(E:B)$, with $A$, $B$ fixed and $E$ arbitrarily chosen. This quantity has a trivial upper bound: we can simply make $\mathrm{EW}(AE)$ and $\mathrm{EW}(BE)$ disconnected, while $\mathrm{EW}(ABE)$ remains connected. One then can directly choose the configuration maximizing $I(A:B|E)$ to make $\mathrm{EWCS}(E:AB)$ infinitely approach $S_{AB}$. Consequently, the upper bound of $\mathrm{EWCS}(E:AB)-\mathrm{EWCS}(E:A)-\mathrm{EWCS}(E:B)$ can be made arbitrarily close to its information theoretical upper bound, $S_{AB}$. 

After ruling out all trivial choices and all EWCS combinations that are not strictly positive, we arrive at the simplest nontrivial EWCS combination, given by
\begin{equation}\label{triieq}
    \mathrm{EWCS}(A:EB)+\mathrm{EWCS}(B:AE)-\mathrm{EWCS}(E:AB)\geq 0.
\end{equation}
The positivity of this quantity in holography can be strictly proven as follows. The terms $\mathrm{EWCS}(A:EB)$ and $\mathrm{EWCS}(B:AE)$ are the areas of the minimal surfaces within $\mathrm{EW}(ABE)$ which separate $A$ from $BE$ and $B$ from $AE$ respectively. The union of these two surfaces naturally serves as a candidate surface that separates $E$ from $AB$ in $\mathrm{EW}(ABE)$. Therefore, the sum of their areas must be greater than or equal to $\mathrm{EWCS}(E:AB)$. In the canonical purification geometry \cite{Dutta:2019gen,Bao:2019zqc,Chu:2019etd,Harper:2020wad} where the system $A^*B^*E^*$ purifies $ABE$, we have 
\begin{equation}
\begin{aligned}
    \mathrm{EWCS}(A:EB)+\mathrm{EWCS}(B:AE)-\mathrm{EWCS}(E:AB)&=\frac12(S_{AA^*}+S_{BB^*}-S_{EE^*})\\
    &=\frac12I(AA^*:BB^*).
\end{aligned}
\end{equation}
It can be observed that this quantity is not only strictly positive, but is in fact greater than $2I(A:B)$:
\begin{equation}
    I(AA^*:BB^*)\geq I(A:BB^*)+I(A^*:BB^*)\geq I(A:B)+I(A^*:B^*),
\end{equation}
where we have used MMI \cite{Hayden:2011ag} and SSA in the first and second inequalities respectively. Note that this relation is derived in the canonical purification geometry and is satisfied only in holography.

To express this combination in a more compact form, we define a quantity called the EWCS triangle information, denoted by $\mathrm{EI}_\Delta(A:B|E)$, as 
\begin{equation}\label{ewtriangledelta}
    \mathrm{EI}_\Delta(A:B|E)=\mathrm{EWCS}(A:EB)+\mathrm{EWCS}(B:AE)-\mathrm{EWCS}(E:AB).
\end{equation}
Note that the CMI can be rewritten in a similar form
\begin{equation}
    I(A:B|E)=I(A:EB)+I(B:AE)-I(E:AB)-I(A:B),
\end{equation}
where we can simply replace the mutual information with the corresponding EWCS to turn the CMI into $\mathrm{EI}_\Delta(A:B|E)$. Note that $I(A:B)$ could also be inserted in the definition for $\mathrm{EI}_\Delta(A:B|E)$, which however is not important when we fix $A$ and $B$, therefore, we omit the mutual information term in this definition. Moreover, the  exact upper bound of $\mathrm{EI}_\Delta(A:B|E)$ is related to the entanglement of assistance in the canonical purification geometry. 

{
In the following, we show that the EWCS triangle information is also bounded by the entanglement of assistance $\mathrm{EoA}(AA^*,BB^*|EE^*)$ in the canonical purification geometry. Since \(AA^*BB^*EE^*\) is a pure state, we have $I(AA^*:BB^*|EE^*)=I(AA^*:BB^*)$.
The EWCS triangle information $\mathrm{EI}_{\Delta}(A:B|E)$ is thus related to the CMI $I(AA^*:BB^*|EE^*)$ by
\begin{equation}\label{EICMI}
    \mathrm{EI}_{\Delta}(A:B|E)=\frac{1}{2}I(AA^*:BB^*)=\frac{1}{2}I(AA^*:BB^*|EE^*).
\end{equation}
As we have proved above that the CMI is bounded by twice the entanglement of assistance in the original geometry, this inequality also holds in the canonical purification geometry: $I(AA^*:BB^*|EE^*)\le 2\mathrm{EoA}(AA^*,BB^*)$.
Then by virtue of \eqref{EICMI} we have
\begin{equation}
    \mathrm{EI}_{\Delta}(A:B|E) \le \mathrm{EoA}(AA^*:BB^*),
\end{equation}
\ie, the upper bound of the EWCS triangle information is the entanglement of assistance of $AA^*$ and $BB^*$ in the canonical purification state.
Note that $\mathrm{EoA}(AA^*:BB^*|EE^*)=2 \mathrm{HE}(A : B|E) = 2\underset{E_A,E_B}{\mathrm{Min}}(\mathrm{EWCS}(A E_A:E_B B))$, which reflects the quantum entanglement between $A$ and $B$ with the help of $E$.}

\subsection{The saturation of {the lower bound of} the EWCS triangle information}\label{qiewcs}
\noindent Here, we briefly discuss the quantum information aspects of the EWCS triangle inequality (\ref{triieq}). As we have shown earlier, it has a clear interpretation in canonical purification\textemdash half of the mutual information $I(AA^*:BB^*)$. In this subsection, we adopt the EoP=EWCS proposal {\cite{Takayanagi:2017knl}}, and investigate this combination for general quantum states beyond holography.

As shown in Figure \ref{SSEOP}, using the surface-state correspondence \cite{Miyaji:2015yva}, we can take $A'$ to be the subsystem that minimizes $S_{AA'}$, and $B'$ the one minimizing $S_{BB'}$. Taking the complement of $AA'BB'E$ as subsystem $O$, we have
\begin{equation}\label{EPRTIproof}
    E_P(A:BE)+E_P(B:AE)=S_{AA'}+S_{BB'}\geq S_{AA'BB'}=S_{EO}\geq E_P(E:AB),
\end{equation}
where the first inequality follows from subadditivity and the second from the definition of $E_P(E:AB)$.
\begin{figure}[h]
\centering
     \includegraphics[width=5cm]{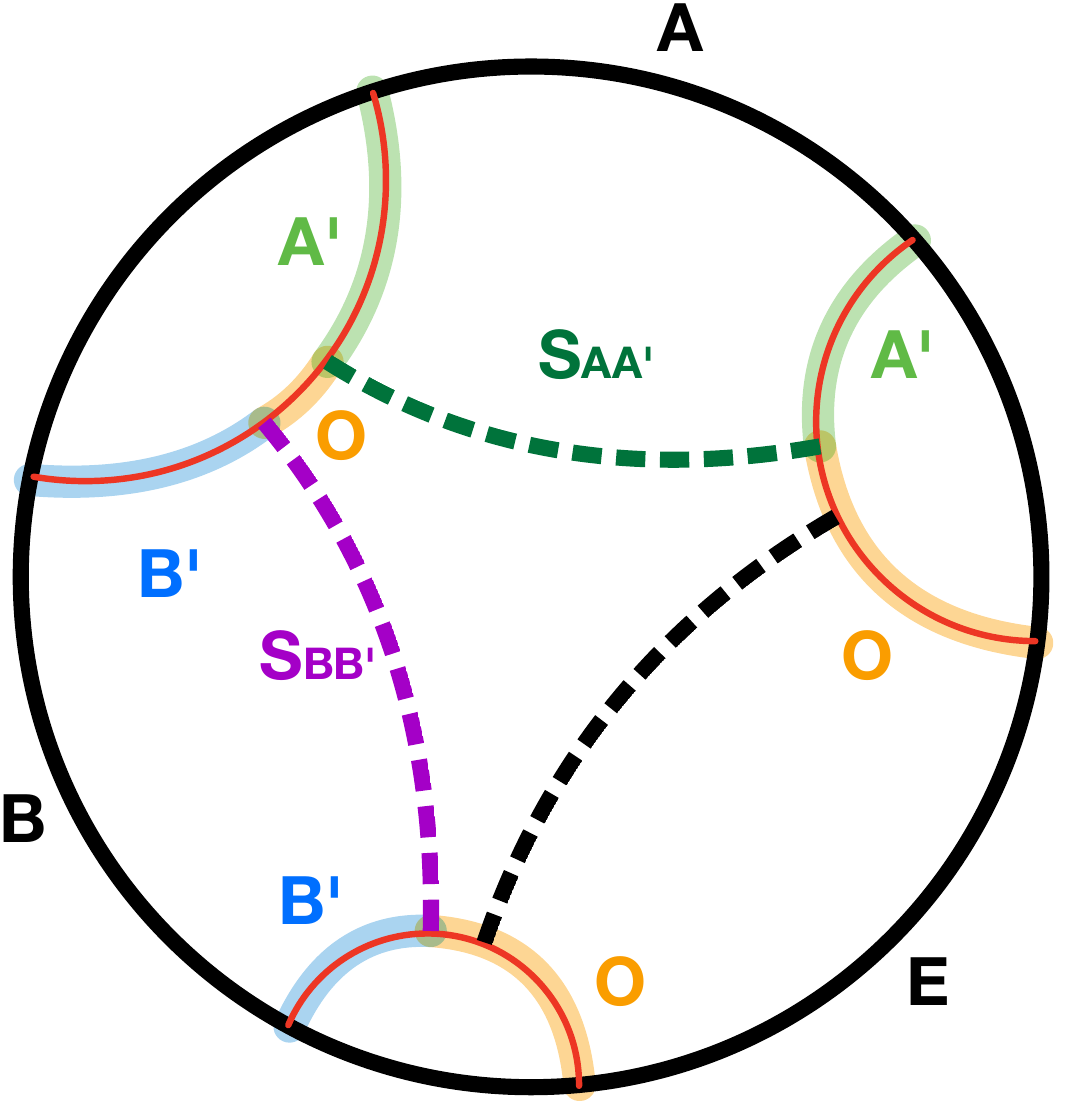}
 \caption{Proof of EWCS triangle inequality in holography.}
\label{SSEOP}
\end{figure}
For general quantum states, can we apply this proof directly? One might worry that the argument breaks down if $A'$ and $B'$ share overlapping degrees of freedom, namely subsystem $U=A'\cap B'$. However, we can show that this cannot happen because of weak monotonicity:
$S_{AA'}+S_{BB'}\geq S_{AA'\backslash U}+S_{BB'\backslash U}.$
This tells us that deleting $U$ from both $A'$ and $B'$ simultaneously does not increase $S_{AA'}+S_{BB'}$. Therefore, in the process of evaluating the entanglement of purification, we can always choose $A'$ and $B'$ so that they share no overlapping d.o.f. 

However, the step in (\ref{EPRTIproof}) that may fail is the assumption that the density matrix $\rho_{AA'BB'E}$ is consistent with the given density matrices $\rho_{AA'BE}$ and $\rho_{ABB'E}$. This issue is part of the quantum marginal problem \cite{Schilling:2014sjf,Schilling:2015wtq}. One special feature of the holographic setting is that the surface-state correspondence always allows us to construct a consistent $\rho_{AA'BB'E}$, thus ensuring the validity of proof (\ref{EPRTIproof}). If such a consistent $\rho_{AA'BB'E}$ does exist, we can claim that the saturation of the EWCS triangle inequality (\ref{triieq}) is equivalent to the fact that
\begin{itemize}
    \item There exist subsystems $A'B'E'$ that purify $\rho_{ABE}$ such that $S_{AA'},S_{BB'},$ and $S_{EE'}$ each attain their minimal values among all possible extensions of $\rho_{ABE}$, with $I(AA':BB')=0$.
\end{itemize}

\section{{$\mathrm{Max}(\mathrm{EI}_{\Delta}(A:B|E))$ with E being one interval and its maximum}}\label{sec3}

\noindent To determine the maximum value of the quantity $\mathrm{EI}_{\Delta}(A:B|E)$ in a holographic system, we proceed as follows. First, we fix the boundary regions 
$A$ and $B$, while varying the third region $E$. We work in the AdS$_3$/CFT$_2$ framework, where $E$ can in principle consist of any number 
$n$ of disjoint intervals. To build intuition, this section focuses on the simplest case: $E$ contains a single interval ($n=1$), {in which the maximization of $\mathrm{EI}_{\Delta}(A:B|E)$ can be computed analytically. The setup of this $n=1$ case is shown in Figure \ref{Fig5}.} 
The general multi-interval case will be addressed in the next section.

{The entanglement wedge cross section (EWCS)\cite{Takayanagi:2017knl} is the minimal area surface that splits the entanglement wedge into two subregions whose boundaries on the conformal boundary are \(A\) and \(B\) respectively.}
{In global AdS$_3$, the EWCS of two intervals $I_1$ and $I_2$ can be represented by the cross ratio $X_{I_1 I_2}$ as\footnote{{We use cross ratios as basic variables in our computation. The definition of a cross ratio is introduced in Appendix \ref{sec2}.}} 
\begin{equation}\label{EWCSAB}
    \mathrm{EWCS}(I_1:I_2)=\frac{c}{3}\log \left(\sqrt{X_{I_1 I_2}}+\sqrt{1+X_{I_1 I_2}} \right),
\end{equation}
where \(c\) is the central charge and we have used the identity \(\frac{L}{G_N}=\frac{2c}{3}\)\cite{Brown:1986nw}.
It is easy to check that \eqref{EWCSAB} holds in the special case where the entanglement wedge of $I_1 I_2$ has {interchange} symmetry, and it is conformal invariant.
Since all other configurations can be brought to be {interchange} symmetric by conformal transformations, \eqref{EWCSAB} holds for all the configurations of $I_1 I_2$.
$\mathrm{EWCS}(I_1:I_2)$ may also be defined by the cross ratios of two gap regions $G_1$ and $G_2$ which are ``inside" and ``outside" $I_1$ and $I_2$, since we have $X_{I_1 I_2}=\frac{1}{X_{G_1 G_2}}$.
Therefore, we can write
\begin{equation}\label{EWCSDD}
    \mathrm{EWCS}(I_1:I_2)=\frac{c}{3}\log \left(\sqrt{\frac{1}{X_{G_1 G_2}}}+\sqrt{1+\frac{1}{X_{G_1 G_2}}} \right).
\end{equation}}

The prescription for 
$\mathrm{EI}_{\Delta}(A:B|E)$ involves a sum of three entanglement wedge cross-section terms. 
{In the case where $E$ consists of a single interval ($n=1$), the whole system has three degrees of freedom due to the conformal symmetry.
The specific geodesic configuration contributing to each term is controlled by the cross ratios among the gap regions, $X_{D_0 D_1}$, $X_{D_0 D_2}$ and $X_{D_1 D_2}$, as dictated by \eqref{EWCSDD}.
Although these could serve as the three independent variables in the analysis of $\mathrm{Max}(\mathrm{EI}_{\Delta}(A:B|E))$, we instead transform them to cross ratios involving the three intervals, $X_{AE}$, $X_{BE}$ and $X_{AB}$, as shown later.
This choice is more convenient because $A$ and $B$ should be fixed and $X_{AB}$ remains constant.}
For each fixed $X_{AB}$, the cross ratios $X_{AE}$ and $X_{BE}$ span a two-dimensional parameter space. At every point in this space, we can compute a value for $\mathrm{EI}_{\Delta}(A:B|E)$.

This parameter space naturally divides into several distinct regions. Within each region, $\mathrm{EI}_{\Delta}(A:B|E)$
 is computed from a different set of geodesics, leading to different functional dependencies on the parameters. The structure of this partition—specifically, the number and arrangement of these regions—depends critically on the value of $X_{AB}$. As $X_{AB}$ is tuned from zero to larger values, this structure undergoes qualitative changes at certain critical points. 
 We refer to the distinct pattern of regions between two successive critical points as a phase. $X_{AB}$ could be viewed as an ``order parameter" then.

Thus, as $X_{AB}$ increases from zero, the system passes through several phases. In each phase, the parameter plane is partitioned differently, and within each sub-region, $\mathrm{EI}_{\Delta}(A:B|E)$
is given by a distinct geometric quantity. The maximum value of
$\mathrm{EI}_{\Delta}(A:B|E)$ within a phase typically occurs at one particular intersection point between these regions. By finding this maximum for each phase, we obtain the maximum value for every $X_{AB}$, allowing us to plot the maximum $\mathrm{EI}_{\Delta}(A:B|E)$
 as a function of 
$X_{AB}$.

\begin{figure}[H]
\centering
\includegraphics[width=5cm]{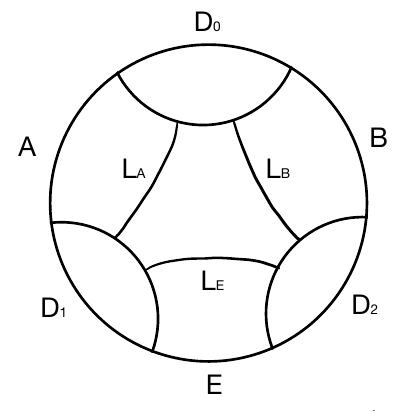}
\caption{Calculation of the EWCS triangle information $\mathrm{EI}_{\Delta}(A:B|E)$ when tuning E, in the case where $n=1$ and the entanglement wedge of $ABE$ is fully connected. Three EWCSs $L_A$, $L_B$ and $L_E$, which are relevant for the computation of $\mathrm{EI}_{\Delta}(A:B|E)$, are shown in the diagram.
 }
\label{Fig5}
\end{figure}

{We now begin the explicit computation for the case where $E$ consists of a single interval ($n=1$) and} analyze the behavior of $\mathrm{EI}_{\Delta}(A:B|E)$ with $X_{AE}$, $X_{BE}$, and $X_{AB}$ being three independent variables. All other cross ratios and EWCSs can be expressed in terms of them. For fixed $X_{AB}$, we maximize $\mathrm{EI}_{\Delta}(A:B|E)$ over the pair ($X_{AE}$, $X_{BE}$), and then vary $X_{AB}$ to study how it affects the maximized value.
For notational convenience, we denote reciprocals with an overbar: $\overline{X}_{AB} := 1/X_{AB}$, $\overline{X}_{AE} := 1/X_{AE}$, $\overline{X}_{BE} := 1/X_{BE}$, etc. 
It is also useful to define the hat symbol of a quantity $Y$ as $\widehat{Y} := e^{(3/c）Y}$, where $Y$ denotes the EWCS or the mutual information; we will often use this notation to remove logarithms and the overall factor $c/3$ appearing in the definition of quantities such as the EWCS and mutual information.

For $X_{AB} < 1$, it is easy to find that
the EWCS triangle information \(\mathrm{EI}_{\Delta}(A:B|E)\) is nonzero only when the entanglement wedge of the region $ABE$ is fully connected. However, for $X_{AB} \geq 1$, it can remain nonzero under an additional condition that in $ABE$, only the entanglement wedge of the subregion $AB$ is connected.
We first analyze the regime $X_{AB} < 1$, where only the fully connected case requires detailed analysis, since non-fully connected configurations yield zero. We will later turn to the case $X_{AB} \geq 1$.

The conditions for the entanglement wedge of $ABE$ to be fully connected are obtained in \eqref{con} of Appendix \ref{sec2.2} as follows
\begin{equation}\label{CC}
\begin{aligned}
   \frac{X_{AB}}{X_{D_1 D_2}} > 1, \quad \frac{1}{X_{D_0 D_1}} > 1, \quad
   \frac{1}{X_{D_0 D_2}} > 1, \quad \frac{1}{X_{D_1 D_2}} > 1,
\end{aligned}    
\end{equation} where the location of the gap regions $D_{0,1,2}$ could be found in Figure \ref{Fig5}.
In the fully connected case, each EWCS term in the combination $\mathrm{EI}_{\Delta}(A:B|E)$ has two possible configurations, as shown in Figure \ref{Fig5}. Therefore we have the following results for each of the terms
\begin{equation}
\begin{aligned}
    \widehat{\mathrm{EWCS}}(A:BE)=\mathrm{Min} (\widehat{L}_A, \widehat{L}_B \widehat{L}_E)\\
    \widehat{\mathrm{EWCS}}(B:AE)=\mathrm{Min} (\widehat{L}_B,  \widehat{L}_A \widehat{L}_E)\\
    \widehat{\mathrm{EWCS}}(E:AB)=\mathrm{Min}(\widehat{L}_E, \widehat{L}_A \widehat{L}_B),\\
\end{aligned}
\end{equation}
where 
\begin{equation}\label{LABE}
\begin{aligned}
    \widehat{L}_A= \sqrt{\overline{X}_{D_0 D_1}}+\sqrt{1+\overline{X}_{D_0 D_1}},\\
    \widehat{L}_B=\sqrt{\overline{X}_{D_0 D_2}}+\sqrt{1+\overline{X}_{D_0 D_2}},\\
    \widehat{L}_E=\sqrt{\overline{X}_{D_1 D_2}}+\sqrt{1+\overline{X}_{D_1 D_2}},\\
\end{aligned}
\end{equation}
as shown in \eqref{EWCSDD}.

The problem of finding the maximal value of $\mathrm{EI}_{\Delta}(A:B|E)=\mathrm{EWCS}(A:EB)+\mathrm{EWCS}(B:EA)-\mathrm{EWCS}(E:AB)$ is now reduced to maximizing the following ratio,
\begin{equation}\label{ratioE}
    \widehat{\mathrm{EI}}_{\Delta}(A:B|E)=\frac{\mathrm{Min}(\widehat{L}_A, \widehat{L}_B \widehat{L}_E) \mathrm{Min}(\widehat{L}_B, \widehat{L}_A \widehat{L}_E)}{\mathrm{Min}(\widehat{L}_E, \widehat{L}_A \widehat{L}_B)}.
\end{equation}
Since $\widehat{L}_A$, $\widehat{L}_B$ and $\widehat{L}_E$ are all larger than $1$ as seen from \eqref{LABE} due to the positivity of the lengths, this ratio has only four possible values\footnote{{For instance, if $\widehat{L}_B \widehat{L}_E<\widehat{L}_A$, we have $\widehat{L}_B <\widehat{L}_A$ and $\widehat{L}_E <\widehat{L}_A$. Then $\widehat{\mathrm{EWCS}}(B:EA)$ should be $\widehat{L}_B$ rather than $\widehat{L}_A \widehat{L}_E$. Therefore, values such as $\widehat{L}_A \widehat{L}_B \widehat{L}_E$ and $\widehat{L}_E^2$ are not possible.
}}
{\begin{equation}\label{posval}
\begin{aligned}
\widehat{\mathrm{EI}}_{\Delta}(A:B|E)=\left\{\begin{matrix}
  \frac{\widehat{L}_A \widehat{L}_B}{\widehat{L}_E}, & \quad \widehat{L}_E<\widehat{L}_A \widehat{L}_B, \; \widehat{L}_A<\widehat{L}_B \widehat{L}_E, \; \widehat{L}_B<\widehat{L}_A \widehat{L}_E\\
  \widehat{L}_A^2, & \widehat{L}_B>\widehat{L}_A \widehat{L}_E\\
  \widehat{L}_B^2, & \widehat{L}_A>\widehat{L}_B \widehat{L}_E\\
  1, & \widehat{L}_E>\widehat{L}_A \widehat{L}_B
\end{matrix}\right.
\end{aligned}
\end{equation}
provided that the entanglement wedge of \(ABE\) is fully connected.}
Which of these four values should be chosen depends on the parameter values of ${X}_{AE}$, ${X}_{BE}$ and ${X}_{AB}$.

We need to express the results \eqref{posval} in terms of the parameters ${X}_{AE}$, ${X}_{BE}$ and ${X}_{AB}$, or equivalently $\overline{X}_{AE}$, $\overline{X}_{BE}$, and $\overline{X}_{AB}$. {The values of }$\widehat{L}_A, \widehat{L}_B, \widehat{L}_E$ are already given by the cross ratios of the gap regions $\overline{X}_{D_0 D_1}$, $\overline{X}_{D_0 D_2}$, $\overline{X}_{D_1 D_2}$ as in \eqref{LABE}.
These can be converted into our preferred free variables  $\overline{X}_{AE}$, $\overline{X}_{BE}$, and $\overline{X}_{AB}$ using the following identities
\begin{equation}\label{idX1}
\begin{aligned}
  \overline{X}_{D_1 D_2}-\overline{X}^2_{AB} \overline{X}_{D_0 D_1} \overline{X}_{D_0 D_2}+\overline{X}_{AB}(1+\overline{X}_{D_0 D_1}+\overline{X}_{D_0 D_2}+\overline{X}_{D_1 D_2})=0,\\
  \overline{X}_{D_0 D_1}-\overline{X}^2_{BE} \overline{X}_{D_0 D_2} \overline{X}_{D_1 D_2}+\overline{X}_{BE}(1+\overline{X}_{D_0 D_1}+\overline{X}_{D_0 D_2}+\overline{X}_{D_1 D_2})=0,\\
  \overline{X}_{D_0 D_2}-\overline{X}^2_{AE} \overline{X}_{D_0 D_1} \overline{X}_{D_1 D_2}+\overline{X}_{AE}(1+\overline{X}_{D_0 D_1}+\overline{X}_{D_0 D_2}+\overline{X}_{D_1 D_2})=0,
\end{aligned}
\end{equation}
derived from their geometric relations. Detailed derivations can be found in Appendix \ref{CRidX1}. From these identities, we can express  $\overline{X}_{D_0 D_1}$, $\overline{X}_{D_0 D_2}$, $\overline{X}_{D_1 D_2}$ using  $\overline{X}_{AE}$, $\overline{X}_{BE}$,  $\overline{X}_{AB}$
\begin{equation}\label{XDD}
\begin{aligned}
  &\overline{X}_{D_0 D_1}=\frac{Y}{2 \overline{X}_{AB} \overline{X}_{AE}},\quad
  \overline{X}_{D_0 D_2}=\frac{Y}{2 \overline{X}_{AB} \overline{X}_{BE}},\quad
  \overline{X}_{D_1 D_2}=\frac{Y}{2 \overline{X}_{AE} \overline{X}_{BE}}\\
  &Y=1+\overline{X}_{AB}+\overline{X}_{AE}+\overline{X}_{BE}+\sqrt{(1+\overline{X}_{AB}+\overline{X}_{AE}+\overline{X}_{BE})^2+4\overline{X}_{AB} \overline{X}_{AE} \overline{X}_{BE}}.
\end{aligned}
\end{equation}
 A useful observation from \eqref{XDD} is 
\begin{equation}\label{idX2}
    \frac{\overline{X}_{BE}}{\overline{X}_{D_0 D_1}}=\frac{\overline{X}_{AE}}{\overline{X}_{D_0 D_2}}=\frac{\overline{X}_{AB}}{\overline{X}_{D_1 D_2}}.
\end{equation} 
This identity can also be demonstrated diagrammatically, as the three different expressions correspond to three different kinds of combinations derived from subtracting diagrams, as shown in Figure \ref{Fig3b}.
The ratio \(\widehat{\mathrm{EI}}_{\Delta}(A:B|E)\) can then be expressed as a function of \(X_{AE}\), \(X_{BE}\), and \(X_{AB}\). The expression is rather lengthy, so we do not write it out here.

We now analyze how the parameter space 
\((X_{AE},X_{BE})\) is partitioned into different regions, each corresponding to a distinct functional expression of \(\mathrm{EI}_{\Delta}(A:B|E)\). At \(X_{AB}<1\), 
the combination \(\mathrm{EI}_{\Delta}(A:B|E)\) may vanish in two regions of the parameter space. The first is the region where the entanglement wedge of $ABE$ is {not totally connected} when \(X_{AE}\) and \(X_{BE}\) are small enough, which corresponds to the parameter region outside \eqref{CC}. 
{This region exists at all values of \(X_{AB}<1\) and is pictured in orange in Figures \ref{FigXABCC1} and \ref{FigXABCC2}.}
The second is the region where \(\widehat{\mathrm{EI}}_{\Delta}(A:B|E)\) picks the branch of 1 in \eqref{posval} when the entanglement wedge of $ABE$ is fully connected. 
In this region, both \(X_{AE}\) and \(X_{BE}\) are large enough and fall within a range where \(\widehat{L}_E > \widehat{L}_A \widehat{L}_B\).
This region only exists at \(X_{AB}< \frac{1}{2}(\sqrt{2} - 1) \approx 0.2071\) and {is pictured in purple in e.g. Figure \ref{FigXABCC1} (a), (b) and (c). The panel (d) of this figure shows that this region is no longer present in the parameter plane for \(X_{AB}=0.3\).}

The regions in the parameter plane for nonzero values of \(\widehat{\mathrm{EI}}_{\Delta}(A:B|E)\) could be divided into three regions corresponding to the first three values in \eqref{posval} {and the three regions are plotted in Figure \ref{FigXABCC1} and \ref{FigXABCC2} in different colors: white for $\frac{\widehat{L}_A \widehat{L}_B}{\widehat{L}_E}$, green for $\widehat{L}_A^2$, and blue for $\widehat{L}_B^2$. }
The boundaries between the four potential values \eqref{posval} ({dashed curves}) of the combination and the four connection conditions \eqref{CC} （{solid curves}） are shown in Figures \ref{FigXABCC1} and \ref{FigXABCC2} with several representative values for \(X_{AB}\).

The three nonzero regions do not necessarily exist simultaneously for all values of \(X_{AB}\). As we will show explicitly, at \(X_{AB}<0.0952\), none of the three nonzero regions could exist in the parameter space. As \(X_{AB}\) increases, one or all three nonzero regions could emerge. Thus, based on the presence or absence—and the distribution—of these three nonzero regions and the two zero regions in the parameter plane, the system can be seen as exhibiting several distinct phases, with \(X_{AB}\) acting as an ``order parameter". Within each phase, the maximum value of the EWCS triangle information \(\mathrm{EI}_{\Delta}(A:B|E)\) can be determined by comparing its values across all regions. As \(X_{AB}\) increases, the system undergoes phase transitions from one phase to another and several critical values of \(X_{AB}\) could be identified. 

For the case of $E$ having only one interval ($n=1$), only {four} phases exist, whereas larger $n$ yields more phases. In the following, we list the four phases for $n=1$: {three} at \(X_{AB}<1\) and the other at \(X_{AB}>1\).

\begin{figure}[h]
\centering
\subfigure[]{
		\begin{minipage}[b]{.47\linewidth}
			\centering
			\includegraphics[scale=0.22]{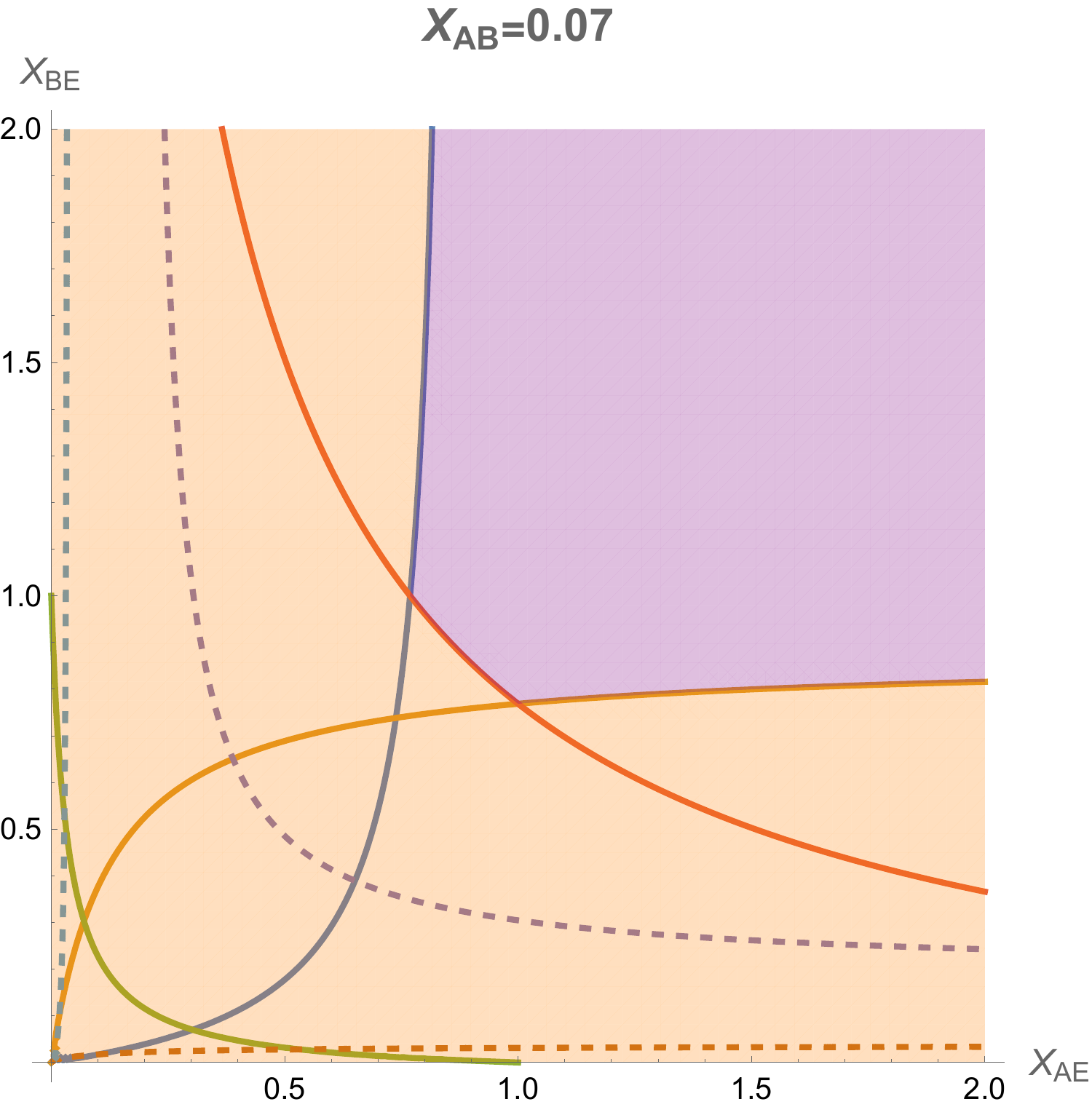}
		\end{minipage}
	}
	\subfigure[]{
		\begin{minipage}[b]{.47\linewidth}
			\centering
			\includegraphics[scale=0.22]{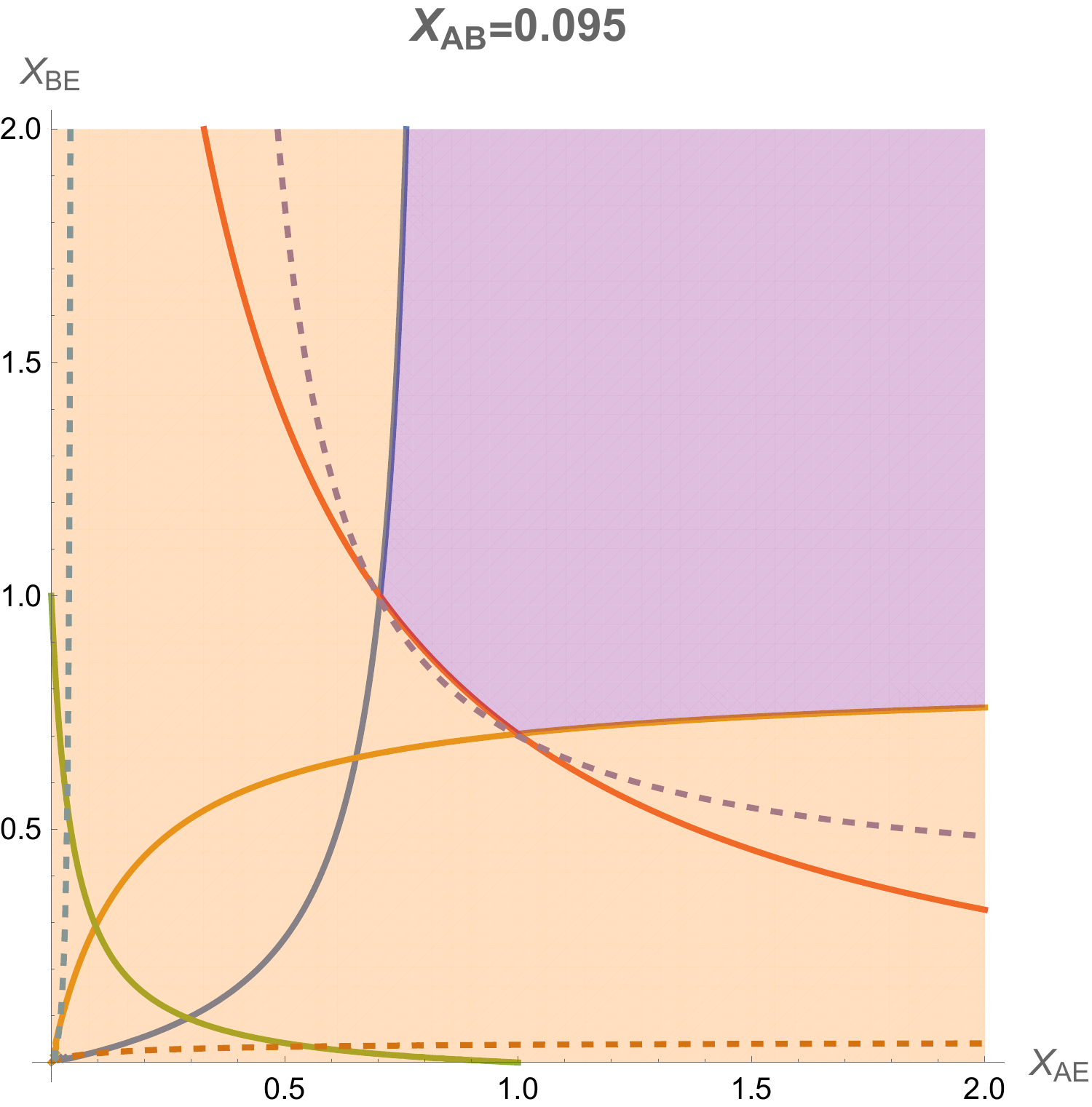}
		\end{minipage}
	}\\
	\subfigure[]{
		\begin{minipage}[b]{.47\linewidth}
			\centering
			\includegraphics[scale=0.22]{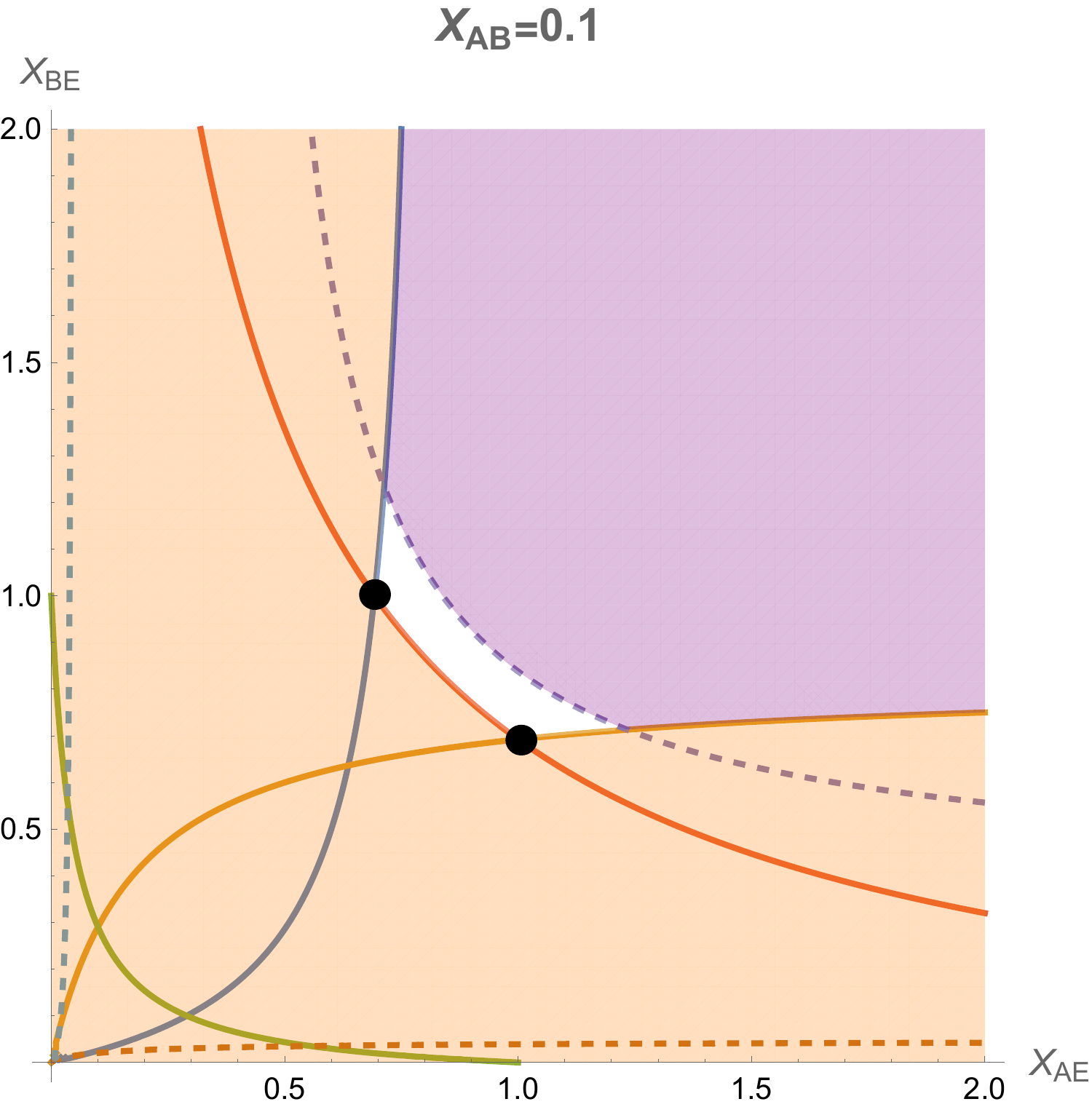}
		\end{minipage}
	}
	\subfigure[]{
		\begin{minipage}[b]{.47\linewidth}
			\centering
			\includegraphics[scale=0.22]{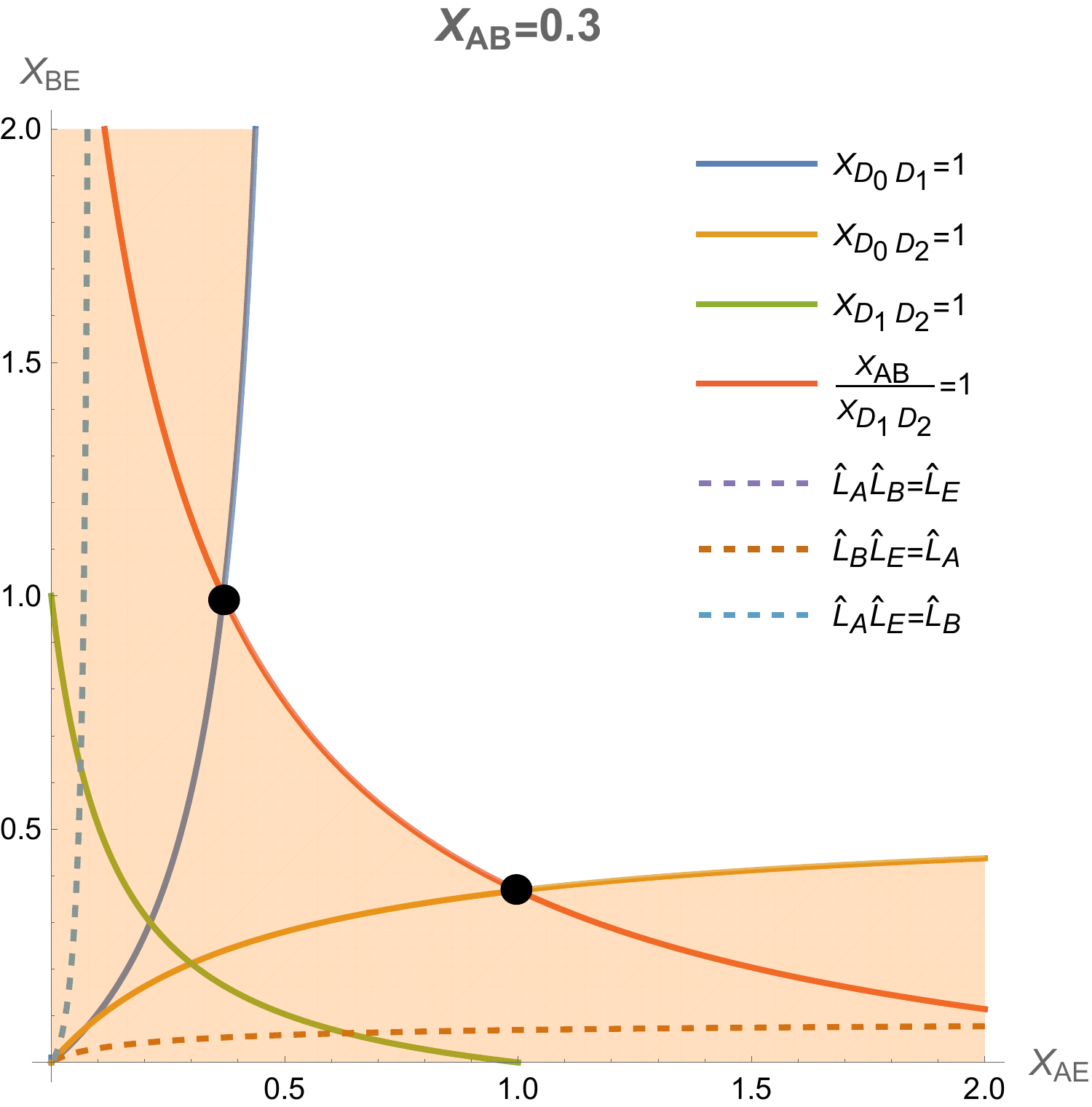}
		\end{minipage}
	}
 \caption{The division of the $(X_{AE}, X_{BE})$ parameter plane into regions where the functional forms of $\widehat{\mathrm{EI}}_{\Delta}(A:B|E)$ are distinct, with four relatively small values of $X_{AB}<1$. The boundaries of the connection conditions and the boundaries of the four possible values of $\widehat{\mathrm{EI}}_{\Delta}(A:B|E)$ are shown in solid and dashed curves respectively.
 The legends of all the 4 sub-figures are placed in the white region of panel (d).
 The orange regions violate at least one of the connection conditions such that $\mathrm{EW}(ABE)$ is not fully connected.
 The purple regions fulfill all the connection conditions as well as the condition $\widehat{L}_E>\widehat{L}_A \widehat{L}_B$, in which $\widehat{\mathrm{EI}}_{\Delta}(A:B|E)$ is 1 in \eqref{posval}.
  $\mathrm{EI}_{\Delta}(A:B|E)$ vanishes in both the orange and purple regions while remaining finite in the white regions in which $\widehat{\mathrm{EI}}_{\Delta}(A:B|E)=(\widehat{L}_A \widehat{L}_B)/\widehat{L}_E$.
 $X_{AB}$ increases across subfigures (a)–(d). 
 Panel (a) shows a representative parameter space of the first phase in which the orange and purple regions cover the whole plane such that $\mathrm{EI}_{\Delta}(A:B|E)$ vanishes in the entire space.
 In panel (b), $X_{AB}$ reaches the first critical point at which the curves $\widehat{L}_E=\widehat{L}_A \widehat{L}_B$(dashed purple), $X_{AB}/X_{D_1 D_2}=1$(solid red) and $X_{D_0 D_1}=1$(solid blue) (or $X_{D_0 D_2}=1$(solid orange)) intersect at a single point.
 Panel (c) and (d) depict the parameter space at selected $X_{AB}$ in the second phase where there exist regions(white) where $\mathrm{EI}_{\Delta}(A:B|E)$ does not vanish.
 The black dots mark the locations at which $\mathrm{EI}_{\Delta}(A:B|E)$ attains its maximum in each of these two sub-figures, namely the intersection points of the solid red and solid blue (or solid orange) curves.
 Since $\mathrm{\widehat{EI}}_{\Delta}(A:B|E)\neq 1$ for any {totally connected configuration of $E$} when $X_{AB}>0.2071$, the purple region disappears in (d).}
\label{FigXABCC1}
\end{figure}
    
    \begin{itemize}
    \item {\bf Phase 1: \(0< X_{AB} \le  0.0952\) with   \(\mathrm{Max}(\widehat{\mathrm{EI}}_{\Delta}(A:B|E))=0\) }

    Figures \ref{FigXABCC1}(a) and (b) illustrate that when \(X_{AB}\) is below a certain threshold, the two regions where \(\mathrm{EI}_{\Delta}(A:B|E)\) vanishes cover the entire plane of the \((X_{AE}, X_{BE})\) parameter space. In this figure, the zero region where the entanglement wedge of $ABE$ is disconnected is colored in orange, while the region where \(\mathrm{EI}_{\Delta}(A:B|E)\) vanishes in the fully connected regime is colored in purple. Note that these figures all possess a reflection symmetry along the $X_{AE}=X_{BE} $ line due to the interchange symmetry of $A$ and $B$. 
    {As \(X_{AB}\) increases from zero, the solid red curve (a boundary of the orange region) and dashed purple curve (the boundary of the purple region) keep moving closer and the critical \(X_{AB}\) value occurs when a middle non-vanishing region (white) between these two curves is about to emerge.}

    We now calculate the critical value of \(X_{AB}\) between the first and the second phases. This critical value is determined by the following critical conditions
\[
\begin{aligned}
   \widehat{L}_E = \widehat{L}_A \widehat{L}_B, \quad X_{BE} = 1, \quad X_{D_0 D_1} = 1 \quad
   \text{or} \quad
   \widehat{L}_E = \widehat{L}_A \widehat{L}_B, \quad X_{AE} = 1, \quad X_{D_0 D_2} = 1,
\end{aligned}
\]
which have been simplified using \eqref{idX2}. Upon exchanging \(D_1\), \(D_2\) and also \(A\), \(B\), the two sets of conditions switch roles due to the interchange symmetry of $A$ and $B$. For simplicity, we can focus on just one set of conditions without loss of generality. Utilizing \eqref{idX1}, the first set of the critical conditions simplifies to
\begin{equation}\label{ConD1}
\begin{aligned}
&\sqrt{\overline{X}_{D_1 D_2}} + \sqrt{1 + \overline{X}_{D_1 D_2}} = \left(\sqrt{\overline{X}_{D_0 D_1}} + \sqrt{1+\overline{X}_{D_0 D_1}}\right)\left(\sqrt{\overline{X}_{D_0 D_2}} + \sqrt{1+\overline{X}_{D_0 D_2}}\right),\\
&2\overline{X}_{D_0 D_1} + 1 + \overline{X}_{D_0 D_2} + \overline{X}_{D_1 D_2} - \overline{X}_{D_0 D_2} \overline{X}_{D_1 D_2} = 0, \quad \overline{X}_{D_0 D_1} = 1.
\end{aligned}    
\end{equation}
Note that we have expressed these conditions in terms of the reciprocals of cross ratios for conciseness of the equations. The only valid solution is
\begin{equation}\label{sol1}
\begin{aligned}
\overline{X}_{D_0 D_1} =& 1, \quad \overline{X}_{D_1 D_2} = \frac{1}{2} (5 + 4 \sqrt{3} + \sqrt{41 + 24 \sqrt{3}}) \approx 10.5075,\\
\overline{X}_{D_0 D_2} =& - 2 \sqrt{62 + 36 \sqrt{3} + 6 \sqrt{41 + 24 \sqrt{3}} + 4 \sqrt{123 + 72 \sqrt{3}}}\\&+\frac{17}{2} + 6 \sqrt{3} + \frac{3}{2} \sqrt{41 + 24 \sqrt{3}} \\
&\approx 1.4207.
\end{aligned}    
\end{equation}
Using \eqref{idX2}, we finally find that at the critical point \(X_{AB} = X_{D_1 D_2} \approx \frac{1}{10.5075} \approx 0.0952\) and \(X_{AE} = X_{D_0 D_2} \approx \frac{1}{1.4207} \approx 0.7039\), with \(X_{BE}=1\) already stated in the critical conditions.

As a result, the combination \(\mathrm{EI}_{\Delta}(A:B|E)\) vanishes for any choice of the interval \(E\) when \(0 < X_{AB} < 0.0952\). Notably, at the critical point the configuration \((X_{AE},X_{BE})=(0.7039, 1)\) occurs in a non-symmetric configuration of \(E\) (\ie, \(X_{AE} \neq X_{BE}\)) and due to the reflection symmetry, the {mirror} point \((X_{AE},X_{BE})=(1,0.7039)\) is also a critical configuration.
\end{itemize}

For \(X_{AB} > 0.0952\), there starts to exist one region in the \((X_{AE}, X_{BE})\) plane where the combination \(\mathrm{EI}_{\Delta}(A:B|E)\) does not vanish {and \(\widehat{\mathrm{EI}}_{\Delta}(A:B|E)\) equals \(\frac{\widehat{L}_A \widehat{L}_B}{\widehat{L}_E}\)} (the white region in Figures \ref{FigXABCC1}(c) and (d)). As \(X_{AB}\) continues to increase, {its maximum will transition through several phases as demonstrated below}. Due to the interchange symmetry of $A$ and $B$, there are always two maximum points {which are depicted as black dots in Figures \ref{FigXABCC1} and \ref{FigXABCC2}}. In the following discussion, we will focus on just one of these points \textemdash {the one that satisfies \(X_{BE}>X_{AE}\)}\textemdash as \(X_{AB}\) varies.

\begin{itemize}
\item {\bf Phase 2: \(0.0952 <X_{AB} \le 0.7039 \)}

\begin{figure}[h]
\centering
	\subfigure[]{
		\begin{minipage}[b]{.47\linewidth}
			\centering
			\includegraphics[scale=0.22]{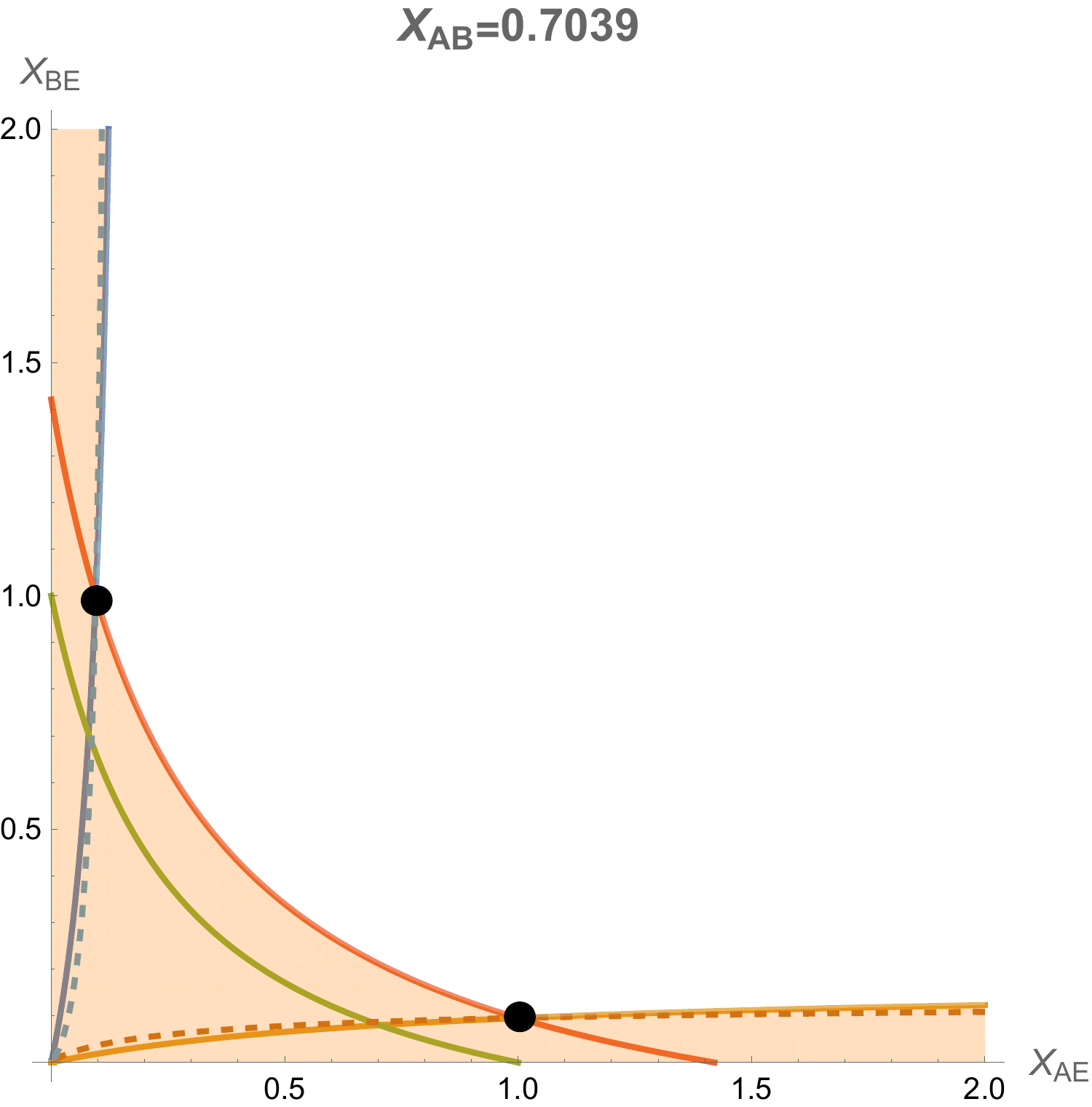}
		\end{minipage}
	}
	\subfigure[]{
		\begin{minipage}[b]{.47\linewidth}
			\centering
			\includegraphics[scale=0.22]{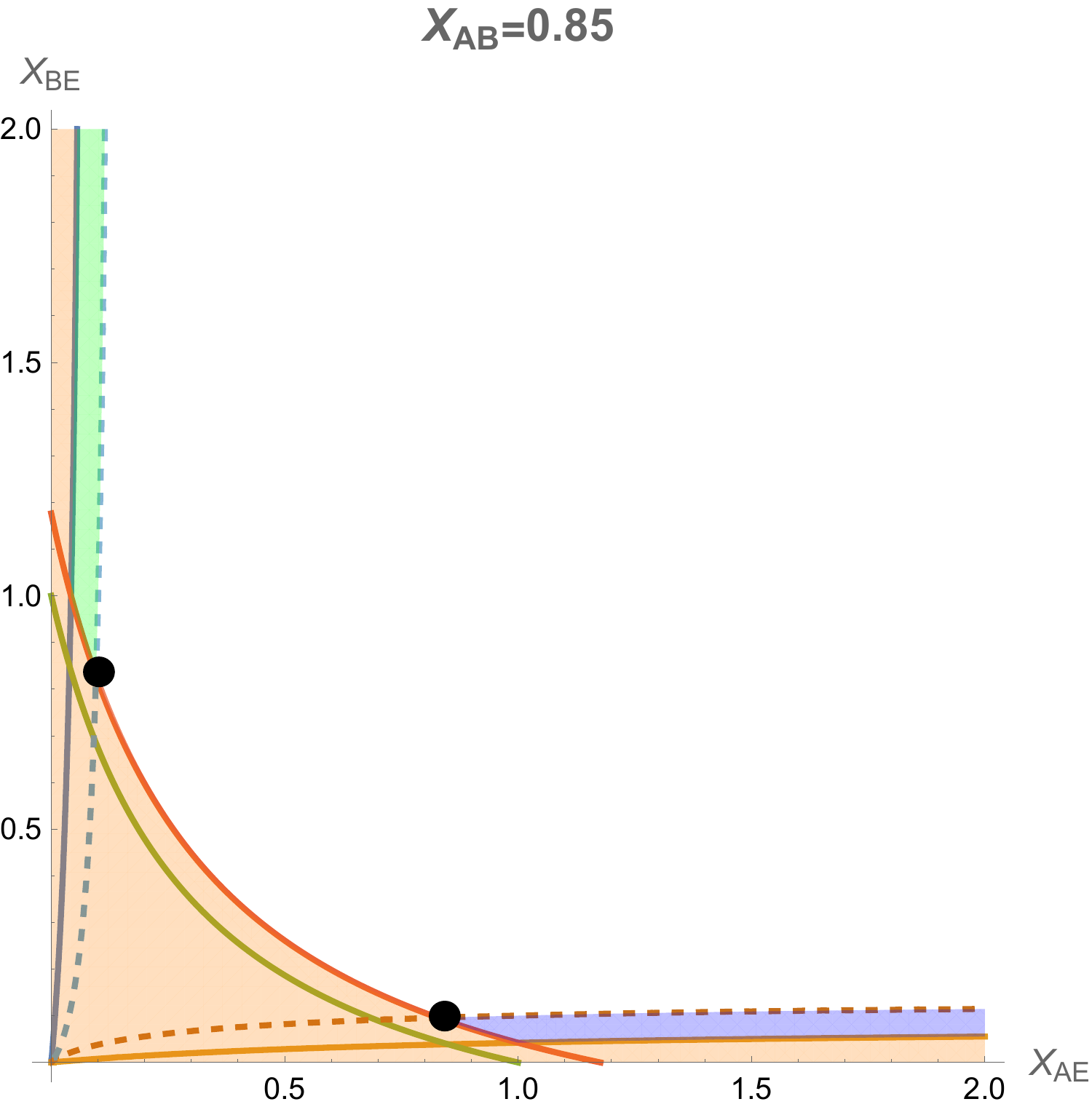}
		\end{minipage}
	}\\
	\subfigure[]{
		\begin{minipage}[b]{.47\linewidth}
			\centering
			\includegraphics[scale=0.22]{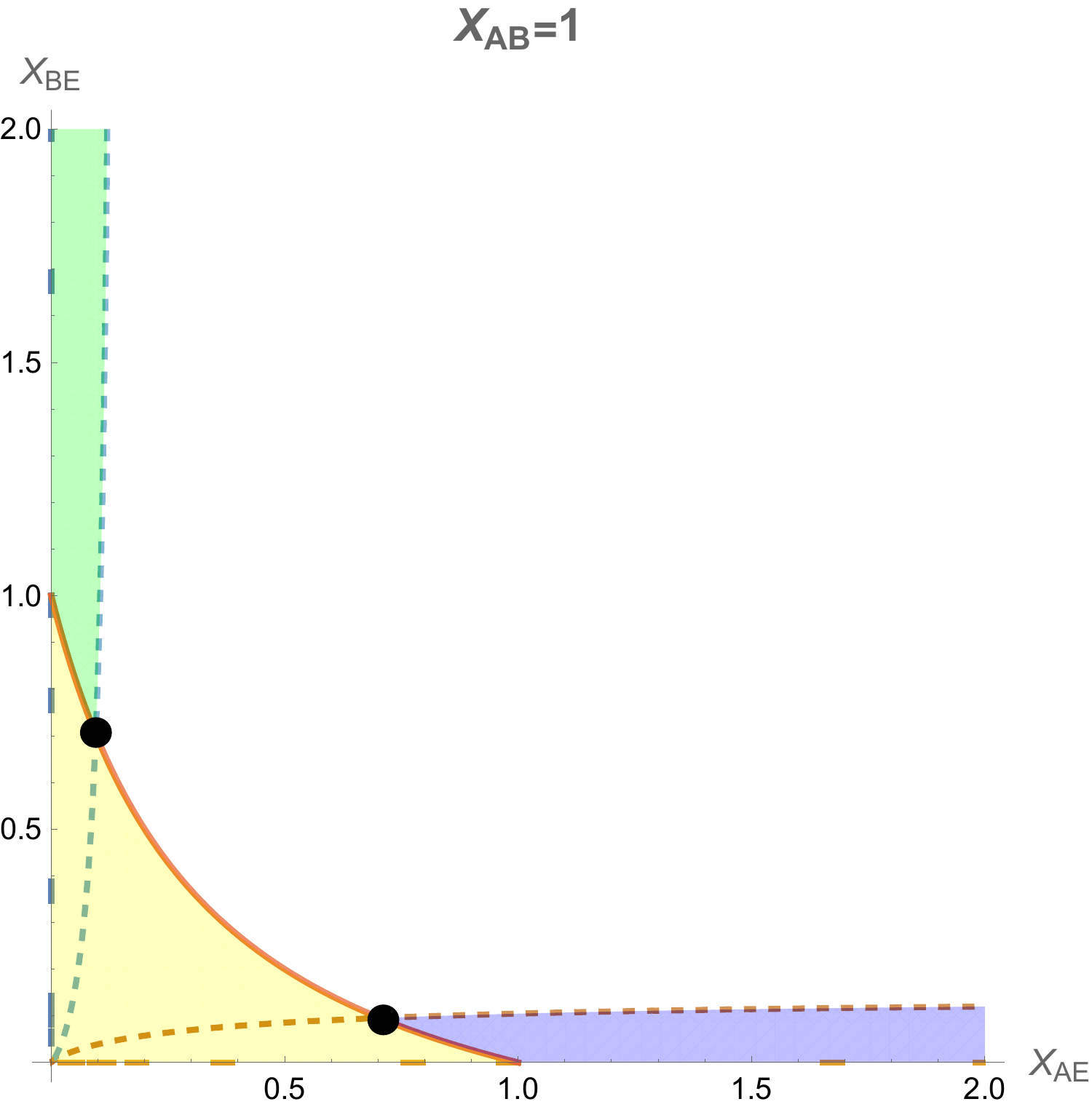}
		\end{minipage}
	}
	\subfigure[]{
		\begin{minipage}[b]{.47\linewidth}
			\centering
			\includegraphics[scale=0.22]{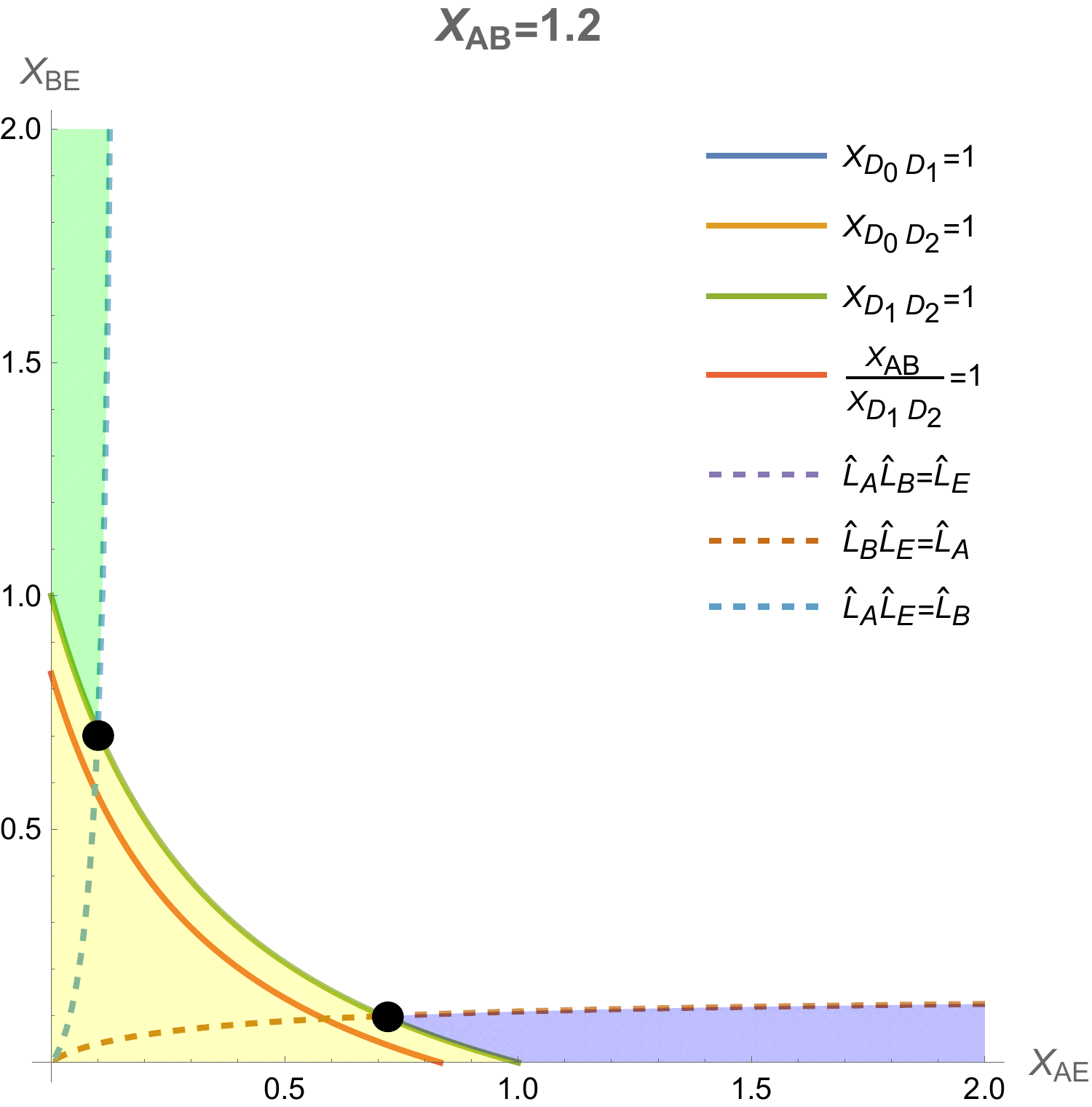}
		\end{minipage}
	}
 \caption{The division of the $(X_{AE},X_{BE})$ parameter plane into regions where the functional forms of $\mathrm{EI}_{\Delta}(A:B|E)$ are distinct, with four relatively large values of $X_{AB}$. 
 The boundaries of the connection conditions and the boundaries of the four possible values of $\widehat{\mathrm{EI}}_{\Delta}(A:B|E)$ are shown in solid and dashed curves respectively.
 The legends of all the four sub-figures are placed in the white region of panel (d). 
 Different shaded regions indicate different functional forms of $\widehat{\mathrm{EI}}_{\Delta}(A:B|E)$.
 The white, green and blue regions correspond to $(\widehat{L}_A \widehat{L}_B)/\widehat{L}_E$, $\widehat{L}_A^2$ and $\widehat{L}_B^2$ respectively.
 In the orange regions, $\mathrm{EI}_{\Delta}(A:B|E)$ vanishes because $\mathrm{EW}(ABE)$ is not fully connected and $X_{AB}<1$.
 In the yellow regions, $\mathrm{EW}(ABE)$ is also not fully connected  but $\mathrm{EW}(AB)$ is connected such that $\mathrm{EI}_{\Delta}(A:B|E)$ does not vanish.
 The black dots mark the locations at which $\widehat{\mathrm{EI}}_{\Delta}(A:B|E)$ attains its maximum in each of the sub-figures.
 $X_{AB}$ increases across subfigures (a)–(d), with each panel showing either a critical point or a phase between critical points.
 (a) The second critical point at which the curves $\widehat{L}_B=\widehat{L}_A \widehat{L}_E$(dashed blue), $X_{AB}/X_{D_1 D_2}=1$(solid red) and $X_{D_0 D_1}=1$(solid blue) (or $X_{D_0 D_2}=1$(solid orange)) intersect at a single point.
 (b) The third phase in which the maximum points locate at the intersection points of the solid red curve and dashed blue (or dashed brown) curve. 
 (c) The third critical point at which the solid red and solid green curves coincide and their intersection with the dashed blue (or dashed brown) curve determine the maximum points.
 (d) The fourth phase in which the maximum points locate at the intersection points of the solid green curve and dashed blue (or dashed brown) curve.}
\label{FigXABCC2}
\end{figure}

In the second phase, \(\mathrm{EI}_{\Delta}(A:B|E)\) reaches its maximum at the intersection point of the solid red and solid blue curve, indicating the following connection conditions are satisfied simultaneously at that point
\[
X_{BE} = 1, \quad X_{D_0 D_1} = 1,
\]
as shown in panel (c) and (d) of Figure \ref{FigXABCC1}.
At this stage, we find the maximum value of the ratio \eqref{ratioE} is in the region  \(\frac{\widehat{L}_A \widehat{L}_B}{\widehat{L}_E}\). 
As \(X_{AB}\) continues to grow, the maximal point moves toward the curve \(\widehat{L}_B = \widehat{L}_A \widehat{L}_E\) (the dashed blue curve), which corresponds to the transition boundary between \(\frac{\widehat{L}_A \widehat{L}_B}{\widehat{L}_E}\) and \(\widehat{L}_A^2\) (see Figures \ref{FigXABCC1}(c), (d) and \ref{FigXABCC2}(a)). The second phase ends when the maximal point reaches this transition boundary (the curve \(\widehat{L}_B = \widehat{L}_A \widehat{L}_E\)). In other words, the critical point between the second and third phases is determined by:
\begin{equation}\label{2ndcond}
  \widehat{L}_B = \widehat{L}_A \widehat{L}_E, \quad  X_{BE} = 1, \quad X_{D_0 D_1} = 1.
\end{equation}
When these equations are explicitly expressed in terms of the cross ratios, we obtain:
\begin{equation}\label{ConD2}
\begin{aligned}
&\sqrt{\overline{X}_{D_0 D_2}} + \sqrt{1 + \overline{X}_{D_0 D_2}} = \left(\sqrt{\overline{X}_{D_0 D_1}} + \sqrt{1+\overline{X}_{D_0 D_1}}\right)\left(\sqrt{\overline{X}_{D_1 D_2}} + \sqrt{1 + \overline{X}_{D_1 D_2}}\right)\\
&2\overline{X}_{D_0 D_1} + 1 + \overline{X}_{D_0 D_2} + \overline{X}_{D_1 D_2} - \overline{X}_{D_0 D_2} \overline{X}_{D_1 D_2} = 0, \quad \overline{X}_{D_0 D_1} = 1.\\
\end{aligned}    
\end{equation}
Comparing this with the condition \eqref{ConD1} that determines the first critical point, the only difference is an exchange between \(\overline{X}_{D_0 D_2}\) and \(\overline{X}_{D_1 D_2}\). Accordingly, the solutions of the cross ratios are swapped: \(\overline{X}_{D_1 D_2} \approx 1.4207\), \(\overline{X}_{D_0 D_1}=1\) and \(\overline{X}_{D_0 D_2} \approx 10.5075\). Consequently, the second critical value is \(X_{AB} = X_{D_1 D_2} \approx \frac{1}{1.4207} \approx 0.7039\). Therefore, the second phase exists between $0.0952<X_{AB}<0.7039$.
    
{As for the maximum of \(\widehat{\mathrm{EI}}_{\Delta}(A:B|E)\) at the second critical point, it can be deduced from the condition \eqref{2ndcond} and the fact that the second critical point locates at the transition boundary of \(\frac{\widehat{L}_A \widehat{L}_B}{\widehat{L}_E}\) and \(\widehat{L}_A^2\).
Thus, we have
\begin{equation}
    \mathrm{Max}(\widehat{\mathrm{EI}}_{\Delta}(A:B|E))=\widehat{L}_A^2=\left(\sqrt{\overline{X}_{D_0 D_1}}+\sqrt{1+\overline{X}_{D_0 D_1}} \right)^2=3+2\sqrt{2}
\end{equation}
at the second critical point.}    
    
    \item {\bf Phase 3: $0.7039<X_{AB} \le 1$}

After this second critical point, the maximum lies at the intersection point of the dashed blue and solid red curves (the black dots shown in \ref{FigXABCC2}(b)) which represent the conditions:
\begin{equation}\label{ConLXc}
\widehat{L}_B = \widehat{L}_A \widehat{L}_E, \quad \text{and} \quad  X_{BE} = X_{D_0 D_1},
\end{equation}
respectively.
Note that the second condition is equivalent to \(\frac{X_{AB}}{X_{D_1 D_2}} = 1\). These conditions determine the maximum point until \(X_{AB}\) reaches 1 (Figure \ref{FigXABCC2}(c)), which is the critical point where the dominance between the connection conditions \(\frac{X_{AB}}{X_{D_1 D_2}} > 1\) and \(\frac{1}{X_{D_1 D_2}} > 1\) switches. If \(X_{AB} > 1\), the configuration determined by \eqref{ConLXc} becomes disconnected (see Figure \ref{FigXABCC2}(b), (c), and (d)). 
The new maximum point for the fully connected configuration should hinge on the dominant connection condition \(X_{D_1 D_2} = 1\) and the transition curve \(\widehat{L}_B = \widehat{L}_A \widehat{L}_E\). Hence, \(X_{AB} = 1\) is the third critical point. 

{The configuration of \(E\) satisfies the following conditions
\begin{equation}\label{3rdcond}
    \widehat{L}_B = \widehat{L}_A \widehat{L}_E, \quad X_{AB} = 1, \quad X_{D_1 D_2} = 1.
\end{equation}
at the third critical point.
When these equations are explicitly expressed in terms of the cross ratios, we obtain:
\begin{equation}\label{ConD3}
\begin{aligned}
&\sqrt{\overline{X}_{D_0 D_2}} + \sqrt{1 + \overline{X}_{D_0 D_2}} = \left(\sqrt{\overline{X}_{D_0 D_1}} + \sqrt{1+\overline{X}_{D_0 D_1}}\right)\left(\sqrt{\overline{X}_{D_1 D_2}} + \sqrt{1 + \overline{X}_{D_1 D_2}}\right)\\
&2\overline{X}_{D_1 D_2} + 1 + \overline{X}_{D_0 D_1} + \overline{X}_{D_0 D_2} - \overline{X}_{D_0 D_1} \overline{X}_{D_0 D_2} = 0, \quad \overline{X}_{D_1 D_2} = 1.
\end{aligned}    
\end{equation}
Comparing this with the condition \eqref{ConD2} that determines the second critical point, the only difference is an exchange between \(\overline{X}_{D_0 D_1}\) and \(\overline{X}_{D_1 D_2}\). Accordingly, the solution of the cross ratios are swapped again. 
Now, we have \(\overline{X}_{D_1 D_2} \approx 1\), \(\overline{X}_{D_0 D_1} \approx 1.4207\) and \(\overline{X}_{D_0 D_2} \approx 10.5075\).
The maximum of \(\widehat{\mathrm{EI}}_{\Delta}(A:B|E)\) at the third critical point can then be implied from \eqref{3rdcond}:
\begin{equation}
    \mathrm{Max}(\widehat{\mathrm{EI}}_{\Delta}(A:B|E))=\widehat{L}_A^2=\left(\sqrt{\overline{X}_{D_0 D_1}}+\sqrt{1+\overline{X}_{D_0 D_1}} \right)^2 \approx 7.5504
\end{equation}
for the reason that the third critical point also lies on the transition boundary of \(\frac{\widehat{L}_A \widehat{L}_B}{\widehat{L}_E}\) and \(\widehat{L}_A^2\), but \(\overline{X}_{D_0 D_1}\) differs from its value at the second critical point.}  

\item {\bf Phase 4: $X_{AB}>1$}

For \(X_{AB} > 1\), the EWCS combination no longer vanishes, even if $\mathrm{EW}(ABE)$ is not fully connected. Intervals \(A\) and \(B\) may be connected to each other while remaining separate from interval \(E\). In this configuration, the combination becomes $2\mathrm{EWCS}(A:B)$ and the ratio \eqref{ratioE} is:
\begin{equation}\label{conAB}
\left(\sqrt{X_{AB}}+\sqrt{1+X_{AB}}\right)^2.
\end{equation}
In contrast, for the fully connected configuration, the maximal point satisfies 
\begin{equation}
\widehat{L}_B = \widehat{L}_A \widehat{L}_E, \quad \text{and} \quad  X_{D_1 D_2} = 1,
\end{equation}
and the maximum value of \eqref{ratioE} is:
\[
\widehat{L}_A^2 = \left(\sqrt{\overline{X}_{D_0 D_1}}+\sqrt{1+\overline{X}_{D_0 D_1}}\right)^2
\].
Since \(X_{AB} = \overline{X}_{D_0(D_1 E D_2)}\) and \(X_{D_0(D_1 E D_2)} > X_{D_0 D_1}\), it follows that \(X_{AB} < \overline{X}_{D_0 D_1}\). Thus, the maximal EWCS combination for the fully connected configuration is larger than that of the disconnected configuration. 
\end{itemize}

Figure \ref{MaxEL} illustrates \(\mathrm{Max}(\widehat{\mathrm{EI}}_{\Delta}(A:B|E))\) as a function of \(X_{AB}\) when region \(E\) consists of a single interval. There are three critical points and four distinct phases, where the maximum of the EWCS triangle information vanishes in the first phase, and grows monotonically in the subsequent three phases. 
As $X_{AB}$ approaches infinity, \(\mathrm{Max}(\mathrm{EI}_{\Delta}(A:B|E))\) behaves as $\frac{4 X_{AB}}{(\sqrt{3}-1)(\sqrt{2}+1)-1} \approx 5.2129 X_{AB}$.
\begin{figure}[H]
\centering
	\includegraphics[scale=0.35]{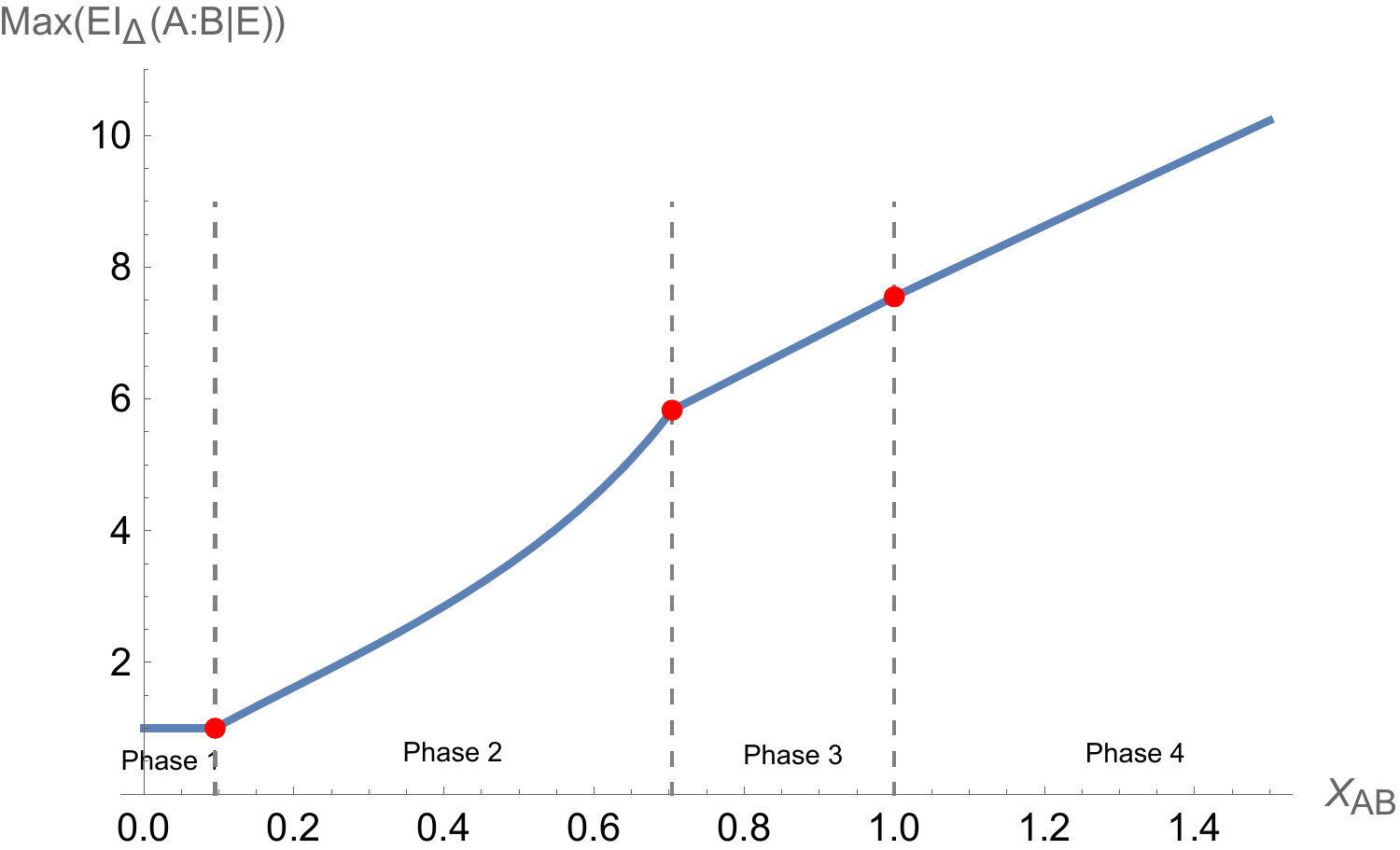}
 \caption{$\mathrm{Max}(\mathrm{EI}_{\Delta}(A:B|E))$ as a function of $X_{AB}$ in the case that region $E$ has only one interval. The phase transition points are denoted as red dots.}
\label{MaxEL}
\end{figure}

\subsection*{The EWCS triangle information and the entanglement of assistance.}

We may now compare \(\mathrm{Max}(\mathrm{EI}_{\Delta}(A:B|E))\) with \(\mathrm{HE}(A:B|E)\), the minimal EWCS that separates A from B in the entanglement wedge of \(ABE\), introduced in Section \ref{sec1}.
On one hand, by definition \(\mathrm{\widehat{HE}}(A:B|E)=\mathrm{Min}(\widehat{L}_A, \widehat{L}_B)\). 
In the above analysis, we mainly focus on the case that \(\widehat{L}_A < \widehat{L}_B\) which is equivalent with \(\widehat{L}_A > \widehat{L}_B\) because of the mirror symmetry between \(A\) and \(B\), so it is convenient to choose \(\mathrm{\widehat{HE}}(A:B|E)=\widehat{L}_A\).
On the other hand, it has been shown above that \(\mathrm{Max}(\mathrm{\widehat{EI}}_{\Delta}(A:B|E))=\widehat{L}_A^2\) in phase 3 and 4.
While in phase 1, \(\mathrm{Max}(\mathrm{\widehat{EI}}_{\Delta}(A:B|E))\) vanishes, and in phase 2, the maximum takes the value of \(\frac{\widehat{L}_A \widehat{L}_B}{\widehat{L}_E}\) and the maximal point lies in a region that \(\widehat{L}_B <\widehat{L}_A \widehat{L}_E\), thus we have \(\mathrm{Max}(\mathrm{\widehat{EI}}_{\Delta}(A:B|E))< \widehat{L}_A^2\) in phase 1 and phase 2.
We therefore conclude that \(\mathrm{Max}(\mathrm{\widehat{EI}}_{\Delta}(A:B|E)) \le \mathrm{\widehat{HE}}(A:B)^2\) and
\begin{equation}
   \mathrm{Max}(\mathrm{EI}_{\Delta}(A:B|E)) \le  2 \mathrm{HE}(A:B|E)
\end{equation}
with the inequality saturated in phase 3 and 4 but not in phase 1 and 2.

\section{Phase structure of $\mathrm{Max}(\mathrm{EI}_{\Delta}(A:B|E))$ with E being n intervals: MPT diagram}\label{sec4}

\noindent {The quantity \((\mathrm{EI}_{\Delta}(A:B|E))\) can be considered as a measure of the entanglement between regions \(A\) and \(B\) with the help of region \(E\)}. As Max\((\mathrm{EI}_{\Delta}(A:B|E))\) should involve a region $E$ that induces the most amount of entanglement between \(A\) and \(B\) facilitated by \(E\), it stands to reason that adding more intervals to \(E\) would strengthen the entanglement. Consequently, the maximum value of \((\mathrm{EI}_{\Delta}(A:B|E))\) could potentially increase as the number \(n\) of intervals in $E$ grows. Thus, in the case where \(E\) has multiple intervals, it becomes intriguing to examine whether there still exists a lower bound for \(X_{AB}\) below which \(\mathrm{EI}_{\Delta}(A:B|E)\) vanishes for any configuration of \(E\). 

In this section, we compute Max\((\mathrm{EI}_{\Delta}(A:B|E))\) for the more general scenario where region \(E\) consists of \(n\) intervals.
{Since there are two gap regions between \(A\) and \(B\), one is ``inside" \(AB\) and the other is ``outside" \(AB\), a general distribution of the \(n\) intervals would be \(n_1\) inside and \(n_2\) outside with \(n_1+n_2=n\).
A simple choice of the distribution is that all the intervals are restricted in one of the gap region.
However, our numeric computation, as will be demonstrated later, suggests that this choice has the minimal value of Max\((\mathrm{EI}_{\Delta}(A:B|E))\) among all possible distributions with fixed total interval number \(n\).
The optimal distribution that maximizes Max\((\mathrm{EI}_{\Delta}(A:B|E))\) turns out to be the most evenly distributed case, \ie, \(n_1=[\frac{n}{2}]\) and \( n_2=n-[\frac{n}{2}]\).
We therefore only explore these two types of interval distributions in this paper:}

\begin{itemize}
    \item Case I: \(n_1=n\) and \( n_2=0\),
    \item Case II: \(n_1=[\frac{n}{2}]\) and \( n_2=n-[\frac{n}{2}]\).
\end{itemize}

We begin in Section \ref{sec4.1} by outlining our numerical method for determining \(\mathrm{Max}(\mathrm{EI}_{\Delta}(A:B|E))\) and identifying the transition conditions that determine the critical points and different phases. 
We developed an informative tool, the generalized multi-phase transition (MPT) diagrams, to indicate the optimal configuration of \(E\) in the case of arbitrary interval number \(n\), in parallel with the upper bounding problem of CMI in \cite{Ju:2024kuc}.
In Section \ref{sec4.2}, we solve the transition conditions and find the values of \(\mathrm{Max}(\mathrm{EI}_{\Delta}(A:B|E))\) at all the critical points.
In Section \ref{secEoA}, we compare \(\mathrm{Max}(\mathrm{EI}_{\Delta}(A:B|E))\) with its upper bound, the entanglement of assistance.
The calculations of the exact values of \(X_{AB}\) at critical points are left to Appendix \ref{sec4.3}.

 \subsection{Max\((\mathrm{EI}_{\Delta}(A:B|E))\) and transition conditions}\label{sec4.1}

\noindent {With \(A\) and \(B\) fixed, the number of degrees of freedom for the choice of $E$ is \(2n\) for general \(n\). While the case \(n=1\) admits analytic solutions, the problem becomes analytically challenging for $n>1$.}
In such cases, it is more practical to use random sampling to approximate the values of Max\((\mathrm{EI}_{\Delta}(A:B|E))\) and identify the corresponding maximal points of free variables. Subsequently, we can examine the conditions that these maximal points satisfy.

\subsubsection{Numerical evaluation}\label{sec4.1.1}

\noindent Here, we outline our random sampling method for a given, fixed \(n\). {We use as free variables the \(2n\) angular coordinates of the endpoints of the intervals \(E_i\) (\(i = 1, 2, 3, \ldots, n\)) on the conformal boundary, instead of cross ratios.}
Unlike cross ratios, angular coordinates are bounded, making them well-suited for random sampling.
{Hence, with \(2n\) angular coordinates spanning the parameter space, the EWCS combination at any point of this space is obtained by evaluating the ratio}
\begin{equation}\label{ratioEn}
    \widehat{\mathrm{EI}}_{\Delta}(A:B|E)=\frac{\widehat{\mathrm{EWCS}}(A:EB) \widehat{\mathrm{EWCS}}(B:EA)}{\widehat{\mathrm{EWCS}}(E:AB)}.
\end{equation}
{Each term in the ratio can be expressed by a set of cross ratios with the relation between geodesic lengths and cross ratios \eqref{EWCSAB}. Then the cross ratios are related to angular coordinates through \eqref{Xac}. }

{To evaluate the ratio \eqref{ratioEn}, we have to specify which configuration of $\mathrm{EW}(ABE)$ should be considered and identify all possible values of each term in that configuration.}
First, we found that it suffices to focus on the case where $\mathrm{EW}(ABE)$ is fully connected.
Otherwise, if there is one interval \(E_j\) that is not connected in $\mathrm{EW}(ABE)$, the problem is then reduced to the case of \(n-1\) intervals.
Thus, we need to screen out the {region in the parameter space that does not correspond to fully connected entanglement wedges.}  All the connection conditions for any \(n\) can be systematically obtained as introduced in Section \ref{sec2.2}.

Second, we need to identify all possible values of \(\widehat{\mathrm{EWCS}}(A:EB)\), \(\widehat{\mathrm{EWCS}}(B:EA)\), and \(\widehat{\mathrm{EWCS}}(E:AB)\). In $\mathrm{EW}(ABE)$, \(\mathrm{EWCS}(A:EB)\) should separate region \(A\) from all intervals belonging to \(B\) and \(E\), as illustrated in Figure \ref{EWCS}(a). Consequently, there must be two endpoints for the candidate geodesics that lie on the RT surface of the two gap regions adjacent to \(A\), specifically \(D_1\) and \(D_{n+2}\).
Suppose we begin at the endpoint on the RT surface of the gap region \(D_1\) and then move along a geodesic to the RT surface of another gap region. Starting again from this second gap region along a new geodesic, we may proceed to a third gap region, and so forth. The path must proceed forward, meaning that if we have already moved to \(E_i\), the next destination cannot be \(E_j\) where \(j<i\), as this would result in a set of self-intersecting geodesics. This process ends when we arrive at the gap region between \(E_n\) and \(A\), namely \(D_{n+2}\). There are \(n\) potential gap regions we can reach in a candidate of \(\mathrm{EWCS}(A:EB)\). Each gap region can either be reached or bypassed, rendering \(\mathrm{EWCS}(A:EB)\) with \(2^n\) potential configurations. Similarly, \(\mathrm{EWCS}(B:EA)\) has \(2^n\) potential configurations, starting from one gap region adjacent to \(B\) and ending at the other.

In contrast, we do not need to survey all \(\mathrm{EWCS}(E:AB)\) configurations; only the blue geodesics depicted in Figure \ref{EWCS}(b) need to be considered.
This holds if we assume that the EWCS combination with \(E\) comprising \(n+1\) intervals exceeds that of \(n\) intervals, which aligns with our expectation. A diagrammatic demonstration is provided in Figure \ref{EWCS}(b).
In this figure, the blue, green and purple geodesics represent three different kind of candidates of EWCS$(E:AB)$. The green geodesics can be divided into two parts, one is a candidate of EWCS$(A:EB)$ and the other is a candidate of EWCS$(B:EA)$.
Because of the division, the value of $\widehat{\mathrm{EWCS}}(E:AB)$ is factorized into a candidate of $\widehat{\mathrm{EWCS}}(A:EB)$ (denoted as $\widehat{L}_A$) and a candidate of $\widehat{\mathrm{EWCS}}(B:EA)$ (denoted as $\widehat{L}_B$).
 On the contrary, such division is impossible for the blue and purple geodesics.
 The blue geodesics connect all the gap regions while the purple ones link a part of the gap regions.
 {If the green geodesics had the smallest total length among all the candidates of EWCS$(E:AB)$, its two factors $\widehat{L}_A$ and $\widehat{L}_B$ must be the smallest among all the candidates of $\widehat{\mathrm{EWCS}}(A:EB)$ and $\widehat{\mathrm{EWCS}}(B:EA)$ respectively.
 Thus we have $\widehat{\mathrm{EWCS}}(A:EB)=\widehat{L}_A$, $\widehat{\mathrm{EWCS}}(B:EA)=\widehat{L}_B$, $\widehat{\mathrm{EWCS}}(E:AB)=\widehat{L}_A \widehat{L}_B$ and the ratio \eqref{ratioEn} does not exceed 1, which is not the maximal value.}
 {If the purple geodesics had the smallest total length, this means the inclusion of more gap regions (and intervals) between the gap regions that the purple geodesics are connecting to does not influence the maximal value.
 However, this is in conflict with our expectation that the EWCS combination in the case where $E$ has $n+1$ intervals is larger than that in the case of $n$ intervals}.
 Therefore, the purple geodesics can also be excluded.
\begin{figure}[H]
\centering
	\subfigure[]{
		\begin{minipage}[b]{.47\linewidth}
			\centering
			\includegraphics[scale=0.7]{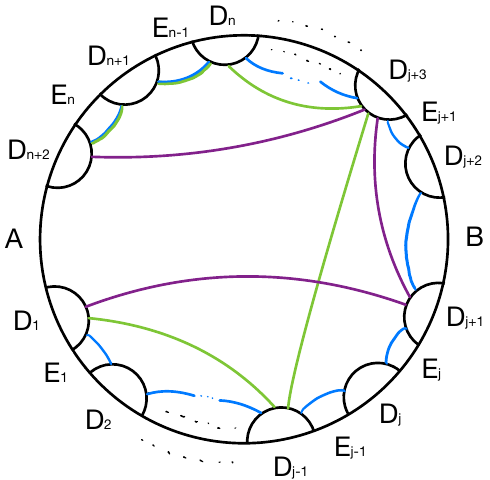}
		\end{minipage}
	}
	\subfigure[]{
		\begin{minipage}[b]{.47\linewidth}
			\centering
			\includegraphics[scale=0.7]{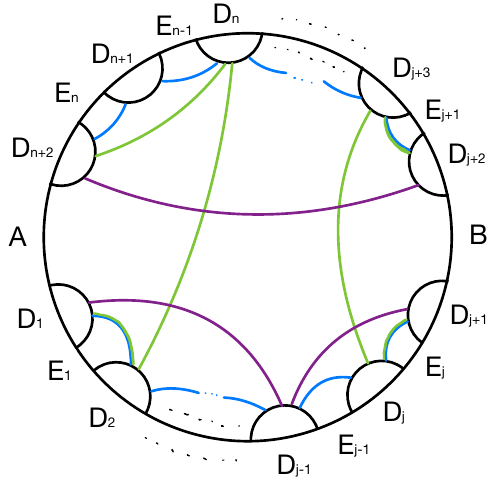}
		\end{minipage}
	}
 \caption{(a) The blue, green and purple geodesics are three candidates of EWCS$(A:EB)$. 
 (b) The blue, green and purple geodesics are three candidates of EWCS$(E:AB)$. Only the blue geodesics need to be considered in the evaluation of EWCS$(E:AB)$.
 }
\label{EWCS}
\end{figure}

Using numerical sampling of random points, we follow a two-step process to approximate the maximum of \(\widehat{\mathrm{EI}}_{\Delta}(A:B|E)\) \eqref{ratioEn}. First, a global sampling is conducted, generating a large number of random points—say, 100,000—in the gap region(s) between  \(A\) and \(B\) (one gap region for Case I, two for Case II). Points that are not fully connected are then filtered out. Evaluating equation \eqref{ratioEn} at the valid points yields an approximate maximum value and the corresponding maximal point.

Next, to achieve higher precision, local sampling is carried out near the maximal point of global sampling. The newly generated points are again screened by the connection conditions and then evaluated by equation \eqref{ratioEn} to find the maximum. Multiple rounds of local sampling may yield a more accurate maximum value and point.

Our numerical computations reveal that maximal points always satisfy several transition conditions simultaneously. {This is parallel to the case of CMI as introduced in Section \ref{secubCMI} and reflects the underlying similarity in their mathematical structures.}
These conditions fall into two categories: 
\begin{itemize}
    \item Connection conditions, which are satisfied by cross ratios (\(X\)s) and indicate that the maximal point lies on the boundary of the connected phase of $\mathrm{EW}(ABE)$.
    \item EWCS conditions, satisfied by \(\widehat{L}\)s, which concern certain geodesic lengths between gap regions.
\end{itemize}
{The difference with the case of CMI is the appearance of the EWCS conditions in addition to the connection conditions.}
Notably, the total number of these transition conditions equals the number of degrees of freedom (\(2n\) if \(E\) involves \(n\) intervals and \(X_{AB}\) is fixed). Thus, the maximum of \(\widehat{\mathrm{EI}}_{\Delta}(A:B|E)\) is determined by these transition conditions.

Moreover, these transition conditions depend on the {value} of \(X_{AB}\) and change at certain critical values, dividing the entire range of \(X_{AB} > 0\) into phases based on the set of transition conditions that the maximal point satisfies. 
For example, suppose phase \(\alpha\) of \(\mathrm{Max}(\widehat{\mathrm{EI}}_{\Delta}(A:B|E))\) occurs when \(X_{AB} \in (a,b)\) with conditions denoted by \(\alpha_1, \alpha_2, \ldots, \alpha_{2n}\), and phase \(\beta\) occurs when \(X_{AB} \in (b,c)\) with conditions \(\beta_1, \beta_2, \ldots, \beta_{2n}\). In this case, we observe
\begin{equation}\label{transition}
\begin{aligned}
    &\alpha_i = \beta_i, \quad i = 1, 2, \ldots, 2n-1,\\
    &\alpha_{2n} \neq \beta_{2n},
\end{aligned}
\end{equation}
{where the equality (inequality) sign means that the two conditions on the two sides are (not) the same.}
At the critical point \(y\), both sets of conditions \(\alpha_i\) (\(i = 1, 2, \ldots, 2n\)) and \(\beta_i\) (\(i = 1, 2, \ldots, 2n\)) are satisfied, providing \(2n+1\) independent conditions to determine the \(2n+1\) variables, which include \(X_{AB}\) and \(2n\) free degrees of freedom of region \(E\).

Through the numerical procedures described above, it is possible to identify all phases and associated transition conditions. 
Note that the transition conditions for a phase can be read off from those of its two endpoints—the two neighboring critical points that bound it—via \eqref{transition}.
Hence it is important to identify the critical points at which the transition conditions that determine \(\mathrm{Max}(\widehat{\mathrm{EI}}_{\Delta}(A:B|E))\) are changed.
We therefore focus on critical points and
we will introduce an illustration tool, the multi-phase transition (MPT) diagram, to indicate the configuration of $E$ with \(\mathrm{Max}(\widehat{\mathrm{EI}}_{\Delta}(A:B|E))\) at the critical points in the following subsections.

\subsubsection{Review of MPT diagrams for CMI upper bound configuration}\label{MPT}

\noindent As introduced in Section \ref{secCMI}, the MPT diagram is a highly informative tool for describing the features of the configuration with maximal CMI. It directly illustrates the (dis)connectivity of entanglement wedges, as well as the phase transition conditions that this configuration must satisfy\footnote{Note that the phase transitions in this subsection, especially in the definition of the multi-phase transition (MPT) diagram, refer to the phase transitions of RT surfaces, namely, the change of the (dis)connectivity of some entanglement wedges. On the other hand, the phases and critical points discussed in Section \ref{sec4.1} correspond to different structures of the value of \((\widehat{\mathrm{EI}}_{\Delta}(A:B|E))\) on the parameter plane.}. Thus, it is of great importance in characterizing the distinct properties of $E$ and the maximal \((\widehat{\mathrm{EI}}_{\Delta}(A:B|E))\) at the critical points that we are interested in. In this subsection, we take the upper bound of CMI as an example, and introduce the general rules of drawing the MPT diagrams and how they encode the essential conditions.

To begin with, as shown in Figure \ref{MPT2}, each diagram consists of three basic elements: circles, dots and legs. We recall that when upper bounding CMI, we are fixing $A$ and $B$ while fine-tuning $E$, which consists of $m$ intervals. The precise position of $E$ is determined by $2m$ phase transition conditions. Depicting this scenario in the diagrams, a circle denotes a single boundary interval as labeled. A dot represents a critical point, which imposes the RT surface phase transition condition between two boundary regions. Specifically, these are the respective unions of boundary intervals that can be traced back to the dot through either of the two legs descending from the dot. 

Based on these elements, there also exist fundamental rules of drawing MPT diagrams. First, uniquely specifying $E$ requires $2m$ dots in total, and each interval of $E$ should be traced back to at least two dots, as the interval has two endpoints. Second, any dot must descend by legs to at most one dot below. This prevents the conditions from overdetermining the system. In particular, a dot connecting directly with two lower dots can impose a constraint inconsistent with those imposed by these dots. Finally, each dot must connect with either $A$ or $B$. As explained in \cite{Ju:2024kuc}, diagrams involving phase transition conditions purely within $E$ cannot maximize the CMI. When intervals $E_1$ and $E_2$ undergo a phase transition, merging them into a single interval by removing the separating gap region leaves the CMI unchanged.

\begin{figure}[h]
\centering
     \includegraphics[width=14cm]{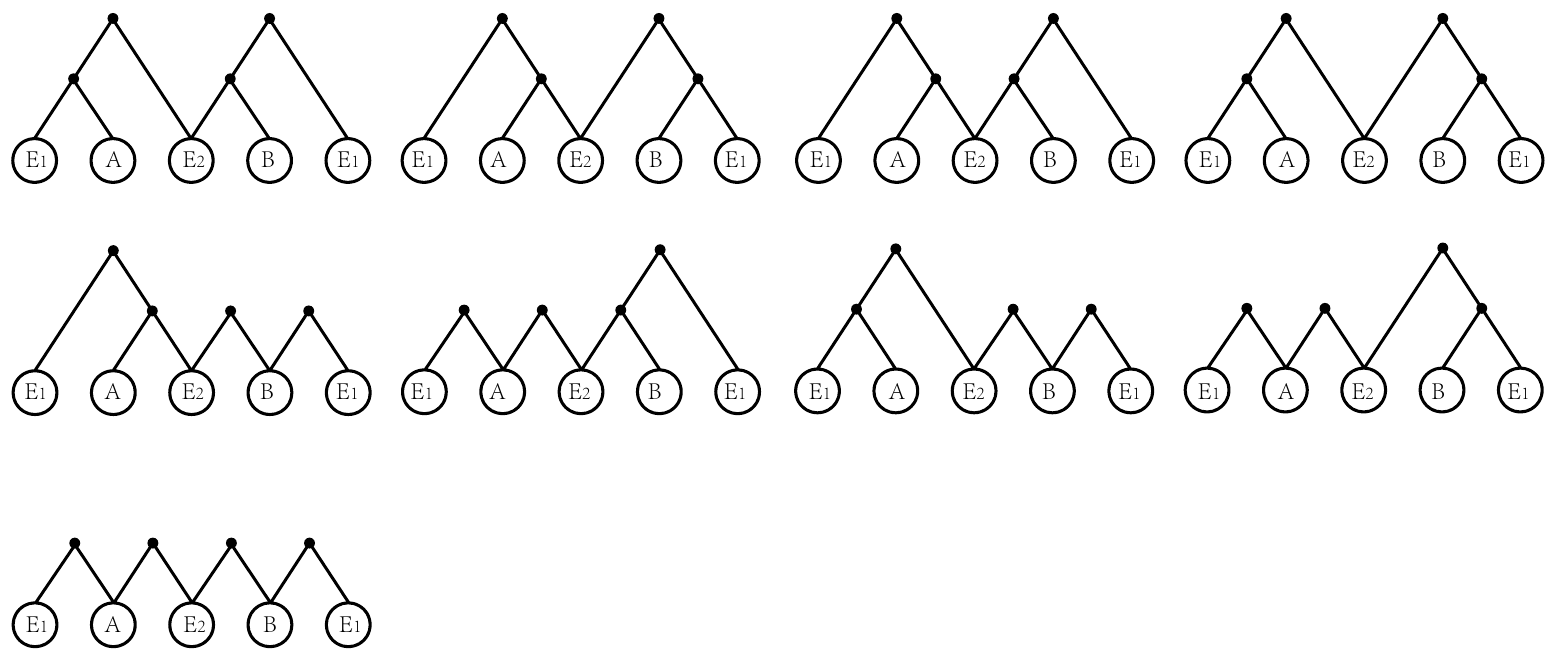}
 \caption{All $m=2$ MPT diagrams for $E$ living in both gap regions inside and outside of $A$ and $B$ that satisfy the three fundamental rules. Their respective CMI values are calculated and compared to obtain the diagram that maximizes CMI.}
\label{MPT2}
\end{figure}

The above rules provide the necessary framework to generate all MPT diagrams with $2m$ phase transition conditions. We then need to compare their respective CMI in search of the diagram with the largest value. 

Through a detailed analysis of the diagrams, we first obtained three additional constraints for the MPT diagram with maximal CMI, which greatly simplify the search, as introduced in Section \ref{secCMI}. We proved them in \cite{Ju:2024kuc} for the general case of arbitrary $m$.
First, $\mathrm{EW}(ABE)$ must be totally connected, since removing any interval $E_i$ which disconnects from $\mathrm{EW}(ABE)$ does not change the CMI. Second, the $\mathrm{EW}(E)$ within a gap region between $A$ and $B$ should be totally disconnected. As in the proof of the third fundamental rule, merging any connected $E_1$ and $E_2$ preserves the CMI value. Finally, in \cite{Ju:2024kuc} we generalized and proved the disconnectivity condition $I(A:E)=I(B:E)=0$, which imposes a particularly strong restriction on the diagrams. 

With these constraints, we successfully identified two key features of the MPT diagram with maximal CMI. First, the dots are arranged in a “zigzag” pattern. Regarding the phase transition between $A$($B$) and $E$ as the successive inclusion of intervals of $E$ from $\mathrm{EW}(A)$ and $\mathrm{EW}(B)$, "zigzag" refers to the pattern that if a new $E_i$ is included from the left of $A$($B$), the subsequent phase transition condition requires the inclusion of the nearest new $E_j$ on the right of $A$, and vice versa. Second, the order of the phase transition between $E_i$ and $A$ is exactly opposite to that between $E_i$ and $B$. Examples for $m=4$ and $m=5$ are shown in Figure \ref{MPT45}. Therefore, we have finally obtained the complete requirements for the CMI upper bound configuration in the language of MPT diagrams. 

\begin{figure}[h]
\centering
     \includegraphics[width=13cm]{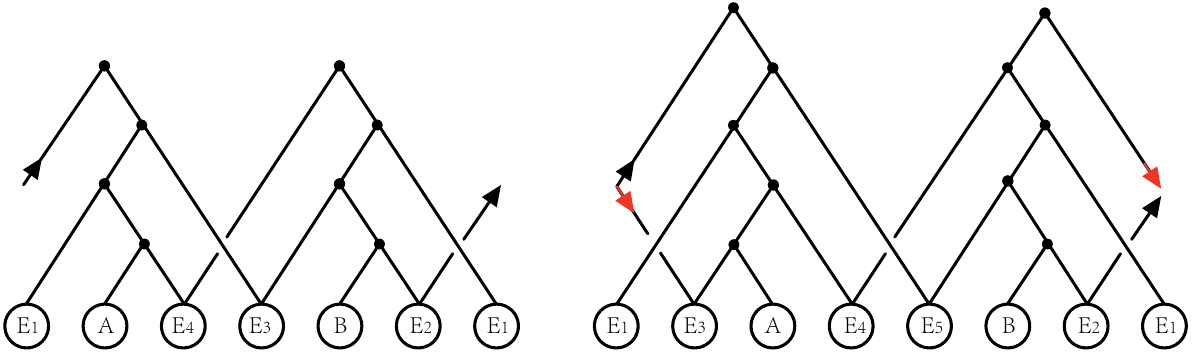}
 \caption{MPT diagrams for $m=4$ and $m=5$ in the two gap region cases that maximize CMI. Legs with arrows of the same color in each diagram are connected.}
\label{MPT45}
\end{figure}

Thus, with the (dis)connectivity of entanglement wedges readily shown in the maximizing diagram, we have derived the general upper bound of CMI (\ref{FinalCMI}) as discussed in Section \ref{secCMI}.

\subsubsection{MPT diagrams for the EWCS triangle information}

\noindent In the case of upper-bounding \((\widehat{\mathrm{EI}}_{\Delta}(A:B|E))\), the MPT diagrams continue to play a crucial rule. Here we systematically analyze the maximal configurations across both case I and case II, covering \(n = 1, 2, 3\) intervals for $E$, and all critical points of \(X_{AB}\). For each case above, the set of transition conditions constraining the configuration with Max\((\widehat{\mathrm{EI}}_{\Delta}(A:B|E))\) can again be represented by an MPT diagram. 
However, it should be noted that, different from the CMI case, now the MPT diagrams must capture both categories of transition conditions \textemdash the connection conditions and the EWCS conditions. The rules of drawing these diagrams should also be modified.

For EWCS triangle information with $n$ intervals in region $E$, there are $n+2$ critical points.
Each critical point is specified by a distinct set of $2n+1$ transition conditions that fix the $2n+1$ degrees of freedom.
These conditions comprise $n+1$ connection conditions and $n$ EWCS conditions.
As shown in Figure \ref{MPTCase}, black dots and legs represent connection conditions while red dots and legs represent EWCS conditions.
Circles also denote boundary intervals as in the case of CMI.

The rules for connection conditions in the MPT diagrams are basically the same with that for the case of CMI, except two differences.
First, the gap regions $D$ are also involved in the diagrams.
Second, each interval of $E$ or $D$ is traced back to at most one dot. 
At the first and second critical points, only intervals of $E$ are involved in the connection conditions.
As the order of critical points increase beyond the second one, there are more and more connection conditions that involve gap regions and fewer and fewer connection conditions related to the intervals of $E$.
Specifically, as the order of critical point increases by one, one interval of $E$ is disconnected with other intervals in the diagram. Instead, a new dot connecting gap regions through legs appears.  
This is reflected in the diagrams as more and more gap regions are connected together by legs and dots, and they form a zigzag pattern.   

The EWCS conditions follow a new set of rules.
In contrast to a connection condition which involves two boundary intervals, an EWCS condition gives a relation among three geodesic lengths, which has the form of
\begin{equation}\label{L=LL}
    \widehat{L}_{D_i D_j} = \widehat{L}_{D_k D_l} \widehat{L}_{D_m D_n},
\end{equation}
where \(L_{D_i D_j}\) denotes the geodesic length between the RT surfaces associated with two specific gap regions \(D_i\) and \(D_j\). The hatted quantity, \(\widehat{L}_{D_i D_j}\), denotes the exponentiated form of this length, as defined in Section \ref{sec3}.
To represent such a three-body relation faithfully in a diagram, a dot should be linked with three legs, each of which is connected to a geodesic between two gap regions.
It is also required to indicate in the diagram which geodesic is on the left-hand-side of \eqref{L=LL} and which two geodesics are on the right-hand-side.
In order to depict both the EWCS conditions and connection conditions of a critical point in a single diagram and make them compatible with each other, we adopt the following rules for representing EWCS conditions:
\begin{itemize}
    \item Basic elements: red dots, red lines, red arrows;
    \item A red line connects either a red dot to a boundary interval or two red dots to each other;
    \item A red line separating two gap regions \(D_i\) and \(D_j\) represents the geodesic length \(\widehat{L}_{D_i D_j}\), see Figure \ref{MPTEWCS} for an example;
    \item Each red dot is  attached to exactly three red lines and encodes an EWCS condition of the form \eqref{L=LL}, which relates the three geodesic lengths represented by those lines;
    \item Among the three red lines attached to a given red dot, one carries an red arrow pointing into the dot (the geodesic length on the left-hand side), and the other two carry red arrows pointing away from the dot (the two geodesic lengths on the right-hand side).
\end{itemize}
\begin{figure}[h]
    \centering
    \includegraphics[scale=1]{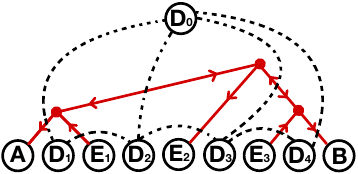}
    \caption{An illustration of the rule of how each red line separates two gap regions and represents a geodesic. Only the EWCS conditions are shown in the figure. \(D_0\) denotes the gap region that is outside \(A\) and \(B\) in Case I. The red line intersecting the dashed black curve that connects the boundary gap regions \(D_i\) and \(D_j\) represents  \(\widehat{L}_{D_i D_j}\). }
    \label{MPTEWCS}
\end{figure}
For instance, as shown in Figure \ref{MPTEWCS}, the MPT diagram for the 1st critical point with \(n=3\) in Case I contains three red dots representing three EWCS conditions. 
The red lines that connecting the leftmost red dot to region \(A\), \(E_1\) and the topmost red dot separate the gap regions \(D_0 D_1\), \(D_1 D_2\) and \(D_0 D_2\) respectively.
Accordingly, the EWCS condition encoded by the leftmost red dot reads \(\widehat{L}_{D_0 D_1}=\widehat{L}_{D_0 D_2} \widehat{L}_{D_1 D_2}\), as indicated by the directions of the red arrows.

Combining the rules for the connection conditions and the EWCS conditions demonstrated above,
the specific MPT diagrams corresponding to all the scenarios we consider are shown in Figure \ref{MPTCase}.

\begin{figure}[h]
\centering
	\subfigure[]{
		\begin{minipage}[b]{1\linewidth}
			\centering
			\includegraphics[scale=0.4]{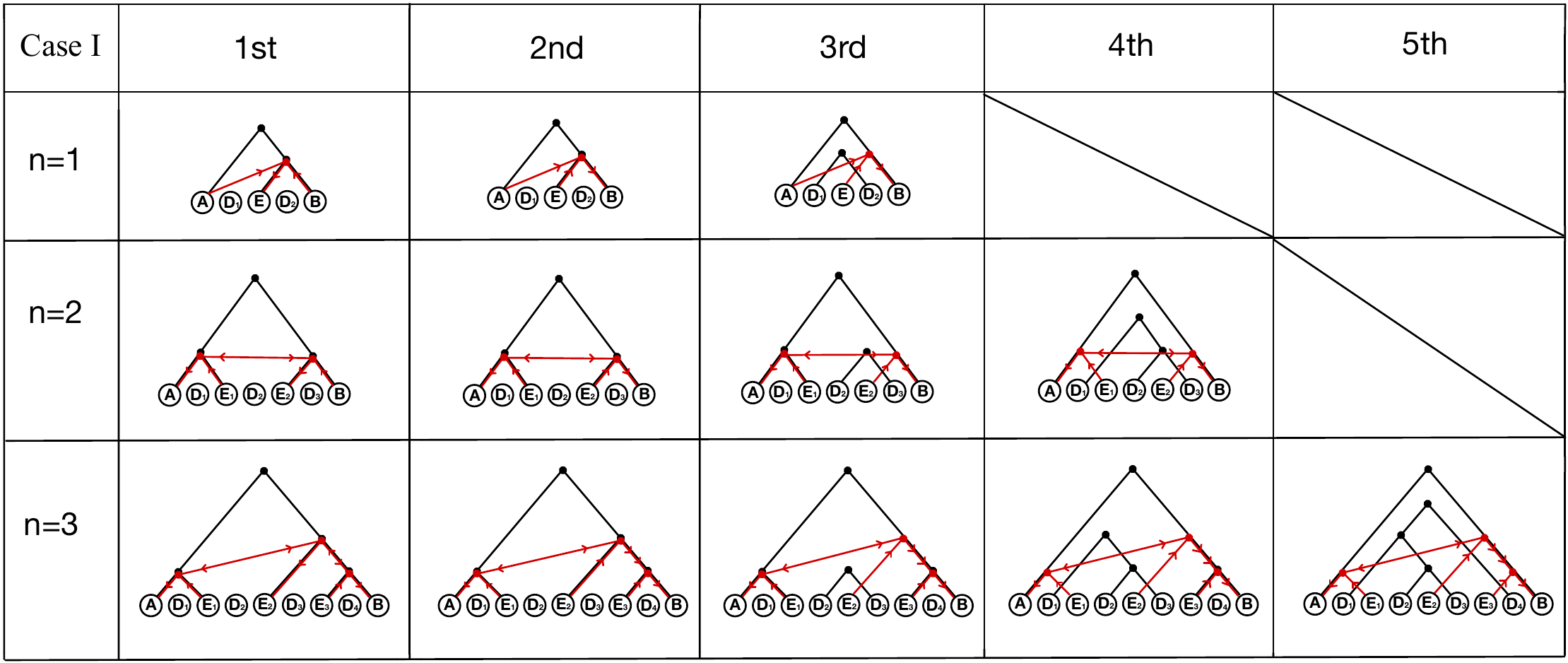}
		\end{minipage}
	}\\
	\subfigure[]{
		\begin{minipage}[b]{1\linewidth}
			\centering
			\includegraphics[scale=0.4]{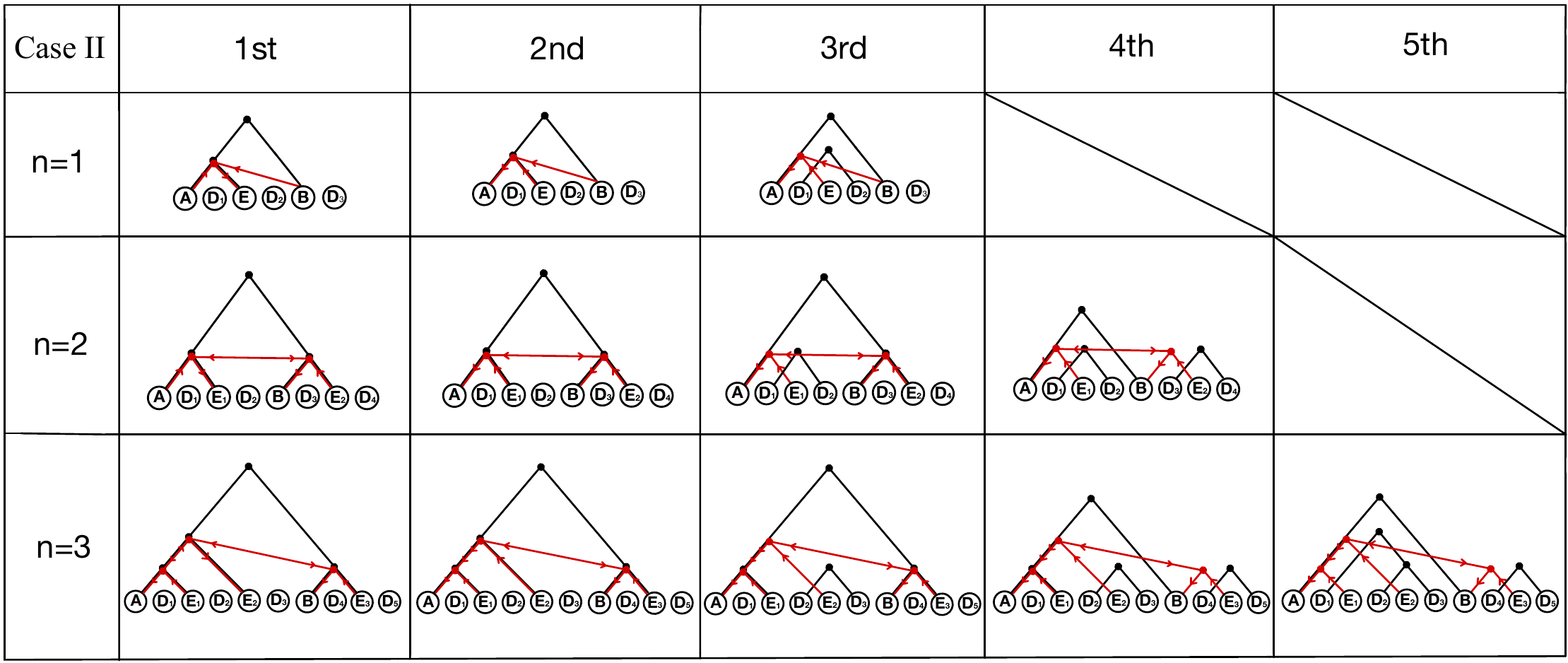}
		\end{minipage}
	}\\
	\subfigure[]{
		\begin{minipage}[b]{1\linewidth}
			\centering
			\includegraphics[scale=0.6]{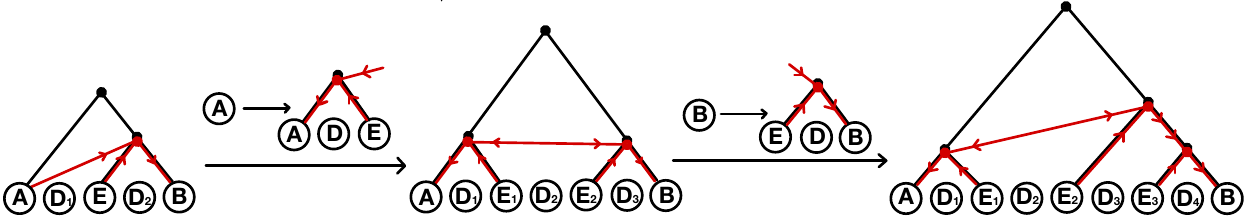}
		\end{minipage}
	}
 \caption{The MPT diagrams of the first five critical points are shown for \(n=1, 2, 3\) respectively: (a) for Case I and (b) for Case II. The black lines and nodes represent connection conditions, while the red lines and nodes indicate EWCS conditions. Panel (c) illustrates the "split" of regions \(A\) and \(B\) when an additional interval of \(E\) is incorporated.
 }
\label{MPTCase}
\end{figure}
The MPT diagrams exhibit some special features. 
First, for any fixed $n$, the EWCS conditions are identical for the second and all subsequent critical points, but they differ from those of the first critical point with one special condition.
Second, the MPT diagrams also have a ``split" property, which allows for the construction of transition conditions for \(n+1\) intervals from those of \(n\) intervals by splitting region \(A\) (or \(B\)) into a smaller \(A\) (or \(B\)), a new gap region, and an additional interval of \(E\). 
For example, at the second critical point, the MPT diagram for \(n=2\) is derived from the diagram for \(n=1\) by replacing region \(A\) with \(A D_1 E_1\) and attaching the transition conditions \(X_{AE_1}=1\) and \(\widehat{L}_{D_0 D_1}=\widehat{L}_{D_0 D_2}\widehat{L}_{D_1 D_2}\), as shown in Figure \ref{MPTCase}(c). 
Similarly, the MPT diagram for \(n=3\) can be obtained by splitting region \(B\) in the \(n=2\) diagram. By iteratively applying this process—alternately splitting \(A\) and \(B\)—one can determine the MPT diagram (\ie, transition conditions) for any arbitrary \(n\). 
As a consequence, the ``split" property of the MPT diagrams gives us the rules to generate critical points and maxima from small interval numbers to general interval numbers and avoid heavy numerical computations.

\subsection{Solving the transition conditions: iteration}\label{sec4.2}

\noindent The rules governing the transition conditions that define critical values of \(X_{AB}\) can be easily discerned from the MPT diagrams, making it straightforward to generalize them to any integer \(n\). The remaining task is to solve the \(2n+1\) conditions to obtain the optimal configuration of \(E_i\) (\(i=1,2,\ldots,n\)) that maximizes \(\widehat{\mathrm{EI}}_{\Delta}(A:B|E)\).
Starting with \(n=1\), configurations corresponding to critical points can be solved iteratively. Since the transition conditions vary between Case I and Case II, we will discuss the iteration processes for these two cases separately.

\subsubsection{Case I}

\noindent It is helpful to divide the fully connected entanglement wedge $\mathrm{EW}(ABE)$ into cells, as illustrated in Figure \ref{CellCaseI}. Each cell is bounded by six geodesics: three are EWCSs, shown in blue, and the other three are RT surfaces, shown in black, of three different gap regions connected by the three EWCSs. 
We define the ``central" gap region, denoted as \(D_c\), as the \(([\frac{n}{2}] + 1)\)-th gap region from left to right, and the ``central" geodesic \(L_0\) as the geodesic connecting the gap region \(D_0\) to the central gap region \(D_c\). \(L_0\) divides the cells to its right from those on its left. Thus, there are \([\frac{n}{2}]\) cells on the left of \(L_0\) and \([\frac{n+1}{2}]\) cells on the right\footnote{Alternatively, in a mirror-symmetrical arrangement, the central gap region \(D_c\) is the \(([\frac{n+1}{2}] + 1)\)-th gap region and there are \([\frac{n+1}{2}]\) cells on the left and \([\frac{n}{2}]\) cells on the right.}. 
The intervals of \(E\) on the left/right are therefore denoted as \(E_{L_i}/E_{R_j}\) (\(i=1, 2, \ldots, [\frac{n}{2}]\), \(j=1, 2, \ldots, [\frac{n+1}{2}]\)). Similarly, the gap regions are denoted as \(D_{L_i}/D_{R_j}\) (\(i=1, 2, \ldots, [\frac{n}{2}]\), \(j=1, 2, \ldots, [\frac{n+1}{2}]\)).
\begin{figure}[H]
 \centering
 \includegraphics[scale=0.7]{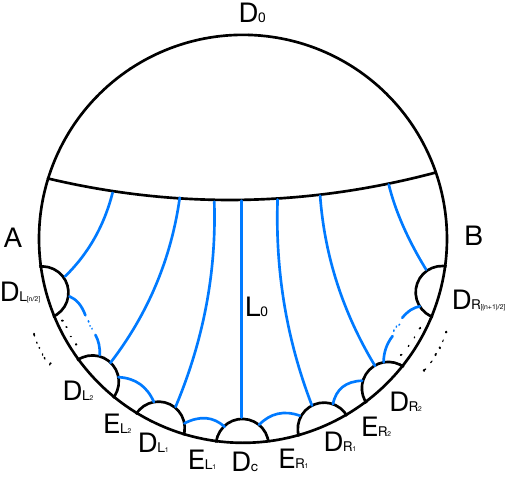}
 \caption{The entanglement wedge of $ABE$ is divided into ``cells" by a series of EWCSs which are shown in blue.
 The ``central" geodesic, labelled as $L_0$, is the boundary between the right cells and left cells.
 The intervals and gap regions that are associated with the right/left cells are labeled with sub-indices $R$/$L$.
 }
\label{CellCaseI}
\end{figure}

The method described above for dividing $\mathrm{EW}(ABE)$ into \(n\) cells is closely linked to the \(2n+1\) transition conditions at critical points. 
According to the ``split" property of MPT diagrams at critical points, the intervals of \(E\) (or cells) are added into $\mathrm{EW}(ABE)$ sequentially as the interval number is increased. 
In particular, the additional cells are first added to the leftmost side of $\mathrm{EW}(ABE)$, followed by the rightmost side, and this process repeats. 
In essence, the \((n+1)\)-th cell does not alter the {optimal} configuration of \(E\) in the first \(n\) cells; instead, it is created by dividing \(A\) or \(B\) into three segments: a new interval \(E_L\) or \(E_R\), a new gap region \(D_L\) or \(D_R\), and a smaller \(A\) or \(B\).
Therefore, a key feature at the critical points is that each addition of a new cell introduces two more transition conditions, while the conditions for the pre-existing cells remain unchanged.

We discuss the critical points of \(\mathrm{Max}(\widehat{\mathrm{EI}}_{\Delta}(A:B|E))\) one by one in the following.
We solve all the transition conditions at each critical point and the results are $2n+1$ cross ratios between gap regions for each critical point, which completely determine the configuration of the whole system \(ABE\).
The evaluation of the critical values of \(X_{AB}\) are shown in Appendix \ref{sec4.3}.
The explicit numerical results for the critical values of \(X_{AB}\) and corresponding values of \(\mathrm{Max}(\widehat{\mathrm{EI}}_{\Delta}(A:B|E))\) are summarized in the tables and figures in Section \ref{sec5}.

\paragraph{The first critical point}

In the case of \(n=1\), assuming the only cell is on the right of \(L_0\), we demonstrated in Section \ref{sec3} that the conditions determining the first critical point (the vanishing critical point) are:
\begin{equation}\label{1stn1}
    \overline{X}_{D_0 D_c}=1, \quad \overline{X}_{BE_{R_1}}=1, \quad \text{and} \quad \widehat{L}_{D_c D_{R_1}}=\widehat{L}_0 \widehat{L}_{D_0 D_{R_1}}.
\end{equation}
As in \eqref{sol1}, the cross ratios related to the first cell, \(\overline{X}_{D_c D_{R_1}}\) and \(\overline{X}_{D_0 D_{R_1}}\), have been found to be \(\overline{X}_{D_c D_{R_1}} \approx 10.5075\) and \(\overline{X}_{D_0 D_{R_1}} \approx 1.4207\). 

For general \(n>1\), the first critical point also corresponds to the upper bound of $X_{AB}$ for which \(\mathrm{Max}(\mathrm{EI}_{\Delta}(A:B|E))\) vanishes\footnote{Note that the EWCS conditions of the first critical point, namely $\widehat{L}_{D_c D_{R_1}}=\widehat{L}_0 \widehat{L}_{D_0 D_{R_1}}$ in \eqref{1stn1}, $\widehat{L}_{D_0 D_{R_{k}}}=\widehat{L}_{D_0 D_{R_{k-1}}} \widehat{L}_{D_{R_{k-1}} D_{R_{k}}}$ in \eqref{1stR}, and $\widehat{L}_{D_0 D_{L_{k}}}=\widehat{L}_{D_0 D_{L_{k-1}}} \widehat{L}_{D_{L_{k-1}} D_{L_{k}}}$ in \eqref{1stL}, indicate that the ratio \eqref{ratioEn} equals \(1\).}. The new transition conditions introduced by the \(k\)-th (\(k>1\)) right cell are:
\begin{equation}\label{1stR}
    \overline{X}_{\left(BD_{R_{[(n+1)/2]}}E_{R_{[(n+1)/2]}}...D_{R_{k+1}}E_{R_{k+1}}\right)E_{R_k}}=1, \quad \widehat{L}_{D_0 D_{R_{k}}}=\widehat{L}_{D_0 D_{R_{k-1}}} \widehat{L}_{D_{R_{k-1}} D_{R_{k}}}.
\end{equation}
Utilizing \eqref{LABE} and \eqref{idX1}, \eqref{1stR} is equivalent to:
\begin{equation}\label{iterEQR}
\begin{aligned}
     &2\overline{X}_{D_0 D_{R_{k-1}}}+1+\overline{X}_{D_0 D_{R_{k}}}+\overline{X}_{D_{R_{k-1}} D_{R_{k}}} - \overline{X}_{D_0 D_{R_{k}}}\overline{X}_{D_{R_{k-1}} D_{R_{k}}}=0,\\
    &\left(\sqrt{\overline{X}_{D_0 D_{R_{k-1}}}}+\sqrt{1+\overline{X}_{D_0 D_{R_{k-1}}}}\right)\left(\sqrt{\overline{X}_{D_{R_{k-1}} D_{R_{k}}}}+\sqrt{1+\overline{X}_{D_{R_{k-1}} D_{R_{k}}}}\right)\\
    &=\sqrt{\overline{X}_{D_0 D_{R_{k}}}}+\sqrt{1+\overline{X}_{D_0 D_{R_{k}}}}.
\end{aligned}
\end{equation}
Therefore, \(\overline{X}_{D_{R_{k-1}} D_{R_{k}}}\) and \(\overline{X}_{D_0 D_{R_{k}}}\) can be determined if \(\overline{X}_{D_0 D_{R_{k-1}}}\) is known. This recurrence relation traces back to the second right cell\footnote{Note that the EWCS condition of the first right cell differs from that of other right cells since the roles of \(\overline{X}_{D_c D_{R_1}}\) and \(\overline{X}_{D_0 D_{R_1}}\) are swapped.}, where \(\overline{X}_{D_{R_{1}} D_{R_{2}}}\) and \(\overline{X}_{D_0 D_{R_{2}}}\) can be expressed in terms of \(\overline{X}_{D_0 D_{R_{1}}}\). Since \(\overline{X}_{D_0 D_{R_{1}}}\) has already been solved in the first right cell using \eqref{1stn1}, all cross ratios related to the right cells can be determined by iterating \eqref{iterEQR}.

The new transition conditions brought by the \(k\)-th left cell are analogous to \eqref{1stR}, with all subscripts ``R" replaced by ``L":
\begin{equation}\label{1stL}
    \overline{X}_{\left(BD_{L_{[n/2]}}E_{L_{[n/2]}}...D_{L_{k+1}}E_{L_{k+1}}\right)E_{L_k}}=1, \quad \widehat{L}_{D_0 D_{L_{k}}}=\widehat{L}_{D_0 D_{L_{k-1}}} \widehat{L}_{D_{L_{k-1}} D_{L_{k}}}.
\end{equation}
Similarly, \eqref{1stL} is equivalent to
\begin{equation}\label{iterEQL}
\begin{aligned}
    &2\overline{X}_{D_0 D_{L_{k-1}}}+1+\overline{X}_{D_0 D_{L_{k}}}+\overline{X}_{D_{L_{k-1}} D_{L_{k}}} - \overline{X}_{D_0 D_{L_{k}}}\overline{X}_{D_{L_{k-1}} D_{L_{k}}}=0,\\
    &\left(\sqrt{\overline{X}_{D_0 D_{L_{k-1}}}}+\sqrt{1+\overline{X}_{D_0 D_{L_{k-1}}}}\right)\left(\sqrt{\overline{X}_{D_{L_{k-1}} D_{L_{k}}}}+\sqrt{1+\overline{X}_{D_{L_{k-1}} D_{L_{k}}}}\right)\\
    &=\sqrt{\overline{X}_{D_0 D_{L_{k}}}}+\sqrt{1+\overline{X}_{D_0 D_{L_{k}}}}.
\end{aligned}
\end{equation}
Likewise, \(\overline{X}_{D_{L_{k-1}} D_{L_{k}}}\) and \(\overline{X}_{D_0 D_{L_{k}}}\) are obtainable if \(\overline{X}_{D_0 D_{L_{k-1}}}\) is known. Since this recursion holds for \(k \ge 1\), all cross ratios pertaining to the left cells can be determined by iterations of \eqref{iterEQL} with the initial condition \(\overline{X}_{D_0 D_{L_{0}}}:=\overline{X}_{D_0 D_c}=1\).

The equations \eqref{iterEQR} and \eqref{iterEQL} have the following mutual form
\begin{equation}\label{iterEQ}
    \begin{aligned}
    &\quad 2\overline{X}_{2j-2}+1+\overline{X}_{2j}+\overline{X}_{2j-1}-\overline{X}_{2j} \overline{X}_{2j-1}=0,\\
    &\sqrt{\overline{X}_{2j}}+\sqrt{1+\overline{X}_{2j}}=\left(\sqrt{\overline{X}_{2j-2}}+\sqrt{1+\overline{X}_{2j-2}}\right)\left(\sqrt{\overline{X}_{2j-1}}+\sqrt{1+\overline{X}_{2j-1}}\right),
    \end{aligned}
\end{equation}
where $\overline{X}_{D_0 D_{R_{k-1}}}$(or $\overline{X}_{D_0 D_{L_{k-1}}}$), $\overline{X}_{D_{R_{k-1}} D_{R_{k}}}$(or $\overline{X}_{D_{L_{k-1}} D_{L_{k}}}$), and $ \overline{X}_{D_0 D_{R_{k}}}$(or $\overline{X}_{D_0 D_{L_{k}}}$) are replaced by $\overline{X}_{2j-2}$, $\overline{X}_{2j-1}$ and $\overline{X}_{2j}$, respectively.
The solution to $\overline{X}_{2k-1}$ and $\overline{X}_{2k}$ is 
\begin{equation}
    \overline{X}_{2j}=f(x), \quad \overline{X}_{2j-1}=g(x),
\end{equation}
where $\overline{X}_{2j-2}$ is abbreviated as $x$ and we have defined
\begin{equation}\label{iteration}
\begin{aligned}
   f(x):=&\frac{1}{2}\left(2+3 x+2 \sqrt{2+6 x+4 x^{2}}+\sqrt{x\left(16+25 x+12 \sqrt{2+6 x+4 x^{2}}\right)}\right),\\
   g(x):=&\left(\sqrt{x(1+x)}-\sqrt{f(x)(1+f(x))}\right)^2-\left(x-f(x) \right)^2.
\end{aligned}
\end{equation}
The initial condition for the right cell is $\overline{X}_{2} =g(1) \approx 1.4207$, while the initial condition for the left cell is $\overline{X}_{0}=1$.
$\overline{X}_{i}$, $(i=1,2,3,...)$ can then be generated from \eqref{iteration} by iteration. 
Restoring the subscripts of cross ratios, for the right cells, we obtain 
\begin{equation}\label{RCR}
\begin{aligned}
   \overline{X}_{D_0 D_{R_{k}}}=\underset{k-1}{\underbrace{f\cdot f \cdot ... \cdot f }}\cdot g(1), \quad k=1,2,3,...\\
   \left\{\begin{matrix}
\overline{X}_{D_{c} D_{R_{1}}}=f(1),\\
\overline{X}_{D_{R_{k-1}} D_{R_{k}}}=g\cdot \underset{k-2}{\underbrace{f\cdot f \cdot ... \cdot f }}\cdot g(1), \quad k=2,3,...
\end{matrix}\right.
\end{aligned}
\end{equation}
For the left cells, we obtain
\begin{equation}\label{LCR}
\begin{aligned}
   \overline{X}_{D_0 D_{L_{k}}}&=\underset{k}{\underbrace{f\cdot f \cdot ... \cdot f }}(1), \quad k=1,2,3,...\\
   \overline{X}_{D_{L_{k-1}} D_{L_{k}}}&=g\cdot \underset{k-1}{\underbrace{f\cdot f \cdot ... \cdot f }}(1), \quad k=1,2,3,...
\end{aligned}
\end{equation}
where ``$\cdot$" denotes the composition of functions.
The configurations of the first critical point for arbitrary $n$ are then obtained. 

Interestingly, iterating the function \( f \) for any positive initial value causes the function value to grow indefinitely, eventually tending toward infinity. As we have
\[
\lim_{x \to +\infty} \frac{f(x)}{x} = 7,
\]
the value after the \((n+1)\)-th iteration is approximately 7 times the value after the \(n\)-th iteration, assuming the latter is large.
In contrast, the function \( g \) transforms a large number into a fixed value:
\[
\lim_{x \to +\infty} g(x) = \frac{9}{7}.
\]
Consequently, as the number of intervals \( n \) approaches infinity, both \(\overline{X}_{D_{L_{k-1}} D_{L_{k}}}\) and \(\overline{X}_{D_{R_{k-1}} D_{R_{k}}}\) converge to \(\frac{9}{7}\), while \(\overline{X}_{D_0 D_{L_{k}}}\) and \(\overline{X}_{D_0 D_{R_{k}}}\) tend to infinity. 
Recalling that \(\overline{X} = \frac{1}{X}\), these results indicate that as \( n \to \infty \), the distance between \( D_0 \) and \( D_{L_{k}} \) tends to infinity, while the distance between \( D_{L_{k-1}} \) and \( D_{L_{k}} \) approaches a constant value. This suggests that the first critical value of \( X_{AB} \) remains finite and greater than zero. The computation of the precise value of this critical point can be found in Appendix \ref{sec4.3}.

\paragraph{The second critical point}

In the case of \(n=1\) and assuming the sole cell is on the right of \(L_0\), the second critical point is determined by the following conditions
\begin{equation}\label{2ndn1}
\overline{X}_{D_0 D_c}=1, \quad \overline{X}_{BE_{R_1}}=1, \quad \widehat{L}_{D_0 D_{R_1}}=\widehat{L}_0 \widehat{L}_{D_c D_{R_1}}.
\end{equation}
The only difference from the conditions for the first critical point is the exchange of \(\widehat{L}_{D_0 D_{R_1}}\) and \(\widehat{L}_{D_c D_{R_1}}\) in the EWCS condition. Consequently, the solutions for \(\overline{X}_{D_0 D_{R_1}}\) and \(\overline{X}_{D_c D_{R_1}}\) are also exchanged. Thus, we have \(\overline{X}_{D_0 D_{R_1}}=f(1)\) and \(\overline{X}_{D_c D_{R_1}}=g(1)\).

Similarly, in the case of \(n>1\), as depicted in the MPT diagram \ref{MPTCase}(a), the only difference between the transition conditions for the first two critical points is the exchange of \(\widehat{L}_{D_0 D_{R_1}}\) and \(\widehat{L}_{D_c D_{R_1}}\) in the EWCS condition for the first right cell. The conditions for other right or left cells remain unchanged from those at the first critical point.
In other words, the transition conditions at the second critical point are \(\overline{X}_{D_0 D_c}=1\), \eqref{1stR}, and \eqref{1stL} with \(k\ge 1\).
The iteration can be applied in a similar manner; however, the initial condition for the right cells is now \(\overline{X}_0=1\), the same as the initial condition for the left cells. Thus, at the second critical point, the cross ratios for the \(k\)-th right cell are identical with that of the \(k\)-th left cell.
We obtain
\begin{equation}\label{2ndCRI}
\begin{aligned}
   \overline{X}_{D_0 D_{R_{k}}}&=\overline{X}_{D_0 D_{L_{k}}}=\underset{k}{\underbrace{f\cdot f \cdot ... \cdot f }}(1), \quad k=1,2,3,...\\
   \overline{X}_{D_{R_{k-1}} D_{R_{k}}}&=\overline{X}_{D_{L_{k-1}} D_{L_{k}}}=g\cdot \underset{k-1}{\underbrace{f\cdot f \cdot ... \cdot f }}(1), \quad k=1,2,3,...
\end{aligned}
\end{equation}
The transition conditions \(\overline{X}_{D_0 D_c}=1\), \eqref{1stR}, and \eqref{1stL} imply that {the following value is the maximum of \(\widehat{\mathrm{EI}}_{\Delta}(A:B|E)\)} 
\begin{equation}
    \mathrm{Max}(\widehat{\mathrm{EI}}_{\Delta}(A:B|E))=\widehat{L}_0^2=\left(\sqrt{\overline{X}_{D_0 D_c}}+\sqrt{1+\overline{X}_{D_0 D_c}} \right)^2=3+2\sqrt{2}
\end{equation}
at the second critical point for arbitrary \(n\) intervals of \(E\).

\paragraph{The third and higher order critical points}

In the case of \(n=1\), the third critical point is determined by
\begin{equation}\label{3rdn1}
\overline{X}_{D_c D_{R_1}}=1, \quad \overline{X}_{BE_{R_1}}=\overline{X}_{D_0 D_c}, \quad \text{and} \quad \widehat{L}_{D_0 D_{R_1}}=\widehat{L}_0 \widehat{L}_{D_c D_{R_1}}.
\end{equation}
Note that the first two conditions imply \(\overline{X}_{AB}=1\) according to \eqref{idX2}. We find \(\overline{X}_{D_0 D_{R_1}}=f(1)\) and \(\overline{X}_{D_0 D_c}=g(1)\).

{In the case of general \(n>1\),} as more cells are added to the entanglement wedge, numerical results show that the transition condition for the first cell remains as given in \eqref{3rdn1}. The conditions for the other right or left cells are the same with those at the first (or second) critical points. The iterative method can still be applied, with the initial conditions for the right and left cells now being \(\overline{X}_{D_0 D_{R_{1}}}=f(1)\) and \(\overline{X}_{D_0 D_c} = g(1)\), respectively. The cross ratios for the {k-th} right and left cells are determined as
\begin{equation}\label{3rdRCR}
\begin{aligned}
   \overline{X}_{D_0 D_{R_{k}}}=\underset{k}{\underbrace{f\cdot f \cdot ... \cdot f }}(1), \quad k=1,2,3,...\\
   \left\{\begin{matrix}
\overline{X}_{D_{c} D_{R_{1}}}=1,\\
\overline{X}_{D_{R_{k-1}} D_{R_{k}}}=g\cdot \underset{k-1}{\underbrace{f\cdot f \cdot ... \cdot f }}(1), \quad k=2,3,...
\end{matrix}\right.
\end{aligned}
\end{equation}
and 
\begin{equation}\label{3rdLCR}
\begin{aligned}
   \overline{X}_{D_0 D_{L_{k}}}=\underset{k}{\underbrace{f\cdot f \cdot ... \cdot f }}\cdot g(1), \quad k=1,2,3,...\\
   \overline{X}_{D_{L_{k-1}} D_{L_{k}}}=g\cdot \underset{k-1}{\underbrace{f\cdot f \cdot ... \cdot f }}\cdot g(1), \quad k=1,2,3,...
\end{aligned}
\end{equation}
Since the EWCS conditions of the third critical point are the same as that of the second critical point, the ratio \eqref{ratioEn} at the third critical point is found to be
\begin{equation}
\begin{aligned}
    \mathrm{Max}(\widehat{\mathrm{EI}}_{\Delta}(A:B|E))=\widehat{L}_0^2&=\left(\sqrt{\overline{X}_{D_0 D_c}}+\sqrt{1+\overline{X}_{D_0 D_c}} \right)^2\\&=\left(\sqrt{g(1)}+\sqrt{1+g(1)} \right)^2\\&\approx 7.5504.
\end{aligned}
\end{equation}
Note that this value is also irrelevant to the number of intervals \(n\). 

Additional critical points exist if \(n > 1\). In fact, for a system with \(n\) intervals, there are \(n+2\) critical points. For the first three critical points, the cross ratios for the first right cell are determined by solving the conditions specific to that cell, and the cross ratios for other cells are then derived through iteration.
The \(k\)-th (\(k=3,4,5,\ldots\)) critical point can be identified by solving the conditions for the \(k-2\) cells near the central geodesic—specifically, the first \([\frac{k-2}{2}]\) left cells and \([\frac{k-1}{2}]\) right cells—and then performing right/left iterations for the other cells, starting from the cross ratios corresponding to the rightmost/leftmost geodesics of these \(k-2\) cells, \ie, \(\widehat{L}_{D_{0}D_{R_{[(k-1)/2]}}}\) and \(\widehat{L}_{D_{0}D_{L_{[(k-2)/2]}}}\).
In these \(k-2\) cells, the conditions for the \(k\)-th (\(k>2\)) critical point are depicted in the MPT diagram in Figure \ref{MPTCase}(a), which illustrates that more and more gap regions are marginally connected. These conditions can be solved numerically.
Thus, by using the iteration method, we can determine the cross ratios of the gap regions at the \(k\)-th critical points for \(n\) intervals from those at the \(k\)-th critical points for \(k-2\) intervals (for any \(n \ge k-2\)).

The precise values of $X_{AB}$ at all critical points in Case I are computed in Appendix \ref{sec4.3}.
The corresponding results are summarized in Table \ref{TableI} and Figure \ref{CPs2}.

\subsubsection{Case II}

\noindent In Case II, it is also convenient to divide the fully connected $\mathrm{EW}(ABE)$ into cells as depicted in Figure \ref{CellCaseII}.
In Case I, all cells are grouped into ``right" or ``left" cells which are separated by a central geodesic, while in Case II, we divide the cells into an ``upper" group and a ``lower" group.
The boundary that separates them is the geodesic connecting the RT surfaces of the region $D_{U_1}$ and $D_{L_1}$.
We have $\left[ \frac{n}{2} \right]$ upper cells and $\left[ \frac{n+1}{2} \right]$ lower cells\footnote{This separation is only a convention we choose. Due to the right-left symmetry (switching $A$ and $B$) and upper-lower symmetry (switching the upper and lower intervals of $E$ ) of the configuration of the system, the boundary can also be the geodesic between the EW of region $D_{U_{[n/2]+1}}$ and $D_{L_{[(n+1)/2]+1}}$, and there may also be $\left[ \frac{n}{2}+1 \right]$ upper cells and $\left[ \frac{n}{2} \right]$ lower cells.}.
The upper/lower intervals of $E$ are therefore denoted as $E_{U_i}$/$E_{L_j}$ ($i=1, 2,...[\frac{n}{2}]$, $j=1,2,...,[\frac{n+1}{2}]$).
Similarly the upper/lower gap regions are denoted respectively as $D_{U_i}$/$D_{L_j}$ ($i=1,2,...[\frac{n}{2}]+1$, $j=1,2,...,[\frac{n+1}{2}]+1$). 
\begin{figure}[H]
 \centering
 \includegraphics[scale=0.7]{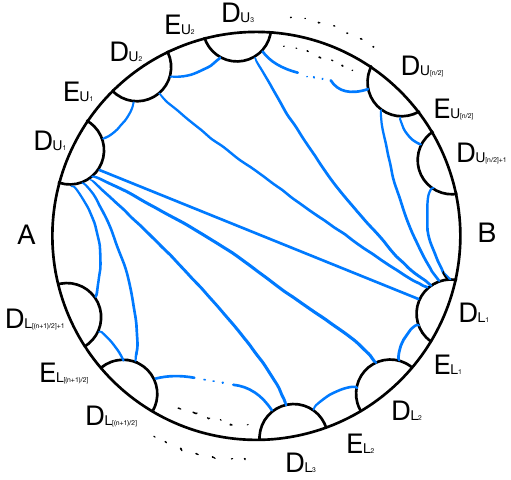}
 \caption{The entanglement wedge of $ABE$ is divided into ``cells" by a series of EWCSs which are shown in blue.
 The geodesic that connects the RT surfaces of $D_{U_1}$ and $D_{L_1}$ is the boundary of the upper cells and lower cells.
 The intervals and gap regions that are associated with the upper/lower cells are labeled with sub-indices $U$/$L$.
 }
\label{CellCaseII}
\end{figure}

As in Case I, the above way of dividing $\mathrm{EW}(ABE)$ into cells are also in accordance with the transition conditions (the MPT diagram in Figure \ref{MPTCase} (b)) of critical points.
The transition conditions for the case of $n$ intervals are also contained in that for the case of $n+1$ intervals.
As the intervals of $E$ (or cells) increase, the new interval can be viewed as added one by one to the lower gap region $D_{L_{[(n+1)/2]+1}}$ and the upper gap region $D_{U_{[n/2]}}$ of $\mathrm{EW}(ABE)$ in turn\footnote{The new cells are added to $D_{L_{[(n+1)/2]+1}}$ first and then to $D_{U_{[n/2]}}$, and this procedure is iterated.}.
With each new cell added to $\mathrm{EW}(ABE)$, we have two more transition conditions while the conditions for the existing cells remain unchanged.
In other words, the $(n+1)$th cell does not interfere with the first $n$ cells, but it is built by splitting $A$/$B$ into three intervals: a new interval $E_U$/$E_L$, a new gap region $D_U$/$D_L$ and a new smaller $A$/$B$.

In the case of $n=1$, there is only one lower cell in $\mathrm{EW}(ABE)$, the configuration of which is exactly the same as the case of $n=1$ in Case I.
We discuss the case of $n>1$ in the following.

\paragraph{The first critical point}

The new transition conditions brought by the $k$th ($k>1$) upper cell are 
\begin{equation}\label{1stU}
    \overline{X}_{\left(BD_{U_{[n/2]+1}}E_{U_{[n/2]}}...D_{U_{k+2}}E_{U_{k+1}}\right)E_{U_k}}=1, \quad \widehat{L}_{D_{L_{1}} D_{U_{k+1}}}=\widehat{L}_{D_{L_{1}} D_{U_{k}}} \widehat{L}_{D_{U_{k}} D_{U_{k+1}}}.
\end{equation}
By using \eqref{LABE} and \eqref{idX1}, \eqref{1stU} is equivalent to
\begin{equation}\label{iterEQU}
    \begin{aligned}
     &\quad 2\overline{X}_{D_{L_{1}} D_{U_{k}}}+1+\overline{X}_{D_{L_{1}} D_{U_{k+1}}}+\overline{X}_{D_{U_{k}} D_{U_{k+1}}}-\overline{X}_{D_{L_{1}} D_{U_{k+1}}}\overline{X}_{D_{U_{k}} D_{U_{k+1}}}=0,\\
    &\left(\sqrt{\overline{X}_{D_{L_{1}} D_{U_{k}}}}+\sqrt{1+\overline{X}_{D_{L_{1}} D_{U_{k}}}}\right)\left(\sqrt{\overline{X}_{D_{U_{k}} D_{U_{k+1}}}}+\sqrt{1+\overline{X}_{D_{U_{k}} D_{U_{k+1}}}}\right)\\
    &=\sqrt{\overline{X}_{D_{L_{1}} D_{U_{k+1}}}}+\sqrt{1+\overline{X}_{D_{L_{1}} D_{U_{k+1}}}}.
    \end{aligned}
\end{equation}
One can therefore find $\overline{X}_{D_{U_{k}} D_{U_{k+1}}}$ and $\overline{X}_{D_{L_{1}} D_{U_{k+1}}}$ provided $\overline{X}_{D_{L_{1}} D_{U_{k}}}$ is known.
This recurrence relation is traced down to the first upper cell where $\overline{X}_{D_{U_{1}} D_{U_{2}}}$ and $\overline{X}_{D_{L_{1}} D_{U_{2}}}$ can be represented by $\overline{X}_{D_{L_{1}} D_{U_{1}}}$.
Since $\overline{X}_{D_{L_{1}} D_{U_{1}}}$ is already solved in the first lower cell, all cross ratios related to the upper cells can be solved by iterations of \eqref{iterEQU}.

The new transition conditions brought by the $k$th lower cell are the same as \eqref{1stU} except all subscripts "U" are switched to "L":
\begin{equation}\label{1stLo}
    \overline{X}_{\left(AD_{L_{[(n+1)/2]+1}}E_{L_{[(n+1)/2]}}...D_{L_{k+1}}E_{L_{k}}\right)E_{L_k}}=1, \quad \widehat{L}_{D_{U_1} D_{L_{k+1}}}=\widehat{L}_{D_{U_1} D_{L_{k}}} \widehat{L}_{D_{L_{k}} D_{L_{k+1}}}.
\end{equation}
Similarly, \eqref{1stLo} is equivalent to
\begin{equation}\label{iterEQLo}
    \begin{aligned}
    &\quad 2\overline{X}_{D_{U_1} D_{L_{k}}}+1+\overline{X}_{D_{U_1} D_{L_{k+1}}}+\overline{X}_{D_{L_{k}} D_{L_{k+1}}}-\overline{X}_{D_{U_1} D_{L_{k+1}}}\overline{X}_{D_{L_{k}} D_{L_{k+1}}}=0,\\
    &\left(\sqrt{\overline{X}_{D_{U_1} D_{L_{k}}}}+\sqrt{1+\overline{X}_{D_{U_1} D_{L_{k}}}}\right)\left(\sqrt{\overline{X}_{D_{L_{k}} D_{L_{k+1}}}}+\sqrt{1+\overline{X}_{D_{L_{k}} D_{L_{k+1}}}}\right)\\
    &=\sqrt{\overline{X}_{D_{U_1} D_{L_{k+1}}}}+\sqrt{1+\overline{X}_{D_{U_1} D_{L_{k+1}}}}.
    \end{aligned}
\end{equation}
Likewise $\overline{X}_{D_{L_{k}} D_{L_{k+1}}}$ and $\overline{X}_{D_{U_1} D_{L_{k+1}}}$ are obtained if $\overline{X}_{D_{U_1} D_{L_{k}}}$ is known.
Since $\overline{X}_{D_{U_1} D_{L_{2}}}$ is already solved in the first lower cell, all cross ratios related to the lower cells can be solved by iterations of \eqref{iterEQLo} with the initial condition $\overline{X}_{D_{U_1} D_{L_{2}}}=g(1)$ .
Note that the EWCS condition of the $2$nd,...,$\left[\frac{n+1}{2}\right]$th lower cells is different from that of the first lower cell in that the role of $\overline{X}_{D_{L_1} D_{L_2}}$ and $\overline{X}_{D_{U_1} D_{L_2}}$ are exchanged.

The equations \eqref{iterEQU} and \eqref{iterEQLo} both have the form \eqref{iterEQ}, where $\overline{X}_{D_{L_1} D_{U_{k}}}$(or $\overline{X}_{D_{U_1} D_{L_{k}}}$), $\overline{X}_{D_{U_{k}} D_{U_{k+1}}}$(or $\overline{X}_{D_{L_{k}} D_{L_{k+1}}}$), and $ \overline{X}_{D_{L_1} D_{U_{k+1}}}$(or $\overline{X}_{D_{U_1} D_{L_{k+1}}}$) are replaced by $\overline{X}_{2j-2}$, $\overline{X}_{2j-1}$ and $\overline{X}_{2j}$, respectively.
As obtained in Case I, the solution to $\overline{X}_{2j-1}$ and $\overline{X}_{2j}$ is
\begin{equation}
    \overline{X}_{2j}=f(x), \quad \overline{X}_{2j-1}=g(x),
\end{equation}
where $\overline{X}_{2j-2}$ is abbreviated as $x$.
The initial condition for the upper cell is $\overline{X}_{0}=1$, while the initial condition for the lower cell is $\overline{X}_{1}=g(1) \approx 1.4207$.
$\overline{X}_{i}$, $(i=1,2,3,...)$ can then be generated from \eqref{iterEQ} by iteration. 
Restoring the subscripts of cross ratios, for the upper cells, we obtain 
\begin{equation}\label{UCR}
\begin{aligned}
   \overline{X}_{D_{L_1} D_{U_{k}}}&=\underset{k-1}{\underbrace{f\cdot f \cdot ... \cdot f }}(1), \quad k=1,2,3,...\\
   \overline{X}_{D_{U_{k}} D_{U_{k+1}}}&=g\cdot \underset{k-1}{\underbrace{f\cdot f \cdot ... \cdot f }}(1), \quad k=1,2,3,...
\end{aligned}
\end{equation}
For the lower cells, we obtain
\begin{equation}\label{LoCR}
\begin{aligned}
   \overline{X}_{D_{U_1} D_{L_{k}}}=\underset{k-2}{\underbrace{f\cdot f \cdot ... \cdot f }}\cdot g(1), \quad k=2,3,...\\
   \left\{\begin{matrix}
\overline{X}_{D_{L_{1}} D_{L_{2}}}=f(1),\\
\overline{X}_{D_{L_{k}} D_{L_{k+1}}}=g\cdot \underset{k-2}{\underbrace{f\cdot f \cdot ... \cdot f }}\cdot g(1), \quad k=2,3,...
\end{matrix}\right.
\end{aligned}
\end{equation}
where ``$\cdot$" denotes the composition of functions.
The configurations of the first critical points for arbitrary $n$ are then obtained.

We see that the upper/lower cells are analogues of the left/right cells and the related cross ratios between gap regions have the same structure as in Case I.
This result implies that in Case II, the first critical value of $X_{AB}$ is also non-vanishing.
These critical values are computed in Appendix \ref{sec4.3}.

\paragraph{The second critical point}

As illustrated in the MPT diagram \ref{MPTCase} (b), the transition condition of the second critical point of Case II is completely parallel to that of Case I. The only difference is that $\widehat{L}_{D_{U_1} D_{L_2}}$ and $\widehat{L}_{D_{L_1} D_{L_2}}$ are exchanged in the EWCS condition of the first lower cell.
The conditions of other upper or lower cells remain the same as that of the first critical points, \ie, \(\overline{X}_{D_{U_{1}}D_{L_{1}}}=1\), \eqref{1stU}, and \eqref{1stLo}.
We obtain
\begin{equation}\label{2ndCRII}
\begin{aligned}
   \overline{X}_{D_{L_{1}} D_{U_{k}}}&=\overline{X}_{D_{U_1} D_{L_{k}}}=\underset{k-1}{\underbrace{f\cdot f \cdot ... \cdot f }}(1), \quad k=1,2,3,...\\
   \overline{X}_{D_{U_{k}} D_{U_{k+1}}}&=\overline{X}_{D_{L_{k}} D_{L_{k+1}}}=g\cdot \underset{k-1}{\underbrace{f\cdot f \cdot ... \cdot f }}(1), \quad k=1,2,3,...
\end{aligned}
\end{equation}
The transition conditions imply that
\begin{equation}
    \mathrm{Max}(\widehat{\mathrm{EI}}_{\Delta}(A:B|E))=\widehat{L}_{D_{L_1} D_{U_1}}^2=\left(\sqrt{\overline{X}_{D_{L_{1}}D_{U_{1}}}}+\sqrt{1+\overline{X}_{D_{L_{1}}D_{U_{1}}}} \right)^2=3+2\sqrt{2}
\end{equation}
at the second critical point for arbitrary \(n\) intervals of \(E\).

\paragraph{The third and higher order critical points}

It is observed from our numerical results that the transition condition of the first cell is
\begin{equation}\label{3rdn1II}
\begin{aligned}
    &\overline{X}_{D_{L_1} D_{L_2}}=1, \quad \overline{X}_{AD_{L_{[(n+1)/2]+1}}E_{L_{[(n+1)/2]}}...D_{L_{3}}E_{L_{2}}}=\overline{X}_{D_{U_1} D_{L_1}}, \\
    &\widehat{L}_{D_{U_1} D_{L_2}}=\widehat{L}_{D_{U_1} D_{L_1}} \widehat{L}_{D_{L_1} D_{L_2}}.
\end{aligned}
\end{equation}
Note that the first two conditions imply
\begin{equation}
\overline{X}_{\left( AD_{L_{[(n+1)/2]+1}}E_{L_{[(n+1)/2]}}...D_{L_{3}}E_{L_{2}} \right)\left( BD_{U_{[n/2]+1}}E_{U_{[n/2]}}...D_{U_{2}}E_{U_{1}} \right)}=1    
\end{equation}
according to \eqref{idX2} and we have $\overline{X}_{D_{U_1} D_{L_2}}=f(1)$ and $\overline{X}_{D_{U_1} D_{L_1}}=g(1)$.
The conditions of other right or left cells remain the same as that of the first (or the second) critical points.
The iteration can still be applied.
The initial condition for the upper and lower cells are now $\overline{X}_{D_{U_1} D_{L_2}} \sim \overline{X}_2=f(1)$ and $ \overline{X}_{D_{U_1} D_{L_2}} \sim \overline{X}_0=g(1)$, respectively.
The cross ratios of the upper and lower cells are found to be
\begin{equation}\label{3rdUCR}
\begin{aligned}
   \overline{X}_{D_{L_1} D_{U_{k}}}=\underset{k-1}{\underbrace{f\cdot f \cdot ... \cdot f }}\cdot g(1), \quad k=1,2,3,...\\
   \overline{X}_{D_{U_{k}} D_{U_{k+1}}}=g\cdot \underset{k-1}{\underbrace{f\cdot f \cdot ... \cdot f }}\cdot g(1), \quad k=1,2,3,...
\end{aligned}
\end{equation}
and
\begin{equation}\label{3rdLoCR}
\begin{aligned}
   \overline{X}_{D_{U_1} D_{L_{k}}}=\underset{k-1}{\underbrace{f\cdot f \cdot ... \cdot f }}(1), \quad k=2,3,...\\
   \left\{\begin{matrix}
\overline{X}_{D_{L_1} D_{L_{2}}}=1,\\
\overline{X}_{D_{L_{k}} D_{L_{k+1}}}=g\cdot \underset{k-1}{\underbrace{f\cdot f \cdot ... \cdot f }}(1), \quad k=2,3,...
\end{matrix}\right.
\end{aligned}
\end{equation}
The ratio \eqref{ratioEn} at the third critical point can be found by the EWCS conditions as
\begin{equation}
\begin{aligned}
\mathrm{Max}(\widehat{\mathrm{EI}}_{\Delta}(A:B|E))=\widehat{L}_{D_{L_1} D_{U_1}}^2&=\left(\sqrt{\overline{X}_{D_{L_{1}}D_{U_{1}}}}+\sqrt{1+\overline{X}_{D_{L_{1}}D_{U_{1}}}} \right)^2\\&=\left(\sqrt{g(1)}+\sqrt{1+g(1)} \right)^2\\&\approx 7.5504,
\end{aligned}
\end{equation}
which is also irrelevant to the number of intervals \(n\).

As in Case I, there are also $n+2$ critical points in the case of $n$ intervals.
The $k$th $(k=3,4,5,...)$ critical point can be found by solving the conditions of the $k-2$ cells near the geodesic connecting the EW of $D_{U_1}$ and $D_{L_1}$ (more precisely, first $[\frac{k-2}{2}]$ upper cells and $[\frac{k-1}{2}]$ lower cells) and then perform the right/left iterations to other cells starting from the cross ratios corresponding to the ``uppermost"/``lowermost'' geodesics.
For the EWCS combination to have $k$ critical points, we need at least $k-2$ cells (\ie, $k-2$ intervals of $E$).
In these $k-2$ cells, the conditions of the $k$th $(k>2)$ critical point are illustrated in Figure \ref{MPTCase} (b) which shows that more and more gap regions are connected in the MPT diagram.
These conditions can be solved numerically.
Based on the numerical solution of the $k-2$ cells, the cross ratios of the gap regions at the $k$th critical points for $n$ intervals $(n \ge k-2)$ can be obtained by iteration \eqref{iterEQ}.
The initial condition of the iteration for the upper/lower cells is $\overline{X}_0=\overline{X}_{D_{L_1} D_{U_{[(k-2)/2]+1}}}$ and $\overline{X}_0=\overline{X}_{D_{U_1} D_{L_{[(k-1)/2]+1}}}$ respectively.

The precise values of $X_{AB}$ at all critical points in Case II are computed in Appendix \ref{sec4.3}.
The corresponding results are summarized in Table \ref{TableII} and Figure \ref{CPs2}.

\subsection{\(\mathrm{Max}(\mathrm{EI}_\Delta (A:B|E))\) and the entanglement of assistance}\label{secEoA}

\noindent We now compare \(\mathrm{Max}(\mathrm{EI}_{\Delta}(A:B|E))\) with \(\mathrm{HE}(A:B)\), which is the minimal EWCS that separates A from B in the entanglement wedge of \(ABE\) and equals one half of the holographic entanglement of assistance \(\mathrm{EoA}(AA^*:BB^*|EE^*)\), in the case of general interval number \(n\).
For a fixed interval number \(n\), at all critical points beyond the second, we find
\(\mathrm{Max}(\mathrm{\widehat{EI}}_{\Delta}(A:B|E)) = \widehat{L}_0^{2}\) in Case I and \(\mathrm{Max}(\mathrm{\widehat{EI}}_{\Delta}(A:B|E)) = \widehat{L}_{D_{L_1} D_{U_1}}^{2}\) in Case II, as shown in the previous subsection. 
In fact, as suggested by \eqref{transition}, these relations hold throughout all phases except the first two. 
Since \(\widehat{L}_0\) and \(\widehat{L}_{D_{L_1} D_{U_1}}\) are indeed the smallest among all possible EWCSs that separate \(A\) and \(B\) in the entanglement wedge of \(ABE\) in Case I and II, respectively, 
these two values coincide with \(\mathrm{\widehat{HE}}(A:B)\) in the corresponding case.
We therefore conclude that \(\mathrm{Max}(\mathrm{\widehat{EI}}_{\Delta}(A:B|E)) \le \mathrm{\widehat{HE}}(A:B)^2\) and
\begin{equation}
   \mathrm{Max}(\mathrm{EI}_{\Delta}(A:B|E)) \le  2 \mathrm{HE}(A:B|E)
\end{equation}
with the inequality saturated when \(X_{AB}\) is larger than the second critical points.

\section{Conclusion and Discussion}\label{sec5}

\noindent In this paper, we introduced the holographic EWCS triangle information, $\mathrm{EI}_\Delta (A:B|E)$, as a linear combination of entanglement wedge cross sections that is both positive and upper bounded, capturing assisted quantum correlations in a holographic tripartite mixed state \(ABE\).
We also defined in the canonical purification state $\lvert \sqrt{\rho_{ABE}} \rangle$ the holographic entanglement of assistance $\mathrm{EoA}(AA^*:BB^*|EE^*)$ which provides a universal bound on $\mathrm{EI}_\Delta$:
\begin{equation}
   \mathrm{EI}_\Delta(A:B|E)\le \mathrm{EoA}(AA^*:BB^*|EE^*) = 2\mathrm{HE}({A:B|E}), 
\end{equation}
showing that $\mathrm{EI}_\Delta$ is upper bounded by the minimal EWCS cuts that separate $A$ and $B$ in the entanglement wedge of \(ABE\).

Specializing to AdS$_3$/CFT$_2$, we compute the maximum value of $\mathrm{EI}_\Delta$ over all possible auxiliary regions $E$ (with a fixed $A$ and $B$) using generalized MPT diagrams, and uncovered a rich phase structure governed by the cross ratio $X_{AB}$. 
In both Case I and II, for each fixed number of intervals in $E$, the maximized $\mathrm{EI}_\Delta$ vanishes below a critical threshold of $X_{AB}$. It then increases monotonically with $X_{AB}$, and diverges as $X_{AB} \to \infty$.
In particular, while \(A\) and \(B\) are uncorrelated for \(X_{AB}<1\) (\ie, \(I(A:B)=0\)), introducing the auxiliary system \(E\) makes the $\mathrm{Max}(\mathrm{EI}_\Delta)$ nonzero as soon as \(X_{AB}\) exceeds the threshold of the first critical point. This highlights the entanglement enhancing assisted by \(E\). 
Moreover, the maximum saturates the upper bound $2\,\mathrm{HE}(A:B|E)$, which equals to the entanglement of assistance in the canonical purification state ($\mathrm{EoA}(AA^*:BB^*|EE^*)$), when $X_{AB}$ exceeds the second critical point.
For fixed $X_{AB}$, the maximum grows with $n$ but approaches a finite limit as $n\to \infty$.
Finally, comparing the Case I and II directly, Case II always yields a larger maximum value than Case I for the same $n$ and $X_{AB}$.

The explicit numerical results of $\mathrm{Max}(\mathrm{EI}_\Delta(A:B|E))$ are summarized in Table \ref{TableI} for the Case I, {where all $n$ intervals of $E$ stay in the same gap region between $A$ and $B$,} and in Table \ref{TableII} for Case II, {where the $n$ intervals of $E$ are split equally in the gap regions between and outside $A$ and $B$,} respectively.
The corresponding maximum values $\mathrm{Max}(\widehat{\mathrm{EI}}_{\Delta}(A:B|E))$ are shown in Figures \ref{CPs1} and \ref{CPs2}.
\begin{table}[!ht]
    \centering
    \begin{tabular}{|l|l|l|l|l|l|l|l|}
    \hline
        n & 1st & 2nd  & 3rd  & 4th  & 5th  & 6th  \\ \hline
        1 & 0.09517028 & 0.70386783 & 1 & - & - & - \\ \hline
        2 & 0.07735027  & 0.51814928 & 0.75757510 & 1 & - & - \\ \hline
        3 & 0.06597732 & 0.49830390 & 0.72418528 & 0.95906706 & 1 & - \\ \hline
        4 & 0.06418252 & 0.47945866 & 0.69937172 & 0.92196698 & 0.96167928 & 1 \\ \hline
        5 & 0.06280714 & 0.47686199 & 0.69502451 & 0.91667430 & 0.95604653 & 0.99419457 \\ \hline
        ...... & ...... & ...... & ...... & ...... & ...... & ...... \\ \hline
        $\infty$ & 0.06229791 & 0.47342947 & 0.69034886 & 0.90991263 & 0.94903516 & 0.98679139 \\ \hline
    \end{tabular}
    \caption{ Critical values of the ``order parameter" $X_{AB}$ at the 1st, 2nd, 3rd... critical points in Case I, where all $n$  intervals of $E$ stay in the same gap region between $A$ and $B$. }\label{TableI}
\end{table}

\begin{table}[!ht]
    \centering
    \begin{tabular}{|l|l|l|l|l|l|l|l|}
    \hline
        n & 1st  & 2nd  & 3rd & 4th  & 5th  & 6th   \\ \hline
        1 & 0.09517028 & 0.70386783 & 1 & - & - & - \\ \hline
        2 & 0.06698730 & 0.49542992 & 0.72081590  & 1 & - & - \\ \hline
        3 & 0.05699941 & 0.47423274 & 0.68488263 & 0.95331201 & 1 & - \\ \hline
        4 & 0.05456067 & 0.45394249 & 0.65749728 & 0.90880380 & 0.95369581 & 1 \\ \hline
        5 & 0.05337342 & 0.45116312 & 0.65281155 & 0.90274552 & 0.94717639 & 0.99322107 \\ \hline
        ...... & ...... & ...... & ...... & ...... & ...... & ...... \\ \hline
        $\infty$ & 0.05279979 & 0.44748598 & 0.64769242 & 0.89473657 & 0.93880823 & 0.98426082 \\ \hline
    \end{tabular}
    \caption{Critical values of the ``order parameter" $X_{AB}$ at the 1st, 2nd, 3rd... critical points in Case II, where the $n$  intervals of $E$ are split equally in the gap regions between and outside $A$ and $B$. }\label{TableII}
\end{table}

\begin{figure}[hbpt]
	\centering
	\includegraphics[scale=0.4]{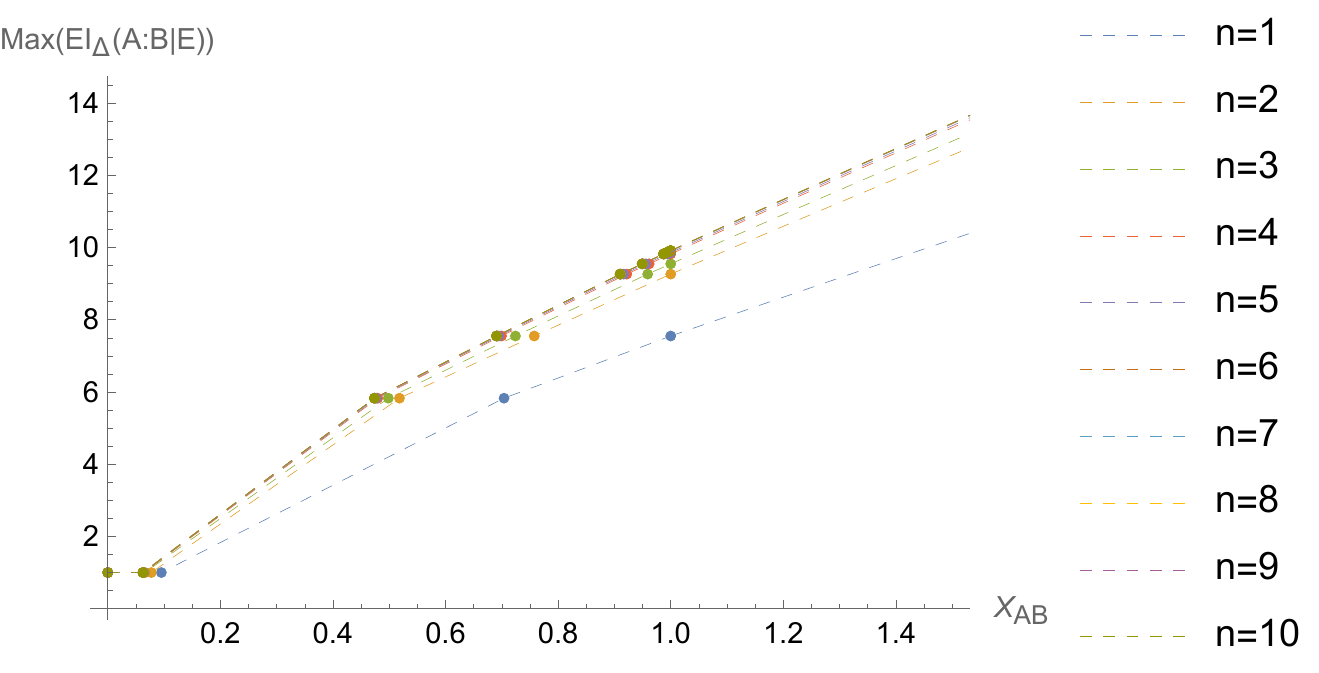}
	\caption{Values of $\mathrm{Max}(\mathrm{EI}_{\Delta}(A:B|E))$ at the critical points in Case I, with different colors indicating different numbers of interval $n$ in region $E$. The maximum value at critical points are shown in dots and connected by dashed lines which indicate the tendency of $\mathrm{Max}(\mathrm{EI}_{\Delta}(A:B|E))$ as a function of $X_{AB}$ in the case of different $n$.
	}\label{CPs1}
\end{figure}

\begin{figure}[hbpt]
	\centering
	\includegraphics[scale=0.4]{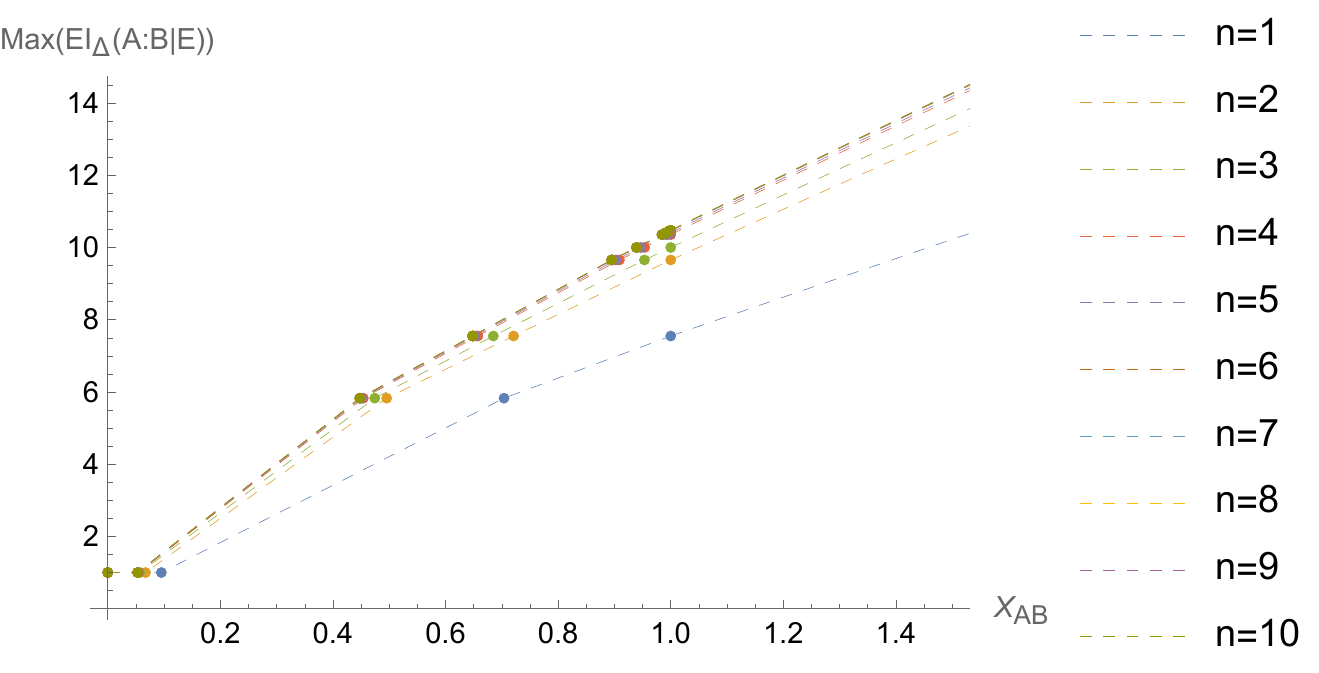}
	\caption{Values of $\mathrm{Max}(\mathrm{EI}_{\Delta}(A:B|E))$ at the critical points in Case II, with different colors indicating different numbers of interval $n$ in region $E$. The maximum value at critical points are shown in dots and connected by dashed lines which indicate the tendency of $\mathrm{Max}(\mathrm{EI}_{\Delta}(A:B|E))$ as a function of $X_{AB}$ in the case of different $n$.
	}\label{CPs2}
\end{figure}

These results position $\mathrm{Max}(\mathrm{EI}_\Delta(A:B|E))$ as a diagnostic of bipartite entanglement between $A$ and $B$ with the help of $E$ in holography.
Natural extensions include higher dimensions and time-dependent settings, and a sharper operational meaning in holography relative to entanglement of assistance.
It would also be interesting to construct alternative combinations of geometric quantities in holography which could serve as candidate diagnostics of multi-partite entanglement.
The general combination may involve not only entanglement entropies and EWCSs, but also other proposed multi-partite entanglement quantifications such as multi-EWCSs and multi-entropies.

\section*{Acknowledgement}

\noindent We thank Bart{\l}omiej Czech, Xuan-Ting Ji and Yi Ling for their valuable discussions.
This work is supported in part by the National Natural Science Foundation of China under Grant No. 12347183, 12505053, and 12035016.

\appendix

\section{Holographic entanglement entropy, cross ratios and connection conditions }\label{sec2}

\noindent In this section we introduce several basic ingredients that are used for the computation of $\mathrm{Max}(\mathrm{EI}_{\Delta}(A:B|E))$.
We work in the context of AdS$_3$/CFT$_2$ and consider pure AdS$_3$ in this section.

\subsection{Holographic entanglement entropy}

\noindent According to the RT proposal \cite{Ryu:2006bv} for the holographic entanglement entropy, the entanglement entropy of a subregion \( A \) in the boundary conformal field theory is proportional to the area of the bulk minimal surface \( \gamma_A \) that is homologous to \( A \)
\begin{equation}
	S_A=\frac{\text{Area}\left(\gamma_A\right)}{4G_N}.
\end{equation}
In AdS$_3$, the minimal surface for a single interval \(A\) reduces to the geodesic that connects the two endpoints of \(A\).
The length of the geodesic is \(2 L \log \frac{R}{\epsilon}\) in Poincaré coordinates, where \(R\) is the length of \(A\), \(\epsilon\) is the UV cutoff and \(L\) is the AdS radius.  
The entanglement entropy of \(A\) is then given by
\begin{equation}
	S_A=\frac{c}{3}\log \frac{R}{\epsilon},
\end{equation}
where \(c\) is the central charge and we have used the identity \(\frac{L}{G_N}=\frac{2c}{3}\)\cite{Brown:1986nw}.

\subsection{Cross ratios and connection conditions}\label{sec2.2}

\noindent A time slice of global AdS$_3$ can be represented as the Poincaré upper half plane.
The conformal boundary is a 1 dimensional straight line (the $x$ axis).
Given two disjoint intervals \(P=(x_i, x_j)\) and \(Q=(x_{k}, x_{l})\) on this line, a cross ratio is defined to be
\begin{equation}\label{CrossratioDef}
    X_{PQ}:=\frac{x_{ij} x_{kl}}{x_{jk} x_{il}},
\end{equation}
where we have denoted \(x_{mn}:=x_m-x_n\).
The cross ratio is constructed to be invariant under one dimensional conformal transformations, \ie, the translation, dilation and special conformal transformation\footnote{Note that in one dimension, there are \(N-3\) independent cross ratios that are constructed from \(4\) out of \(N\) points. In \(D\) dimensions, there are \(D\cdot N\) variables to build cross ratios
and \((D+2)(D+1)/2\) constraints coming from the freedom of conformal transformations. Thus in general, the number of independent cross ratios is \(D\cdot N-(D+2)(D+1)/2\). Therfore, we have only \(1\) independent cross ratio given \(4\) points.
All other cross ratios constructed from the \(4\) points can be represented by \eqref{CrossratioDef}.}. 

Since the time slice of AdS$_3$ is isometric under conformal transformations, it has other representations in a plane, which are conformal to the upper half plane.
A well known representation is the Poincaré disk.
The cross ratio of two intervals \(P=(\theta_i, \theta_j)\) and \(Q=(\theta_{k}, \theta_{l})\) on the conformal boundary is 
\begin{equation}\label{Xac}
    X_{PQ}=\frac{\sin \frac{\theta_{ij}}{2} \sin \frac{\theta_{kl}}{2}}{\sin \frac{\theta_{jk}}{2} \sin \frac{\theta_{il}}{2}},
\end{equation}
where \(\theta_{mn}:=\theta_m-\theta_n\). 
Here, $\theta_i, \theta_j, \theta_{k}, \theta_{l}$ are angular coordinates of four endpoints on the boundary. The intervals are defined counter-clockwise, \ie, $\theta_j$($\theta_{l}$) is counter-clockwise to $\theta_i$($\theta_{k}$).
We use this representation in Section \ref{sec4} in the numerical computations.

The cross ratios can be used to indicate the connectivity of the entanglement wedge of boundary intervals.
For instance, the connectivity of the EW for two disjoint intervals $P$ and $Q$ is determined by their mutual information (MI), which can be represented by the cross ratio $X_{PQ}$ as
\begin{equation}\label{IPQ}
	\begin{aligned}
		I(P:Q)&=S_{P}+S_{Q}-S_{PQ}\\
		&=\frac{c}{3}\log\left(\mathrm{max}\left\{1,\;\; X_{PQ} \right\}\right).
	\end{aligned}
\end{equation}
The EW of $P$ and $Q$ is connected if the MI is nonzero, \ie, $X_{PQ}>1$, otherwise the MI vanishes and the EW is disconnected.

For $3$ intervals $A$, $B$ and $E$, as depicted in Figure \ref{Fig3a}, there are $5$ possible configurations of EW.
\begin{figure}[H]
\centering
     \includegraphics[width=8cm]{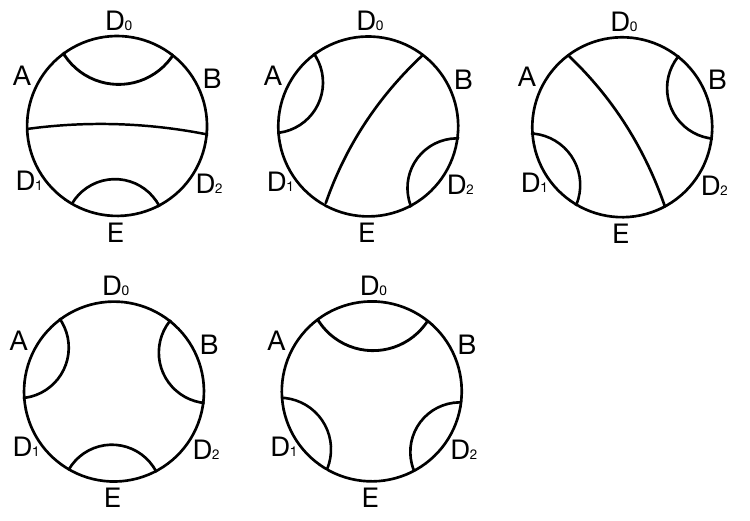}
 \caption{The 5 possible configurations for 3 intervals $A$, $B$ and $E$.
 The 3 gap regions between them are denoted as $D_0$, $D_1$ and $D_2$.}
\label{Fig3a}
\end{figure}
To ensure that we have a connected configuration, all the disconnected configurations must have larger geodesic length than that of the totally connected configuration.
Thus there are $4$ connection conditions that constrain the cross ratios, namely
\begin{equation}\label{con}
    \frac{X_{AB}}{X_{D_1 D_2}}>1,\quad \frac{1}{X_{D_0 D_1}}>1,\quad \frac{1}{X_{D_0 D_2}}>1,\text{ and } \frac{1}{X_{D_1 D_2}}>1.
\end{equation}
Each condition corresponds to a configuration that is not totally connected. 
One can read off these conditions by subtracting the totally connected diagram from each of the disconnected diagrams and demanding all these differences to be positive, as shown in Figure \ref{Fig3b}.
\begin{figure}[h]
\centering
     \includegraphics[width=10cm]{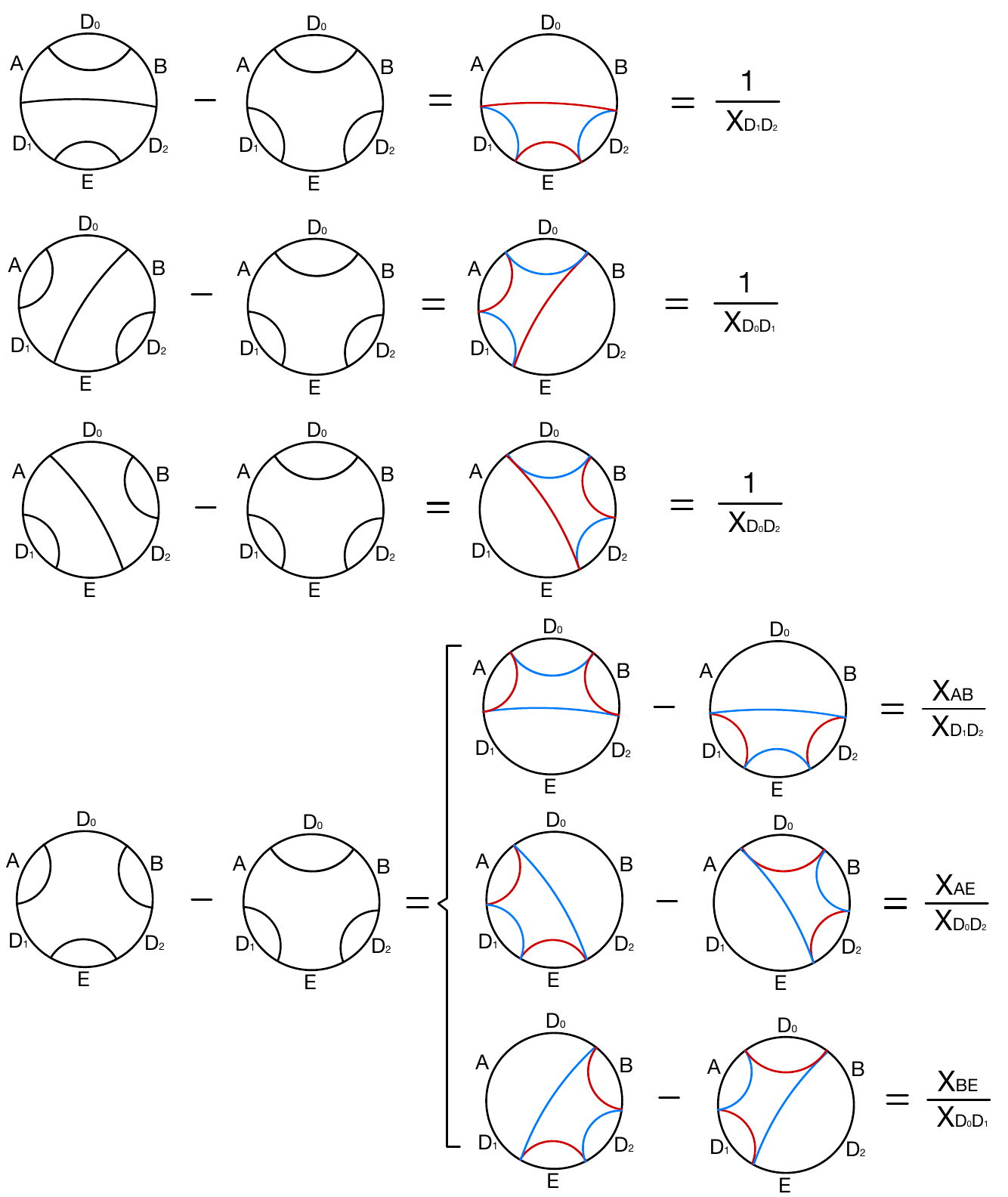}
 \caption{A diagrammatic representation of the connection conditions \eqref{con}. On the leftmost of the equalities, the 4 disconnected diagrams are subtracted by the totally connected diagram, resulting in 4 functions of cross ratios on the rightmost of the equalities, where we have omitted the common factor $\frac{c}{3}$ and the logarithm.
 The cross ratios are read off from combinations of the geodesics (RT surfaces) which are shown in the middle of the equalities.
 The red/blue curves represent the geodesics whose lengths have positive/negative signs in the calculation of the cross ratios.
 The combinations are not unique. 3 different combinations of the geodesics are listed for the last subtraction, which correspond to 3 different sets of cross ratios.}
\label{Fig3b}
\end{figure}
Note that the expressions of the connection conditions in terms of cross ratios are not unique, because there are identities among cross ratios.
For example, we have $\frac{X_{AB}}{X_{D_1 D_2}}=\frac{X_{AE}}{X_{D_0 D_2}}=\frac{X_{BE}}{X_{D_0 D_1}}$.
We will prove this equation in the following section.

To find all the connection conditions for $n$ intervals, one may first look for all the possible configurations for the corresponding EW and then subtract the connected diagram from each disconnected diagram to find all the connection conditions.
The possible configurations are related to the partition of interval number $n$.
For example, the partition of $n=5$ is 
\begin{equation}
\begin{aligned}
    5=&1+1+1+1+1\\
     =&2+1+1+1\\
     =&2+2+1\\
     =&3+1+1\\
     =&3+2\\
     =&4+1.\\
\end{aligned}
\end{equation}
The configurations can be classified according to how intervals are connected in that configuration, thereby each partition of $n$ corresponds to a kind of configuration.
For instance, the partition $3+1+1$ corresponds to the configurations where three intervals of $E$ are connected while the other two intervals are disconnected and also isolated from the connected three intervals, as illustrated in Figure \ref{Fig4}(a).
\begin{figure}[h]
\centering
\subfigure[]{
		\begin{minipage}[b]{.9\linewidth}
			\centering
			\includegraphics[scale=0.5]{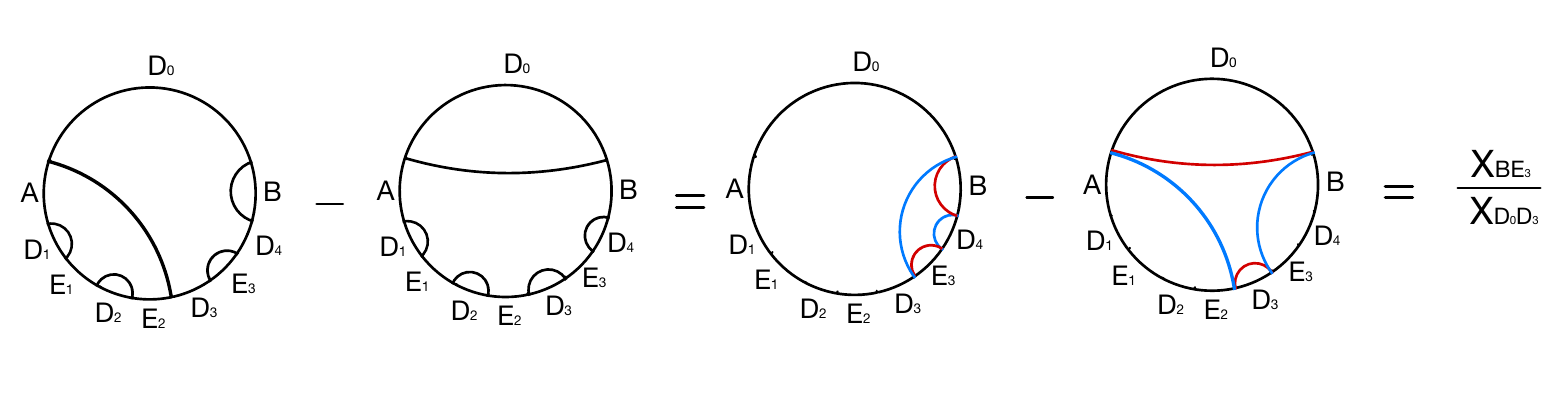}
		\end{minipage}
	}\\
	\subfigure[]{
		\begin{minipage}[b]{.9\linewidth}
			\centering
			\includegraphics[scale=0.5]{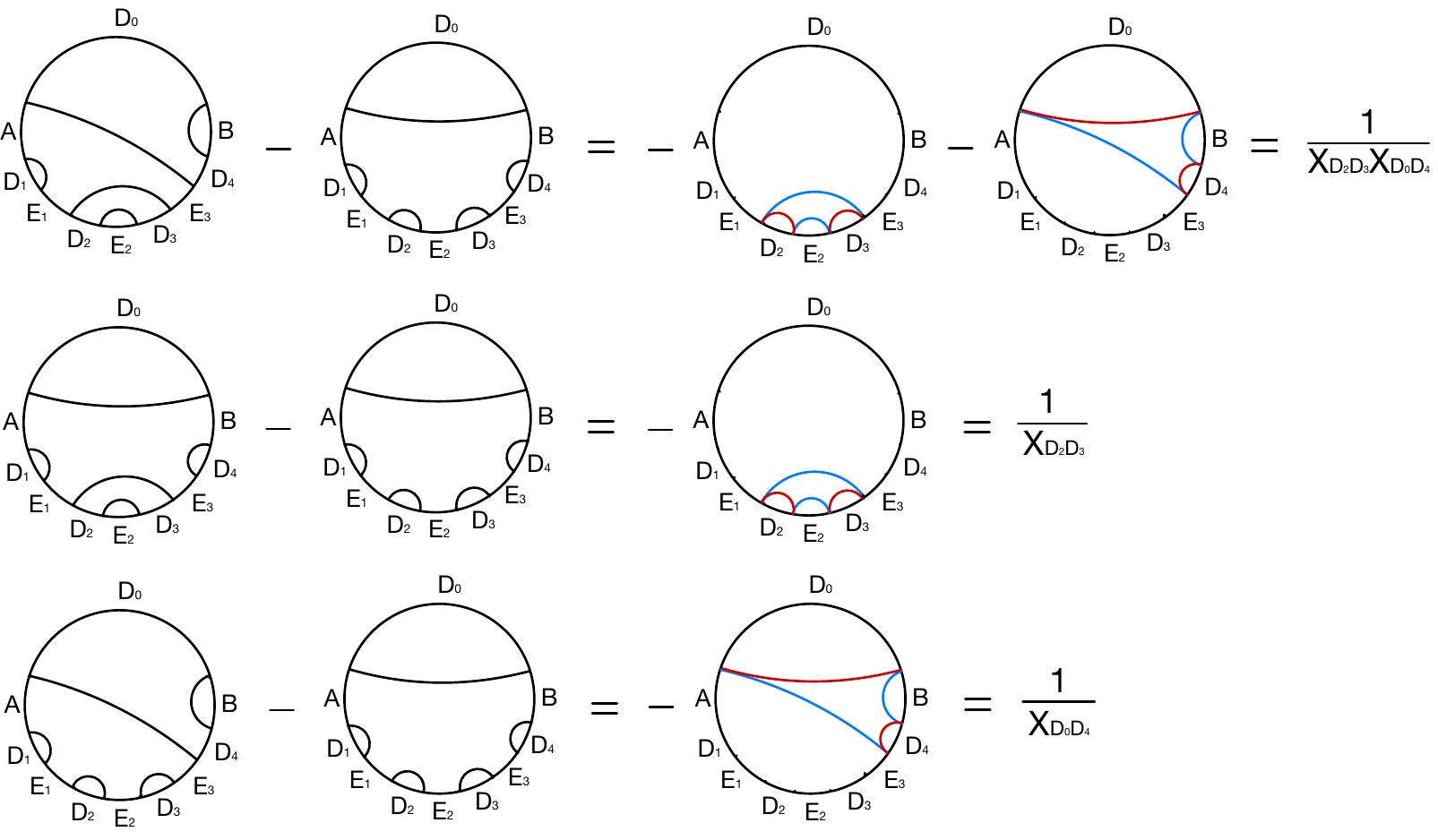}
		\end{minipage}
	}
 \caption{(a) A configuration of partition $3+1+1$ which does not have any ``jump". This configuration corresponds to $\frac{X_{BE_3}}{X_{D_0 D_3}}$ after subtracted by the totally connected configuration.  
 (b) The first line shows a configuration of partition $3+1+1$ which has a ``jump" and (after subtracted by the totally connected configuration) corresponds to $\frac{1}{X_{D_2 D_3} X_{D_0 D_4}}$.
 The second and third line depict two configurations that correspond to $\frac{1}{X_{D_2 D_3}}$ and $\frac{1}{X_{D_0 D_4}}$ respectively.
 The connection condition $\frac{1}{X_{D_2 D_3} X_{D_0 D_4}}>1$ that corresponds to the configuration in the first line is redundant if the connection conditions correspond to the second and third line ($\frac{1}{X_{D_2 D_3}}>1$ and $\frac{1}{X_{D_0 D_4}}>1$) are fulfilled.
 The red/blue curves in the diagrams represent the geodesics whose lengths have positive/negative signs in the calculation of the cross ratios.}
\label{Fig4}
\end{figure}
However, it is not necessary to examine all the configurations. 
Configurations containing non-adjacent connected intervals (``jumps") lead to repeated connection conditions with those that do not have ``jumps" and are redundant, see Figure \ref{Fig4}(b).
After excluding configurations with ``jumps", the problem of finding all connection conditions boils down to counting the possible ways to partition the intervals into groups of successive intervals.
For each gap region between intervals, one may put a partition or not, so that there are $2^n$ ways in total to perform such partition in $n$ gap regions. But this counts the one with no partition, which corresponds to the totally connected configuration, and $n$ that has only one partition, which are not associated with any EWs.
Therefore, there are $2^n-n-1$ connection conditions for $n$ intervals, each of which can be represented in terms of cross ratios in a systematic way.

\section{The derivation of the equation (\ref{idX1})}\label{CRidX1}

\noindent In this appendix, we derive the equation (\ref{idX1}) from the configuration shown in Figure \ref{Figabc}.
\begin{figure}[h]
 \centering
 \includegraphics[scale=1]{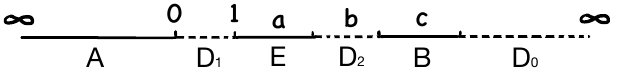}
 \caption{The configuration of the boundary systems $A$, $B$ and $E$.
 The left endpoint of $A$, the right endpoint of $A$ and the left endpoint of $E$ are located at infinity, $0$ and $1$ respectively. $a$, $b$ and $c$ are the lengths of $E$, $D_2$ and $B$ respectively.
 }
\label{Figabc}
\end{figure}
Three of the six endpoints of the intervals can be fixed by the three degrees of freedom of conformal transformations.
Thus, without loss of generality, we choose the interval $A$ as the negative half line of the $x$ axis and set the left endpoint of $E$ to be 1.
The lengths of intervals $E$, $D_2$ and $B$ are the remaining three degrees of freedom, which are denoted as $a$, $b$ and $c$ respectively.
The cross ratios $\overline{X}_{AB}$, $\overline{X}_{AE}$, $\overline{X}_{AE}$, $\overline{X}_{D_0 D_1}$, $\overline{X}_{D_0 D_2}$, $\overline{X}_{D_1 D_2}$ can be represented by $a$, $b$, $c$ as
\begin{equation}\label{abc}
\begin{aligned}
    \overline{X}_{AE}=\frac{1}{a}, \quad \overline{X}_{BE}=\frac{b(a+b+c)}{a c}, \quad \overline{X}_{AB}=\frac{1+a+b}{c}, \\ 
    \overline{X}_{D_0 D_1}=a+b+c, \quad \overline{X}_{D_0 D_2}=\frac{c}{b}, \quad \overline{X}_{D_1 D_2}=\frac{a(1+a+b)}{b}.
\end{aligned}
\end{equation}
The equation (\ref{idX1}) can then be verified from \eqref{abc} directly.

\section{Critical points of $X_{AB}$ for general $n$}\label{sec4.3}

\noindent In Section \ref{sec4.2}, we have solved the transition conditions at all the critical points and obtained all the $2n+1$ independent cross ratios between gap regions {$(X_{D_i D_j})$} for each critical points. 
{These $2n+1$ cross ratios completely determine the configuration of the entanglement wedge of $ABE$, such that all other variables such as the cross ratios between intervals $A$, $B$, $E_i$ can be expressed by them.}
In this appendix, we derive the precise critical values of $X_{AB}$ from these $2n+1$ cross ratios.

Note that a cell is constructed by 3 gap regions and 3 intervals that are denoted as $I_1, I_2, I_3$ and $G_1, G_2, G_3$ respectively, as seen in Figure \ref{Cell}.
\begin{figure}[H]
 \centering
 \includegraphics[scale=0.6]{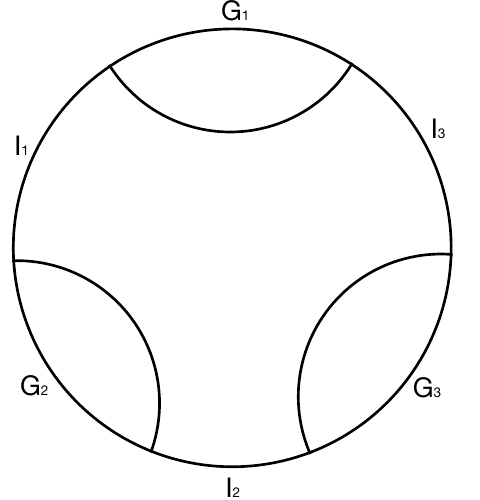}
 \caption{A cell consists of three intervals and three gap regions denoted respectively as $I_j$ and $G_j$ where $j=1,2,3$.
 }
\label{Cell}
\end{figure}
The configuration of a cell involves 6 cross ratios, 3 between the gap regions and 3 between the intervals.
One fact that we will take advantage of in the following repeatedly is that, for each cell, the configuration is fixed by 3 cross ratios among the 6 cross ratios. 
Once we obtain 3 cross ratios, the other 3 can be solved from the following 3 equations,
\begin{equation}\label{Xeq3}
\begin{aligned}
  \overline{X}_{G_1 G_2}-\overline{X}^2_{I_2 I_3} \overline{X}_{G_2 G_3} \overline{X}_{G_3 G_1}+\overline{X}_{I_2 I_3}(1+\overline{X}_{G_1 G_2}+\overline{X}_{G_2 G_3}+\overline{X}_{G_3 G_1})=0,\\
  \overline{X}_{G_2 G_3}-\overline{X}^2_{I_3 I_1} \overline{X}_{G_3 G_1} \overline{X}_{G_1 G_2}+\overline{X}_{I_3 I_1}(1+\overline{X}_{G_1 G_2}+\overline{X}_{G_2 G_3}+\overline{X}_{G_3 G_1})=0,\\
  \overline{X}_{G_3 G_1}-\overline{X}^2_{I_1 I_2} \overline{X}_{G_1 G_2} \overline{X}_{G_2 G_3}+\overline{X}_{I_1 I_2}(1+\overline{X}_{G_1 G_2}+\overline{X}_{G_2 G_3}+\overline{X}_{G_3 G_1})=0,
\end{aligned}
\end{equation}
which are just a rewrite of equation \eqref{idX1}, but here $I_i$ and $G_i$ represent arbitrary intervals and gap regions respectively.

\subsection{Case I}

\noindent We first compute $X_{AB}$ at each critical point in Case I with the knowledge of the cross ratios between gap regions. The configuration of this case is shown in Figure \ref{CellCaseI}.

\paragraph{The 1st and 2nd critical points}
  
We start from the first right cell, the intervals and gap regions of which are defined as $I_1=A D_{L_{\left[n/2 \right]}}  E_{L_{\left[n/2 \right]}}...D_{L_{1}}  E_{L_{1}}$, $I_2=E_{R_{1}}$, $I_3=\left.B D_{R_{\left[(n+1)/2 \right]}}  E_{R_{\left[(n+1)/2 \right]}} ...D_{R_{2}}  E_{R_{2}}\right.$, $G_1=D_0$, $G_2=D_c$, $G_3=D_{R_{1}}$.
The cross ratios between the three gap regions $G_1$, $G_2$ and $G_3$ are obtained in \ref{sec4.2} and $\overline{X}_{I_3 I_1}$ can be obtained from \eqref{Xeq3} to be
\begin{equation}
\begin{aligned}
\overline{X}_{I_3 I_1}=&\frac{1+\overline{X}_{G_1 G_2}+\overline{X}_{G_2 G_3}+\overline{X}_{G_3 G_1}}{2 \overline{X}_{G_3 G_1} \overline{X}_{G_1 G_2}}\\
&+\frac{\sqrt{(1+\overline{X}_{G_1 G_2}+\overline{X}_{G_2 G_3}+\overline{X}_{G_3 G_1})^2+4 \overline{X}_{G_1 G_2} \overline{X}_{G_2 G_3} \overline{X}_{G_3 G_1}}}{2 \overline{X}_{G_3 G_1} \overline{X}_{G_1 G_2}}
\end{aligned}
\end{equation}

Then one may consider a new bigger cell that connects the intervals $I_1'=\left. A D_{L_{\left[n/2 \right]}}  E_{L_{\left[n/2 \right]}}\right.$ \\ $\left....D_{L_{1}}  E_{L_{1}}\right.$, $I_2'=E_{R_{2}}$, $I_3'=B D_{R_{\left[(n+1)/2 \right]}}  E_{R_{\left[(n+1)/2 \right]}}...D_{R_{3}}  E_{R_{3}}$ and is bordered by the gap regions $G_1'=D_0$, $G_2'=D_c E_{R_{1}} D_{R_{1}}$, $G_3'=D_{R_{2}}$.
Note that we have $\overline{X}_{G_{1}' G_{2}'}=1/\overline{X}_{I_3 I_1}$, $\overline{X}_{I_{2}' I_{3}'}=1$ and $\overline{X}_{G_{3}' G_{1}'}=f(g(1))$ (according to \eqref{RCR}). 
Thus, we have 3 known cross ratios and the configuration of the new cell can be solved.
From \eqref{Xeq3}, we obtain
\begin{equation}
    \overline{X}_{I_3' I_1'}=\frac{\overline{X}_{G_1' G_2'}+(1+\overline{X}_{G_1' G_2'}+\overline{X}_{G_3' G_1'})\overline{X}_{I_{2}' I_{3}'}}{\overline{X}_{G_1' G_2'}(\overline{X}_{G_3' G_1'} \overline{X}_{I_{2}' I_{3}'}-1)}.
\end{equation}

Keep going on, we can iterate this computation.
By merging $G_2$, $I_2$ and $G_3$ as our new $G_2$, removing a right interval ($E$) and a gap region ($D$) from $I_3$ and letting the removed interval ($E$) to be the new $I_2$ each time, we have a new cell. The cross ratio $\overline{X}_{I_3 I_1}$ thus obtained can then be regarded as the input for the next cell.
More generally, we define a series of cells with the $k$th ($k=1,2,3,...,[(n+1)/2]$) cell related to the intervals $I_1^{(k)}=A D_{L_{\left[n/2 \right]}}  E_{L_{\left[n/2 \right]}}...D_{L_{1}}  E_{L_{1}}$, $I_2^{(k)}=E_{R_{k}}$, $I_3^{(k)}=B D_{R_{\left[(n+1)/2 \right]}}  E_{R_{\left[(n+1)/2 \right]}}...D_{R_{k+1}}  E_{R_{k+1}}$ and the gap regions $G_1^{(k)}=D_0$, $G_2^{(k)}=D_c E_{R_{1}} D_{R_{1}}...E_{R_{k-1}} D_{R_{k-1}}$, $G_3^{(k)}=D_{R_k}$.
We find the following iteration relation\footnote{$\overline{X}_{I_3^{(n+1)} I_1^{(n+1)}}$ and $\overline{X}_{G_3 G_1}^{(n+1)}$ have been renamed as $\overline{X}_{I_3 I_1}^{(n+1)}$ and $\overline{X}_{G_3 G_1}^{(n+1)}$ respectively for conciseness.}
\begin{equation}
    \overline{X}_{I_3 I_1}^{(n+1)}=\frac{\overline{X}_{G_3 G_1}^{(n+1)}+1}{\overline{X}_{G_3 G_1}^{(n+1)}-1}\overline{X}_{I_3 I_1}^{(n)}+\frac{2}{\overline{X}_{G_3 G_1}^{(n+1)}-1},
\end{equation}
where 
\begin{equation}
    \overline{X}_{G_3 G_1}^{(n)}=\underset{n-1}{\underbrace{f\cdot f \cdot ... \cdot f }}\cdot g(1), \quad n=1,2,3,....
\end{equation}
This process can be continued until $I_3=B$ and we obtain $\overline{X}_{B (A D_{L_{\left[n/2 \right]}}  E_{L_{\left[n/2 \right]}}...D_{L_{1}}  E_{L_{1}})}$.
Then one may do a similar iteration in the left cells until $I_1=A$ and finally we could find $X_{AB}$.
The first and the second critical points can be obtained in this way, the results of which are listed in the first two columns of Table \ref{TableI}.

\paragraph{The 3rd and higher order critical points}

As demonstrated in Section \ref{sec4.2}, at the $k$th ($k \ge 3$) critical point, the configuration of the $k-2$ cells near the central geodesic are fixed by the transition conditions shown in Figure \ref{MPTCase} (a).
Through these conditions, boundary coordinates of the endpoints (instead of cross ratios) of each interval in these cells are obtained.
For one of such conditions, we have 
\begin{equation}\label{XAB4}
    X_{(A D_{L_{\left[n/2 \right]}}  E_{L_{\left[n/2 \right]}}...D_{L_{[(k-2)/2]+1}} E_{L_{[(k-2)/2]+1}})(B D_{R_{\left[(n+1)/2 \right]}}  E_{R_{\left[(n+1)/2 \right]}}...D_{R_{[(k-1)/2]+1}}  E_{R_{[(k-1)/2]+1}})}=1.
\end{equation}
On the other hand, the transition conditions outside these $k-2$ cells are the same as that of the first three critical points, which can also be seen in Figure \ref{MPTCase} (a).
Let
\begin{equation}
\begin{aligned}
    I_1&=A D_{L_{\left[n/2 \right]}}  E_{L_{\left[n/2 \right]}}...D_{L_{[(k-2)/2]+1}} E_{L_{[(k-2)/2]+1}}\\
    I_2&=E_{R_{[(k-1)/2]+1}}\\
    I_3&=B D_{R_{\left[(n+1)/2 \right]}}  E_{R_{\left[(n+1)/2 \right]}}...D_{R_{[(k-1)/2]+2}} E_{R_{[(k-1)/2]+2}}\\
    G_1&=D_0\\
    G_2&=D_{L_{[(k-2)/2]}}E_{L_{[(k-2)/2]}}...D_c...E_{R_{[(k-1)/2]}}D_{R_{[(k-1)/2]}}\\
    G_3&=D_{R_{[(k-1)/2]+1}},
\end{aligned}
\end{equation}
and noting that equation \eqref{XAB4} is equivalent to $X_{G_1 G_2}=1$,
we may perform the same iteration as the first three critical points.  
The explicit numerical results are summarized in the tables and figures in Section \ref{sec5}.
In particular, the results of the 3-6th critical points are listed in the 3-6th columns of Table \ref{TableI}.

\subsection{Case II}

\noindent In Case II, we consider the configuration shown in Figure \ref{CellCaseII}. 
The method of obtaining $X_{AB}$ is different from the method in Case I.
This is due to a lack of a mutual boundary of all the cells in Case II, in contrast to Case I, where such mutual boundary exists: the gap region $D_0$.
In Case II, $D_{U_1}$ is the boundary of all lower cells and $D_{L_1}$ is the boundary of all upper cells.
However our goal is the same. We need to merge more and more intervals and gap regions meanwhile removing more and more intervals and gap regions from $I_1$ and $I_3$ until $I_1=A$ and $I_3=B$.

In the following, the cross ratios between intervals (gap regions) will be denoted as $\overline{X}^{L(k)}_{I_i I_j}$ or $\overline{X}^{U(k)}_{I_i I_j}$ ($\overline{X}^{L(k)}_{G_i G_j}$ or $\overline{X}^{U(k)}_{G_i G_j}$) with the upper index $L/U$ indicating the lower/upper cells, $(k)$ denoting the $k$th cell in the iteration process and $i,j=1,2,3$.

\paragraph{The 1st and 2nd critical points}
  
We start first by merging lower cells and then merge upper cells. 
In the first lower cell, \ie, $I_1^{(1)}=A D_{L_{[(n+1)/2]+1}} E_{L_{[(n+1)/2]}}...D_{L_3} E_{L_2}$, $I_2^{(1)}=E_{L_1}$, $I_3^{(1)}=B D_{U_{[n/2]+1}} E_{U_{[n/2]}}...D_{U_2} E_{U_1}$, $G_1^{(1)}=D_{U_1}$, $G_2^{(1)}=D_{L_2}$, $G_3^{(1)}=D_{L_1}$, we find from \eqref{Xeq3} that
\begin{equation}
    \overline{X}^{L(1)}_{I_3 I_1}=\frac{2\overline{X}^{L(1)}_{G_3 G_1}+1+\overline{X}^{L(1)}_{G_1 G_2}}{(\overline{X}^{L(1)}_{G_1 G_2}-1)\overline{X}^{L(1)}_{G_3 G_1}},
\end{equation}
where $\overline{X}^{L(1)}_{G_1 G_2}=g(1)$ and $\overline{X}^{L(1)}_{G_3 G_1}=1$.
We then merge $D_{L_1}$, $E_{L_1}$, $D_{L_2}$ into a new bigger gap region $G_3^{(2)}$ and consider the new cell that corresponds to intervals $\left.I_1^{(2)}=A D_{L_{[(n+1)/2]+1}} \right.$\\ $\left.E_{L_{[(n+1)/2]}}...D_{L_4} E_{L_3}\right.$, $I_2^{(2)}=E_{L_2}$, $I_3^{(2)}=I_3^{(1)}$, and is bordered by gap regions $G_1^{(2)}=G_1^{(1)}$, $G_2^{(2)}=D_{L_3}$, $G_3^{(2)}$.
The two variables that may be used in the iterative computation of $X_{AB}$ are $\overline{X}^{L(2)}_{I_3 I_1}$ and $\overline{X}^{L(2)}_{G_2 G_3}$.
Note that $\overline{X}^{L(2)}_{G_3 G_1}=1/\overline{X}^{L(1)}_{I_3 I_1}$ and $\overline{X}^{L(2)}_{I_1 I_2}=1$, which follows from the transition condition (Figure \ref{MPTCase} (b)), $\overline{X}^{L(2)}_{I_3 I_1}$ and $\overline{X}^{(2)}_{G_2 G_3}$ can be then obtained from the 2nd and 3rd equations in \eqref{Xeq3}:
\begin{equation}\label{ULstep2}
    \begin{aligned}
    \overline{X}^{L(2)}_{I_3 I_1}=\frac{2+(1+\overline{X}^{L(2)}_{G_1 G_2})\overline{X}^{L(1)}_{I_3 I_1}}{\overline{X}^{L(2)}_{G_1 G_2}-1}\\
    \overline{X}^{L(2)}_{G_2 G_3}= \frac{2+(1+\overline{X}^{L(2)}_{G_1 G_2})\overline{X}^{L(1)}_{I_3 I_1}}{(\overline{X}^{L(2)}_{G_1 G_2}-1)\overline{X}^{L(1)}_{I_3 I_1}}.
    \end{aligned}
\end{equation}
We now repeat this process. 
By merging $G_2^{(2)}$, $I_2^{(2)}$ and $G_3^{(2)}$ into a larger gap region $G_2^{(3)}$,  meanwhile removing the lower interval $E_{L_3}$ and the gap region $D_{L_4}$ from $I_3^{(2)}$, and finally letting $I_2^{(3)}=E_{L_3}$, we have a new cell.
The cross ratios $\overline{X}^{L(3)}_{I_3 I_1}$ and $\overline{X}^{L(3)}_{G_2 G_3}$ can be represented by $\overline{X}^{L(2)}_{I_3 I_1}$ and $\overline{X}^{L(3)}_{G_1 G_2}$ which are known variables from the previous step \eqref{ULstep2} and the transition condition $\overline{X}^{L(3)}_{G_1 G_2}=\overline{X}_{D_{U_1} D_{L_4}}$ respectively.
This process can be repeated $N=\left[\frac{n+1}{2} \right]$ times until $I_1^{(N)}=A$, $I_2^{(N)}=E_{L_{[(n+1)/2]}}$ and $I_3^{(N)}=I_3^{(1)}$.
At the $k$th step, we have 
\begin{equation}\label{ULstepk}
    \begin{aligned}
    \overline{X}^{L(k)}_{I_3 I_1}=\frac{2+(1+\overline{X}^{L(k)}_{G_1 G_2})\overline{X}^{L(k-1)}_{I_3 I_1}}{\overline{X}^{L(k)}_{G_1 G_2}-1}\\
    \overline{X}^{L(k)}_{G_2 G_3}= \frac{2+(1+\overline{X}^{L(k)}_{G_1 G_2})\overline{X}^{L(k-1)}_{I_3 I_1}}{(\overline{X}^{L(k)}_{G_1 G_2}-1)\overline{X}^{L(k-1)}_{I_3 I_1}},
    \end{aligned}
\end{equation}
where $\overline{X}^{L(k)}_{G_1 G_2}=\overline{X}_{D_{U_1} D_{L_{k+1}}}$ are known from \eqref{LoCR}.
The outcome of these iteration processes are two series of cross ratios related to the lower cells: $\overline{X}^{L(k)}_{I_3 I_1}$ and $\overline{X}^{L(k)}_{G_2 G_3}$ ($k=1,2,3,...,\left[\frac{n+1}{2} \right]$).

The parallel iteration process can be performed in the upper cells and the outcome are two series of cross ratios related to these cells: $\overline{X}^{U(k)}_{I_3 I_1}$ and $\overline{X}^{U(k)}_{G_2 G_3}$ ($k=1,2,3,...,\left[\frac{n}{2} \right]$).

$X_{AB}$ can finally be found by a new iteration that takes advantage of two of the above four series of cross ratios.
In the case where region $E$ has $n$ intervals, we have $\left[\frac{n+1}{2} \right]$ lower cells and $\left[\frac{n}{2} \right]$ upper cells.
We define a new series of cells by putting all the upper intervals and gap regions into $G_1$ and combining more and more lower intervals and gap regions to form $G_3$, \ie, $G_1^{(k)}=D_{U_1} E_{U_1}...D_{U_{[n/2]}} E_{U_{[n/2]}} D_{U_{[n/2]+1}}$, $G_2^{(k)}=D_{L_{k+1}}$, $G_3^{(k)}=D_{L_{1}} E_{L_{1}} ... D_{L_{k-1}} E_{L_{k-1}} D_{L_{k}}$ where $k=1,2,3,...,\left[\frac{n+1}{2} \right]$.
At $k=1$, since $\left. \overline{X}_{G_3^{(1)} G_1^{(1)}}=\overline{X}_{D_{L_1} (D_{U_1} E_{U_1}...D_{U_{[n/2]}} E_{U_{[n/2]}} D_{U_{[n/2]+1}})}\right.$\\ $\left.=1/\overline{X}^{U([n/2])}_{I_3 I_1}\right.$, $\overline{X}_{G_2^{(1)} G_3^{(1)}}=\overline{X}_{D_{L_2} D_{L_1}}=\overline{X}^{L(1)}_{G_2 G_3}=f(1)$, and $\overline{X}_{I_1^{(1)} I_2^{(1)}}=1$, we have three variables, and we obtain
\begin{equation}
    \overline{X}_{I_3^{(1)} I_1^{(1)}}= \overline{X}^{L(1)}_{G_2 G_3} \overline{X}^{U([n/2])}_{I_3 I_1}.
\end{equation}
At general $k$, since $\overline{X}_{G_3^{(k)} G_1^{(k)}}=1/\overline{X}_{I_3^{(k-1)} I_1^{(k-1)}}$, $\overline{X}_{I_1^{(k)} I_2^{(k)}}=1$ and $\overline{X}_{G_2^{(k)} G_3^{(k)}}=\overline{X}^{L(k)}_{G_2 G_3}$, we have
\begin{equation}
    \overline{X}_{I_3^{(k)} I_1^{(k)}}=\frac{\overline{X}_{G_2^{(k)} G_3^{(k)}}}{\overline{X}_{G_3^{(k)} G_1^{(k)}}}=\overline{X}^{L(k)}_{G_2 G_3} \overline{X}_{I_3^{(k-1)} I_1^{(k-1)}},
\end{equation}
where $k=1,2,3,...,\left[\frac{n+1}{2} \right]$.
By iteration, we end up at
\begin{equation}
    \overline{X}_{AB}^{n}=\overline{X}_{I_3^{([(n+1)/2])} I_1^{([(n+1)/2])}}=\left(\prod_{k=1}^{[(n+1)/2]} \overline{X}^{L(k)}_{G_2 G_3}\right) \overline{X}^{U([n/2])}_{I_3 I_1},
\end{equation}
where $\overline{X}_{AB}^{n}$ is the first critical point of the cross ratio of $A$ and $B$ if region $E$ has $n$ intervals.
Alternatively, one may define another series of cells by putting all the lower intervals and gap regions into $G_1$ and combining more and more lower intervals and gap regions to form $G_3$, \ie, $G_1^{(k)}=D_{L_{1}} E_{L_{1}} ... D_{L_{[(n+1)/2]}} E_{L_{[(n+1)/2]}} D_{L_{[(n+1)/2]+1}}$, $G_2^{(k)}=D_{U_{k+1}}$, $G_3^{(k)}=D_{U_1} E_{U_1}...D_{U_{k-1}} E_{U_{k-1}} D_{U_{k}}$ where $k=1,2,3,...,\left[\frac{n}{2} \right]$. 
By a parallel iteration, we have
\begin{equation}
    \overline{X}_{AB}^{n}=\overline{X}_{I_3^{([n/2])} I_1^{([n/2])}}=\left(\prod_{k=1}^{[n/2]} \overline{X}^{U(k)}_{G_2 G_3}\right) \overline{X}^{L([(n+1)/2])}_{I_3 I_1}.
\end{equation}
The final results of $X_{AB}$ are the same.

The second critical point can be obtained from exactly the same iteration process.
The only difference is that the initial inputs in the first lower cell are now $\overline{X}^{L(1)}_{G_1 G_2}=f(1)$ (instead of $\overline{X}^{L(1)}_{G_1 G_2}=g(1)$), and $\overline{X}^{L(1)}_{G_3 G_1}=1$.
The results of the first and the second critical points are listed in the first two columns of Table \ref{TableII}.

\paragraph{The 3rd and higher order critical points}

As demonstrated in Section \ref{sec4.2}, the configuration of the $k$th ($k=3,4,5,...$) critical point when $E$ has $n$ ($n=k-2,k-1,k,...$) intervals is obtained by iteration starting from the configuration of the $k$th critical point when $E$ has $k-2$ intervals.
The latter configuration contains $k-2$ cells ($\left[\frac{k-2}{2}\right]$ upper cells and $\left[\frac{k-1}{2}\right]$ lower cells) which are fixed by the transition conditions shown in Figure \ref{MPTCase} (b).
Through these conditions, boundary coordinates of the endpoints (instead of cross ratios) of each interval in these cells are obtained.
In particular, we have $X_{AB}=1$.

If there are $n$ intervals (cells) and $n>k-2$, the transition conditions outside the $k-2$ cells are the same as that of the first two critical points (Figure \ref{MPTCase} (b)) and the iteration process described above still applies.
We define a series of lower cells,
\begin{equation}
\begin{aligned}
    I_1^{(k)}&=A D_{L_{\left[(n+1)/2 \right]+1}}  E_{L_{\left[(n+1)/2 \right]}} D_{L_{\left[(n+1)/2 \right]}}...D_{L_{\left[(k-1)/2\right]+3}} E_{L_{\left[(k-1)/2\right]+2}}\\
    I_2^{(k)}&=E_{L_{\left[(k-1)/2\right]+1}}\\
    I_3^{(k)}&=B D_{U_{\left[n/2 \right]+1}}  E_{U_{\left[n/2 \right]}} D_{U_{\left[n/2 \right]}}...D_{U_{2}} E_{U_{1}}\\
    G_1^{(k)}&=D_{U_1}\\
    G_2^{(k)}&=D_{L_{\left[(k-1)/2\right]+2}}\\
    G_3^{(k)}&=D_{L_{1}} E_{L_1}... D_{L_{\left[(k-1)/2\right]}} E_{L_{\left[(k-1)/2\right]}} D_{L_{\left[(k-1)/2\right]}+1},
\end{aligned}
\end{equation}
and perform the same iteration as the first two critical points of Case II. The output is a series of cross ratios denoted as
\begin{equation}
    \overline{X}^{L(j)}_{G_2 G_3}:=\overline{X}_{D_{L_{\left[ (k-1)/2 \right]+j+1}} (D_{L_{1}} E_{L_{1}} ... D_{L_{\left[ (k-1)/2 \right]+j-1}} E_{L_{\left[ (k-1)/2 \right]+j-1}} D_{L_{\left[ (k-1)/2 \right]+j}})},
\end{equation}
where $j=1,2,...,\left[\frac{n+1}{2} \right]-\left[\frac{k-1}{2} \right]$. 
Similarly we define a series of upper cells
\begin{equation}
    \begin{aligned}
    I_1^{(k)}&=B D_{U_{\left[n/2 \right]+1}}  E_{U_{\left[n/2 \right]}} D_{U_{\left[n/2 \right]}}...D_{U_{\left[(k-2)/2\right]+3}} E_{U_{\left[(k-2)/2\right]+2}}\\
    I_2^{(k)}&=E_{U_{\left[(k-2)/2\right]+1}}\\
    I_3^{(k)}&=A D_{L_{\left[(n+1)/2 \right]+1}}  E_{L_{\left[(n+1)/2 \right]}} D_{L_{\left[(n+1)/2 \right]}}...D_{L_{2}} E_{L_{1}}\\
    G_1^{(k)}&=D_{L_1}\\
    G_2^{(k)}&=D_{U_{\left[(k-2)/2\right]+2}}\\
    G_3^{(k)}&=D_{U_{1}} E_{U_1}... D_{U_{\left[(k-2)/2\right]}} E_{U_{\left[(k-2)/2\right]}} D_{U_{\left[(k-2)/2\right]+1}},
\end{aligned}
\end{equation}
which are prepared for the same iteration. 
The iteration gives us another series of cross ratios denoted as $\overline{X}^{L(j)}_{G_2 G_3}:=\overline{X}_{D_{U_{\left[ (k-2)/2 \right]+j+1}} (D_{U_{1}} E_{U_{1}} ... D_{U_{\left[ (k-2)/2 \right]+j-1}} E_{U_{\left[ (k-2)/2 \right]+j-1}} D_{U_{\left[ (k-2)/2 \right]+j}})}$, where $j=1,2,...,\left[\frac{n}{2} \right]-\left[\frac{k-2}{2} \right]$. 

We can then compute $X_{AB}$ from the above two series.
The result is  
\begin{equation}
    \overline{X}_{AB}^{n}= \prod_{i=1}^{\left[n/2 \right]-\left[(k-2)/2 \right]} \overline{X}^{U(i)}_{G_2 G_3}  \prod_{j=1}^{\left[(n+1)/2 \right]-\left[(k-1)/2 \right]} \overline{X}^{L(j)}_{G_2 G_3}.
\end{equation}
The explicit numerical results are summarized in the tables and figures in Section \ref{sec5}.
In particular, the 3-6th critical points are listed in the 3-6th columns of Table \ref{TableII}.

\bibliography{main}

@article{Ju:2024kuc,
    author = "Ju, Xin-Xiang and Pan, Wen-Bin and Sun, Ya-Wen and Wang, Yuan-Tai and Zhao, Yang",
    title = "{More on the upper bound of holographic n-partite information}",
    eprint = "2411.19207",
    archivePrefix = "arXiv",
    primaryClass = "hep-th",
    doi = "10.1007/JHEP03(2025)184",
    journal = "JHEP",
    volume = "03",
    pages = "184",
    year = "2025"
}

@article{Ju:2024hba,
    author = "Ju, Xin-Xiang and Pan, Wen-Bin and Sun, Ya-Wen and Zhao, Yang",
    title = "{Holographic multipartite entanglement from the upper bound of $n$-partite information}",
    eprint = "2411.07790",
    archivePrefix = "arXiv",
    primaryClass = "hep-th",
    month = "11",
    year = "2024"
}

@article{Ju:2025tgg,
    author = "Ju, Xin-Xiang and Sun, Ya-Wen and Zhao, Yang",
    title = "{Upper bound of holographic entanglement entropy combinations}",
    eprint = "2505.11059",
    archivePrefix = "arXiv",
    primaryClass = "hep-th",
    doi = "10.1007/JHEP09(2025)085",
    journal = "JHEP",
    volume = "09",
    pages = "085",
    year = "2025"
}

@article{Ju:2025eyn,
    author = "Ju, Xin-Xiang and Liu, Bo-Hao and Sun, Ya-Wen and Xu, Bo-Yu and Zhao, Yang",
    title = "{Holographic multipartite entanglement structures in IR modified geometries}",
    eprint = "2512.20397",
    archivePrefix = "arXiv",
    primaryClass = "hep-th",
    month = "12",
    year = "2025"
}

@article{Smolin:2005khf,
    author = "Smolin, John A. and Verstraete, Frank and Winter, Andreas",
    title = "{Entanglement of assistance and multipartite state distillation}",
    eprint = "quant-ph/0505038",
    archivePrefix = "arXiv",
    doi = "10.1103/PhysRevA.72.052317",
    journal = "Phys. Rev. A",
    volume = "72",
    pages = "052317",
    year = "2005"
}

@article{Hayden:2011ag,
    author = "Hayden, Patrick and Headrick, Matthew and Maloney, Alexander",
    title = "{Holographic Mutual Information is Monogamous}",
    eprint = "1107.2940",
    archivePrefix = "arXiv",
    primaryClass = "hep-th",
    reportNumber = "BRX-TH-638, BRX-TH-638",
    doi = "10.1103/PhysRevD.87.046003",
    journal = "Phys. Rev. D",
    volume = "87",
    number = "4",
    pages = "046003",
    year = "2013"
}

@article{Miyaji:2015yva,
    author = "Miyaji, Masamichi and Takayanagi, Tadashi",
    title = "{Surface/State Correspondence as a Generalized Holography}",
    eprint = "1503.03542",
    archivePrefix = "arXiv",
    primaryClass = "hep-th",
    reportNumber = "YITP-15-16, IPMU15-0024",
    doi = "10.1093/ptep/ptv089",
    journal = "PTEP",
    volume = "2015",
    number = "7",
    pages = "073B03",
    year = "2015"
}

@article{Horodecki:2005vvo,
    author = "Horodecki, Micha{\l} and Oppenheim, Jonathan and Winter, Andreas",
    title = "{Partial quantum information}",
    eprint = "quant-ph/0505062",
    archivePrefix = "arXiv",
    doi = "10.1038/nature03909",
    journal = "Nature",
    volume = "436",
    number = "7051",
    pages = "673--676",
    year = "2005"
}

@inproceedings{DiVincenzo:1998zz,
    author = "DiVincenzo, David P. and Fuchs, Christopher A. and Mabuchi, Hideo and Smolin, John A. and Thapliyal, Ashish and Uhlmann, Armin",
    title = "{Entanglement of assistance}",
    booktitle = "{1st NASA Conference on Quantum Computing and Quantum Communications}",
    eprint = "quant-ph/9803033",
    archivePrefix = "arXiv",
    month = "2",
    year = "1998"
}

@article{Ryu:2006bv,
    author = "Ryu, Shinsei and Takayanagi, Tadashi",
    title = "{Holographic derivation of entanglement entropy from AdS/CFT}",
    eprint = "hep-th/0603001",
    archivePrefix = "arXiv",
    reportNumber = "NSF-KITP-06-11, NSF-KITP-06-11",
    doi = "10.1103/PhysRevLett.96.181602",
    journal = "Phys. Rev. Lett.",
    volume = "96",
    pages = "181602",
    year = "2006"
}

@article{Hubeny:2007xt,
    author = "Hubeny, Veronika E. and Rangamani, Mukund and Takayanagi, Tadashi",
    title = "{A Covariant holographic entanglement entropy proposal}",
    eprint = "0705.0016",
    archivePrefix = "arXiv",
    primaryClass = "hep-th",
    reportNumber = "DCPT-07-13, KUNS-2069",
    doi = "10.1088/1126-6708/2007/07/062",
    journal = "JHEP",
    volume = "07",
    pages = "062",
    year = "2007"
}

@inproceedings{Schilling:2014sjf,
    author = "Schilling, Christian",
    title = "{The Quantum Marginal Problem}",
    eprint = "1404.1085",
    archivePrefix = "arXiv",
    primaryClass = "quant-ph",
    month = "4",
    year = "2014"
}

@mastersthesis{Schilling:2015wtq,
    author = "Schilling, Christian",
    title = "{Quantum Marginal Problem and its Physical Relevance}",
    eprint = "1507.00299",
    archivePrefix = "arXiv",
    primaryClass = "quant-ph",
    reportNumber = "DISS.{\textasciitilde}ETH No.{\textasciitilde}21748",
    doi = "10.3929/ethz-a-010139282",
    type = "Other thesis",
    month = "7",
    year = "2015"
}

@article{Umemoto:2018jpc,
    author = "Umemoto, Koji and Zhou, Yang",
    title = "{Entanglement of Purification for Multipartite States and its Holographic Dual}",
    eprint = "1805.02625",
    archivePrefix = "arXiv",
    primaryClass = "hep-th",
    reportNumber = "YITP-18-41",
    doi = "10.1007/JHEP10(2018)152",
    journal = "JHEP",
    volume = "10",
    pages = "152",
    year = "2018"
}

@article{Takayanagi:2017knl,
    author = "Takayanagi, Tadashi and Umemoto, Koji",
    title = "{Entanglement of purification through holographic duality}",
    eprint = "1708.09393",
    archivePrefix = "arXiv",
    primaryClass = "hep-th",
    reportNumber = "YITP-17-89, IPMU17-0115",
    doi = "10.1038/s41567-018-0075-2",
    journal = "Nature Phys.",
    volume = "14",
    number = "6",
    pages = "573--577",
    year = "2018"
}

@article{Harper:2020wad,
    author = "Harper, Jonathan",
    title = "{Multipartite entanglement and topology in holography}",
    eprint = "2006.02899",
    archivePrefix = "arXiv",
    primaryClass = "hep-th",
    reportNumber = "BRX-TH-6664",
    doi = "10.1007/JHEP03(2021)116",
    journal = "JHEP",
    volume = "03",
    pages = "116",
    year = "2021"
}

@article{Bao:2019zqc,
    author = "Bao, Ning and Cheng, Newton",
    title = "{Multipartite Reflected Entropy}",
    eprint = "1909.03154",
    archivePrefix = "arXiv",
    primaryClass = "hep-th",
    doi = "10.1007/JHEP10(2019)102",
    journal = "JHEP",
    volume = "10",
    pages = "102",
    year = "2019"
}

@article{Dutta:2019gen,
    author = "Dutta, Souvik and Faulkner, Thomas",
    title = "{A canonical purification for the entanglement wedge cross-section}",
    eprint = "1905.00577",
    archivePrefix = "arXiv",
    primaryClass = "hep-th",
    doi = "10.1007/JHEP03(2021)178",
    journal = "JHEP",
    volume = "03",
    pages = "178",
    year = "2021"
}

@article{Akers:2019gcv,
    author = "Akers, Chris and Rath, Pratik",
    title = "{Entanglement Wedge Cross Sections Require Tripartite Entanglement}",
    eprint = "1911.07852",
    archivePrefix = "arXiv",
    primaryClass = "hep-th",
    doi = "10.1007/JHEP04(2020)208",
    journal = "JHEP",
    volume = "04",
    pages = "208",
    year = "2020"
}

@article{Hayden:2021gno,
    author = "Hayden, Patrick and Parrikar, Onkar and Sorce, Jonathan",
    title = "{The Markov gap for geometric reflected entropy}",
    eprint = "2107.00009",
    archivePrefix = "arXiv",
    primaryClass = "hep-th",
    doi = "10.1007/JHEP10(2021)047",
    journal = "JHEP",
    volume = "10",
    pages = "047",
    year = "2021"
}

@article{Gadde:2022cqi,
    author = "Gadde, Abhijit and Krishna, Vineeth and Sharma, Trakshu",
    title = "{New multipartite entanglement measure and its holographic dual}",
    eprint = "2206.09723",
    archivePrefix = "arXiv",
    primaryClass = "hep-th",
    reportNumber = "TIFR/TH/22-34",
    doi = "10.1103/PhysRevD.106.126001",
    journal = "Phys. Rev. D",
    volume = "106",
    number = "12",
    pages = "126001",
    year = "2022"
}

@article{Penington:2022dhr,
    author = "Penington, Geoff and Walter, Michael and Witteveen, Freek",
    title = "{Fun with replicas: tripartitions in tensor networks and gravity}",
    eprint = "2211.16045",
    archivePrefix = "arXiv",
    primaryClass = "hep-th",
    doi = "10.1007/JHEP05(2023)008",
    journal = "JHEP",
    volume = "05",
    pages = "008",
    year = "2023"
}

@article{Bao:2015bfa,
    author = "Bao, Ning and Nezami, Sepehr and Ooguri, Hirosi and Stoica, Bogdan and Sully, James and Walter, Michael",
    title = "{The Holographic Entropy Cone}",
    eprint = "1505.07839",
    archivePrefix = "arXiv",
    primaryClass = "hep-th",
    reportNumber = "CALT-TH-2015-020, IPMU15-0074, SLAC-PUB-16294, SU-ITP-15-08, CALT-TH 2015-020, IPMU15-0074, SLAC-PUB-16294, SU-ITP-15/08",
    doi = "10.1007/JHEP09(2015)130",
    journal = "JHEP",
    volume = "09",
    pages = "130",
    year = "2015"
}

@article{Hernandez-Cuenca:2022pst,
    author = "Hern{\'a}ndez-Cuenca, Sergio and Hubeny, Veronika E. and Rota, Massimiliano",
    title = "{The holographic entropy cone from marginal independence}",
    eprint = "2204.00075",
    archivePrefix = "arXiv",
    primaryClass = "hep-th",
    doi = "10.1007/JHEP09(2022)190",
    journal = "JHEP",
    volume = "09",
    pages = "190",
    year = "2022"
}

@article{Balasubramanian:2013lsa,
    author = "Balasubramanian, Vijay and Chowdhury, Borun D. and Czech, Bartlomiej and de Boer, Jan and Heller, Michal P.",
    title = "{Bulk curves from boundary data in holography}",
    eprint = "1310.4204",
    archivePrefix = "arXiv",
    primaryClass = "hep-th",
    doi = "10.1103/PhysRevD.89.086004",
    journal = "Phys. Rev. D",
    volume = "89",
    number = "8",
    pages = "086004",
    year = "2014"
}

@article{Gadde:2023zzj,
    author = "Gadde, Abhijit and Krishna, Vineeth and Sharma, Trakshu",
    title = "{Towards a classification of holographic multi-partite entanglement measures}",
    eprint = "2304.06082",
    archivePrefix = "arXiv",
    primaryClass = "hep-th",
    doi = "10.1007/JHEP08(2023)202",
    journal = "JHEP",
    volume = "08",
    pages = "202",
    year = "2023"
}

@article{He:2020xuo,
    author = "He, Temple and Hubeny, Veronika E. and Rangamani, Mukund",
    title = "{Superbalance of Holographic Entropy Inequalities}",
    eprint = "2002.04558",
    archivePrefix = "arXiv",
    primaryClass = "hep-th",
    doi = "10.1007/JHEP07(2020)245",
    journal = "JHEP",
    volume = "07",
    pages = "245",
    year = "2020"
}

@article{Liu:2024ulq,
    author = "Liu, Bowei and Zhang, Junjia and Ohyama, Shuhei and Kusuki, Yuya and Ryu, Shinsei",
    title = "{Multiwavefunction overlap and multientropy for topological ground states in (2+1) dimensions}",
    eprint = "2410.08284",
    archivePrefix = "arXiv",
    primaryClass = "cond-mat.str-el",
    reportNumber = "RIKEN-iTHEMS-Report-24 KYUSHU-HET-297",
    doi = "10.1103/tjcf-yryh",
    journal = "Phys. Rev. B",
    volume = "112",
    number = "12",
    pages = "125160",
    year = "2025"
}

@article{Balasubramanian:2024ysu,
    author = "Balasubramanian, Vijay and Kang, Monica Jinwoo and Murdia, Chitraang and Ross, Simon F.",
    title = "{Signals of multiparty entanglement and holography}",
    eprint = "2411.03422",
    archivePrefix = "arXiv",
    primaryClass = "hep-th",
    doi = "10.1007/JHEP06(2025)068",
    journal = "JHEP",
    volume = "06",
    pages = "068",
    year = "2025"
}

@article{Basak:2024uwc,
    author = "Basak, Jaydeep Kumar and Malvimat, Vinay and Yoon, Junggi",
    title = "{A New Genuine Multipartite Entanglement Measure: from Qubits to Multiboundary Wormholes}",
    eprint = "2411.11961",
    archivePrefix = "arXiv",
    primaryClass = "hep-th",
    month = "11",
    year = "2024"
}

@article{Bao:2025psl,
    author = "Bao, Ning and Furuya, Keiichiro and Naskar, Joydeep",
    title = "{Tripartite Entanglement Signal from Multipartite Entanglement of Purification}",
    eprint = "2509.08209",
    archivePrefix = "arXiv",
    primaryClass = "hep-th",
    month = "9",
    year = "2025"
}

@article{Ju:2023dzo,
    author = "Ju, Xin-Xiang and Liu, Bo-Hao and Pan, Wen-Bin and Sun, Ya-Wen and Wang, Yuan-Tai",
    title = "{Squashed entanglement from generalized Rindler wedge}",
    eprint = "2310.09799",
    archivePrefix = "arXiv",
    primaryClass = "hep-th",
    doi = "10.1007/JHEP09(2025)006",
    journal = "JHEP",
    volume = "09",
    pages = "006",
    year = "2025"
}

@article{Freedman:2016zud,
    author = "Freedman, Michael and Headrick, Matthew",
    title = "{Bit threads and holographic entanglement}",
    eprint = "1604.00354",
    archivePrefix = "arXiv",
    primaryClass = "hep-th",
    reportNumber = "BRX-TH-6302, NSF-KITP-16-051",
    doi = "10.1007/s00220-016-2796-3",
    journal = "Commun. Math. Phys.",
    volume = "352",
    number = "1",
    pages = "407--438",
    year = "2017"
}

@article{Bhattacharya:2014vja,
    author = "Bhattacharya, Jyotirmoy and Hubeny, Veronika E. and Rangamani, Mukund and Takayanagi, Tadashi",
    title = "{Entanglement density and gravitational thermodynamics}",
    eprint = "1412.5472",
    archivePrefix = "arXiv",
    primaryClass = "hep-th",
    reportNumber = "DCPT-14-77, YITP-104, PMU14-0359",
    doi = "10.1103/PhysRevD.91.106009",
    journal = "Phys. Rev. D",
    volume = "91",
    number = "10",
    pages = "106009",
    year = "2015"
}

@article{Araki:1970ba,
    author = "Araki, H. and Lieb, E. H.",
    title = "{Entropy inequalities}",
    doi = "10.1007/BF01646092",
    journal = "Commun. Math. Phys.",
    volume = "18",
    pages = "160--170",
    year = "1970"
}

@article{Li:2014bep,
    author = "Li, Ke and Winter, Andreas",
    title = "{Squashed Entanglement, $\mathbf {k}$ -Extendibility, Quantum Markov Chains, and Recovery Maps}",
    eprint = "1410.4184",
    archivePrefix = "arXiv",
    primaryClass = "quant-ph",
    doi = "10.1007/s10701-018-0143-6",
    journal = "Found. Phys.",
    volume = "48",
    number = "8",
    pages = "910--924",
    year = "2018"
}

@article{Bao:2018fso,
    author = "Bao, Ning and Chatwin-Davies, Aidan and Remmen, Grant N.",
    title = "{Entanglement of Purification and Multiboundary Wormhole Geometries}",
    eprint = "1811.01983",
    archivePrefix = "arXiv",
    primaryClass = "hep-th",
    doi = "10.1007/JHEP02(2019)110",
    journal = "JHEP",
    volume = "02",
    pages = "110",
    year = "2019"
}

@article{Brown:1986nw,
    author = "Brown, J. David and Henneaux, M.",
    title = "{Central Charges in the Canonical Realization of Asymptotic Symmetries: An Example from Three-Dimensional Gravity}",
    doi = "10.1007/BF01211590",
    journal = "Commun. Math. Phys.",
    volume = "104",
    pages = "207--226",
    year = "1986"
}

@article{Chu:2019etd,
    author = "Chu, Jinwei and Qi, Runze and Zhou, Yang",
    title = "{Generalizations of Reflected Entropy and the Holographic Dual}",
    eprint = "1909.10456",
    archivePrefix = "arXiv",
    primaryClass = "hep-th",
    doi = "10.1007/JHEP03(2020)151",
    journal = "JHEP",
    volume = "03",
    pages = "151",
    year = "2020"
}

@article{Iizuka:2025bcc,
    author = "Iizuka, Norihiro and Lin, Simon and Nishida, Mitsuhiro",
    title = "{Why many-partite entanglement is essential for holography}",
    eprint = "2504.01625",
    archivePrefix = "arXiv",
    primaryClass = "hep-th",
    month = "4",
    year = "2025"
}
\bibliographystyle{JHEP}

\end{document}